\pdfoutput=1
\documentclass[fleqn, 11pt, letterpaper]{article}
\usepackage[english]{babel}
\usepackage{lmodern}
\usepackage{amssymb}
\usepackage{color}
\usepackage[lmargin=1in,rmargin=1in,bottom=1in,top=1in,twoside=False]{geometry}
\usepackage{esvect}
\usepackage{enumitem}
\usepackage{amsfonts}
\usepackage{amsmath}
\usepackage{amssymb}
\usepackage{amsthm}
\usepackage[english]{babel}
\usepackage{complexity}
\usepackage{mathrsfs}
\usepackage{graphics}
\usepackage[pagewise]{lineno}
\usepackage{mathtools}
\usepackage{microtype}
\usepackage[defaultlines=4,all]{nowidow}
\usepackage[many]{tcolorbox}
\usepackage[colorinlistoftodos,prependcaption,textsize=scriptsize]{todonotes}

\usepackage{xcolor}
\usepackage{xspace}
\usepackage{makecell}

\usepackage[hidelinks]{hyperref}
\usepackage[nameinlink]{cleveref}

\usetikzlibrary{calc,arrows, automata, fit, shapes, arrows.meta, positioning, decorations.pathmorphing, patterns}

\tcolorboxenvironment{theorem}{
    colframe=cyan, sharp corners=west, 
    colback=cyan!10,
	rightrule=0pt,
	toprule=0pt,
	bottomrule=0pt,
	top=0pt,
	right=0pt,
	bottom=0pt,
 }
\tcolorboxenvironment{corollary}{
    colframe=blue, sharp corners=west, 
    colback=blue!10,
	rightrule=0pt,
	toprule=0pt,
	bottomrule=0pt,
	top=0pt,
	right=0pt,
	bottom=0pt,
 }
\tcolorboxenvironment{lemma}{
    colframe=red, sharp corners=west,
	colback=red!10,
	rightrule=0pt,
	toprule=0pt,
	bottomrule=0pt,
	top=0pt,
	right=0pt,
	bottom=0pt,
 }

 \tcolorboxenvironment{result}{
    colframe=yellow, sharp corners=west, 
    colback=yellow!10,
	rightrule=0pt,
	toprule=0pt,
	bottomrule=0pt,
	top=0pt,
	right=0pt,
	bottom=0pt,
 }
\tcolorboxenvironment{proposition}{
    colframe=green, sharp corners=west, 
    colback=green!10,
	rightrule=0pt,
	toprule=0pt,
	bottomrule=0pt,
	top=0pt,
	right=0pt,
	bottom=0pt,
 }

\renewcommand{\epsilon}{\varepsilon}
\renewcommand{\phi}{\varphi}

\newcommand{\ie}{i.\,e.,\xspace}
\newcommand{\eg}{e.\,g.,\xspace}

\newcommand{\Cmc}{\ensuremath{\mathcal{C}}\xspace}

\newcommand{\Fmc}{\ensuremath{\mathcal{F}}\xspace}

\newcommand{\Oof}{\mathcal{O}}
\renewcommand{\phi}{\varphi}

\renewcommand{\epsilon}{\varepsilon}
\newcommand{\budget}{b}

\newcommand{\Ttt}{\ensuremath{\mathtt{T}}\xspace}
\newcommand{\Ctt}{\ensuremath{\mathtt{C}}\xspace}

\renewcommand{\NP}{\textsf{NP}\xspace}
\renewcommand{\FPT}{\textsf{FPT}\xspace}
\renewcommand{\PP}{\textsf{P}\xspace}
\renewcommand{\W}[1]{\textsf{W[#1]}\xspace}
\renewcommand{\PSPACE}{\textsf{PSPACE}\xspace}
\newcommand{\kct}{\ensuremath{\binom{\kappa}{2}}\xspace}
\newcommand{\kctwo}{\kct}

\usepackage{xspace}

\newcommand{\rmst}{\textsc{Rainbow MST}\xspace}
\newcommand{\rst}{\textsc{Rainbow ST}\xspace}
\newcommand{\rmat}{\textsc{Rainbow Matching}\xspace}
\newcommand{\mdis}{\textsc{Matching Discovery}\xspace}

\newcommand{\sdis}{\textsc{Vertex Cut Discovery}\xspace}

\newcommand{\clique}{\textsc{Clique}\xspace}
\newcommand{\pidis}{\textsc{$\Pi$-{\sc Discovery}}\xspace}
\newcommand{\rbpi}{\textsc{Rainbow}-$\Pi$\xspace}
\newcommand{\wrbpi}{\textsc{Weighted Rainbow}-$\Pi$\xspace}
\newcommand{\mcc}{\textsc{Multicolored Clique\xspace}}

\theoremstyle{remark}
\newtheorem{theorem}{Theorem}[section]
\newtheorem{corollary}{Corollary}[section]
\newtheorem{definition}{Definition}[section]
\newtheorem{lemma}{Lemma}[section]
\newtheorem{proposition}{Proposition}[section]

\newtheorem{claim}{Claim}
\crefname{corollary}{Corollary}{Corollaries}
\crefname{lemma}{Lemma}{Lemmas}
\crefname{section}{Section}{Sections}

\newtheorem{result}{}
\newtheorem*{result*}{}
\newtheorem*{remark*}{Remark}

\begin{document}

\title{Solution discovery via reconfiguration for problems in \PP}


\author{
Mario Grobler\\ University of Bremen, Germany\\ \texttt{grobler@uni-bremen.de} 
\and 
Stephanie Maaz\\ University of Waterloo, Canada \\\texttt{smaaz@uwaterloo.ca}
\and 
Nicole Megow\\ University of Bremen, Germany \\\texttt{nicole.megow@uni-bremen.de}
\and 
Amer E.~Mouawad\\
American University of Beirut, Lebanon\\
\texttt{aa368@aub.edu.lb}
\and
Vijayaragunathan Ramamoorthi\\University of Bremen, Germany\\\texttt{vira@uni-bremen.de}
\and
Daniel Schmand \\ University of Bremen, Germany \\\texttt{schmand@uni-bremen.de}
\and 
Sebastian Siebertz\\
University of Bremen, Germany\\
\texttt{siebertz@uni-bremen.de}}


\maketitle

\begin{abstract}
\noindent In the recently introduced framework of solution discovery via reconfiguration [Fellows et al., ECAI 2023], we are given an initial configuration of $k$ tokens on a graph and the question is whether we can transform this configuration into a feasible solution (for some problem) via a bounded number $b$ of small modification steps.
In this work, we study solution discovery variants of polynomial-time solvable problems, namely \textsc{Spanning Tree Discovery}, \textsc{Shortest Path Discovery}, \textsc{Matching Discovery}, and \textsc{Vertex/Edge Cut Discovery} in the unrestricted token addition/removal model, the token jumping model, and the token sliding model. 
In the unrestricted token addition/removal model, we show that all four discovery variants remain in~\PP. 
For the toking jumping model we also prove containment in \PP, except for \textsc{Vertex/Edge Cut Discovery}, for which we prove $\NP$-completeness.
Finally, in the token sliding model, almost all considered problems become $\NP$-complete, the exception being \textsc{Spanning Tree Discovery}, which remains polynomial-time solvable. We then study the parameterized complexity of the $\NP$-complete problems and provide a full classification of tractability with respect to the parameters solution size (number of tokens) $k$ and transformation budget (number of steps) $b$. 
Along the way, we observe strong connections between the solution discovery variants of our base problems and their \emph{(weighted) rainbow variants} as well as their \emph{red-blue variants with cardinality constraints}. 
\end{abstract}
\section{Introduction}\label{sec:intro}

Classically, in an optimization problem, we are given a problem instance and the task is to compute an optimal solution. However, in many applications and real-world scenarios we are already given a current non-optimal or infeasible solution for an instance. Depending on the application, it might be desirable to find an optimal or feasible solution via a bounded number of small modification steps starting from the current solution.

Very prominent examples for such ``systems'' are typically settings where humans are involved in the system and big changes to the running system are not easily implementable or even accepted. 
When optimizing public transport lines, shift plans, or when assigning workers to tasks it is clearly desirable to aim for an optimal solution that is as similar as possible to the current state of the system. 
Fellows et al.~\cite{sol-discovery} recently introduced the \emph{solution discovery via reconfiguration} framework addressing the computational aspects of such problems. 
In their model, an optimizer is given a problem instance together with a current (possibly) infeasible solution. The aim is to decide whether a feasible solution to the given problem can be constructed by applying only a bounded number of changes to the current state. 
We extend this line of work and focus on core polynomial-time solvable problems on graphs, namely \textsc{Spanning Tree}, \textsc{Shortest Path}, \textsc{Matching}, and \textsc{Vertex/Edge Cut}. 
More precisely, for any of the four aforementioned problems, say $\Pi$, we consider instances consisting of a graph $G$, a budget $\budget$, and a starting configuration of $k$ tokens, which is not necessarily a feasible solution for $\Pi$ (and where tokens either occupy vertices or edges of $G$). 
The goal is to decide whether one can transform the starting configuration into a feasible solution for $\Pi$ by applying at most $\budget$ ``local changes''.

The solution discovery framework is inspired by approaches transforming one solution to another such as \emph{local search}, \emph{reoptimization}, and \emph{combinatorial reconfiguration}. \emph{Local search} is an algorithmic paradigm that is based on iterative improvement of solutions in a previously defined neighborhood. In contrast to our setting, \emph{local search} typically improves the current solution in each step, while we allow arbitrary configurations between the starting and ending configurations (our only restriction is that each vertex/edge can be occupied by at most one token). In \emph{reoptimization} the aim is also to compute  optimal solutions starting from optimal solutions of ``neighboring'' instances (distance between instances being usually defined as the number of vertex/edge addition/deletion required to make the two graphs isomorphic). Closely related is the field of \emph{sensitivity analysis}, a very classical area studying how sensitive an optimal solution is (how it reacts) to small changes in the input. In \emph{combinatorial reconfiguration} we are also given a starting solution, but additionally a \emph{target solution}, and very often constraints on the intermediate steps, \eg that every intermediate step maintains a valid solution. In our setting the target and intermediate steps are not explicitly specified, but we aim for any final configuration satisfying some desired properties. 

As an 
illustrating example application, consider the scenario in which a city experiences severe weather conditions, leading to a rapid increase in river water levels. The city relies on the protection of a dam,  but there is a foreseeable risk that the dam may eventually fail. We assume that the dam breaks (continuously) from point $A$ to point $B$, where the (shortest) distance between $A$ and~$B$ in our graph-view of the world is exactly $k$ vertices (including $A$ and $B$). In anticipation of such emergencies, the city has also placed (at least) $k$ sandbags in fixed locations across the city so that one can move them as fast as possible to avoid greater damage (we here make the simplifying assumption that each unit of the broken dam can be fixed by a single sandbag). 
When the dam breaks we can easily compute a shortest path between $A$ and $B$, which also allows us to compute the minimum number of sandbags required to stop the flowing water. 
However, computing such a shortest path is not enough in this situation. Instead, we additionally need to account for sandbags having to move to the appropriate locations as quickly as possible. This corresponds exactly to the problem of finding a shortest path (between two fixed vertices) that is quickly reachable from a predefined set of positions in the graph. In other words, this motivates the study of discovery variants of the \textsc{Shortest Path} problem.

An alternative perspective at solution discovery problems is as follows. Consider a vertex (resp.\ edge) selection problem $\Pi$ on graphs. 
Assume that each element (vertex or edge) in the solution of size $k$ must be \emph{supported} by one of $k$ support points, which are located at $k$ different elements of the input graph, and which can each support exactly one element of the solution. 
The cost of supporting an element is measured by the distance to the chosen support point. 
The problem of deciding whether there exists a solution that can be supported with cost $\budget$ 
 corresponds to a discovery variant of  problem $\Pi$.

Before we proceed and state our results we quickly recall the solution discovery framework of Fellows et al.~\cite{sol-discovery}  (see preliminaries for formal definitions). 
Consider an instance of a vertex (resp.\ edge) selection problem $\Pi$ on graphs, where some vertices (resp.\ edges) of the input graph are occupied by (distinct) tokens. 
A token may be moved (in a specific way depending on the concrete model) for a cost of $1$. 
In the \emph{unrestricted token addition/removal} model\footnote{We call the model ``unrestricted'' to differentiate it from the addition/removal model usually considered in reconfiguration problems as the latter imposes a lower or upper bound on the number of tokens in the graph at all times.}, an existing token may be removed, or a new token may be placed on an unoccupied vertex (resp.\ edge) for a cost of 1. 
In the \emph{token jumping} model, a token may be moved from one vertex (resp.\ edge) to an arbitrary unoccupied vertex (resp.\ edge) for a cost of 1.
In the \emph{token sliding} model, a token may be moved to a neighboring unoccupied vertex (resp.\ edge) for a cost of 1.
The goal is to move the tokens such that they form a valid solution for problem $\Pi$ within the given budget.
We remark that these notions of token moves have also been studied in the realm of combinatorial reconfiguration~\cite{ito2011complexity,DBLP:journals/tcs/KaminskiMM12}. 

Fellows et al.~\cite{sol-discovery} considered the solution discovery variants of (hard) fundamental graph problems, namely \textsc{Vertex Cover}, \textsc{Independent Set}, \textsc{Dominating Set} and \textsc{Vertex Coloring}. 
The complexity of solution discovery for \textsc{Vertex Coloring} in the color flipping model was studied in~\cite{garnero2018fixing} under the name $k$-\textsc{Fix} and in the color swapping model in~\cite{de2019complexity} under the name $k$-\textsc{Swap}.
Since these problems, which we call \emph{base problems}, are \NP-complete, it is not surprising that their solution discovery variants are also \NP-complete in all of the aforementioned models. 
In this work, we continue the examination of the solution discovery via reconfiguration framework and focus on the discovery variants of polynomial-time solvable base problems, namely \textsc{Spanning Tree Discovery}, \textsc{Shortest Path Discovery}, \textsc{Matching Discovery}, and 
 \textsc{Vertex/Edge Cut Discovery}. 
When a base problem is polynomial-time solvable, one may, given an instance with a partial or infeasible solution, 
efficiently compute an optimal solution from scratch. 
However, as previously illustrated, there are situations in which a solution that is close to a currently established configuration is more desirable. 
As we show in this work, the complexity of solution discovery variants of problems in \P\ is 
far from being trivial. 

We observe strong connections between the solution discovery variants of our base problems and their \emph{weighted rainbow variants} as well as their \emph{red-blue variants with cardinality constraints}. 
An instance of a weighted rainbow vertex (resp.\ edge) selection problem consists of a weighted vertex (resp.\ edge) colored graph, and the solution of such an instance may not contain two vertices (resp.\ edges) of the same color,
while collecting a certain amount of weight.
We show that if (the parameterized version of) the weighted rainbow variant of problem~$\Pi$ admits a fixed-parameter tractable (\FPT-) algorithm, then this algorithm can be used to design a fixed-parameter tractable algorithm for the solution discovery variant of~$\Pi$ in the token sliding model.
Similarly, solving the solution discovery variant of a problem in the token jumping model boils down to solving the red-blue variant of that same problem. 
Here, an instance of a red-blue vertex (resp.\ edge) selection problem consists of a graph where every vertex (resp.\ edge) is either colored red or blue and two integer parameters $k$ and $b$. The goal is to find a solution of size $k$ that contains at most $b$ blue vertices (resp.\ edges).

\subsection{Our results}

We provide a full classification of tractability vs.\ intractability with respect to the classical as well as the parameterized complexity of the aforementioned solution discovery problems in all three token models (\Cref{tab:overview}). 
Moreover, we prove some results for rainbow problems as well as red-blue problems which we believe to be of independent interest. 

\smallskip
Our main results can be summarized as follows:
\begin{result}
\textsc{Spanning Tree Discovery}, \textsc{Shortest Path Discovery}, \textsc{Matching Discovery}, and \textsc{Vertex/Edge Cut Discovery} are polynomial-time solvable in the unrestricted token addition/removal model.
\end{result}
\begin{result}
\textsc{Spanning Tree Discovery}, \textsc{Shortest Path Discovery} and \textsc{Matching Discovery} are polynomial-time solvable, while \textsc{Vertex/Edge Cut Discovery} is $\NP$-complete in the token jumping model.
\end{result}

\begin{result}
\textsc{Spanning Tree Discovery} is polynomial-time solvable, while \textsc{Shortest Path Discovery}, \textsc{Matching Discovery}, and \textsc{Vertex/Edge Cut Discovery} are $\NP$-complete in the token sliding model.
\end{result}

We then consider the parameterized complexity of the $\NP$-complete discovery problems and establish the following connection with their rainbow variants. 

\begin{result}
    Meta theorem: For an optimization problem $\Pi$, if the \textsc{weighted Rainbow-$\Pi$} problem (parameterized by solution size $k$) admits an \FPT algorithm, 
    then \textsc{$\Pi$-Discovery} (parameterized by solution size $k$) admits an \FPT algorithm in the token sliding model. 
\end{result}

\FPT algorithms for the \textsc{Weighted Rainbow Shortest Path} problem~\cite{alon1995color}, \textsc{Weighted} \rmat problem~\cite{GuptaRSZ19}, and the {\sc Weighted Rainbow Vertex/Edge Cut}  problem (which we provide in this paper) immediately imply \FPT algorithms for the discovery variants parameterized by~$k$ (in the token sliding model). 
We demonstrate the power of our meta theorem by using it to show that \textsc{Vertex/Edge Cut Discovery} is \FPT with respect to parameter $k$ in the sliding model. For \textsc{Shortest Path Discovery} and \textsc{Matching Discovery}, we give more intuitive and direct \FPT  algorithms which also achieve better running times.
We conclude by studying the parameterized complexity of all hard problems (not covered by the meta-theorem) when parameterized by either $k$ or $b$, obtaining the following results.

\begin{result}
\textsc{Shortest Path Discovery}, \textsc{Matching Discovery}, and \textsc{Vertex/Edge Cut Discovery} are \FPT when parameterized by the solution size $k$ in the token sliding model. Furthermore, \textsc{Vertex/Edge Cut Discovery} is \FPT when parameterized by $k$ in the token jumping model.
\end{result}

\begin{result}
\textsc{Shortest Path Discovery} is \FPT, while  \textsc{Matching Discovery} and \textsc{Vertex/Edge Cut Discovery} are \W{1}-hard when parameterized by the budget $\budget$ in the token sliding model. Furthermore, \textsc{Vertex/Edge Cut Discovery} is \W{1}-hard when parameterized by the budget~$\budget$ in the token jumping model.
\end{result}

\begin{table}
    \centering
    \begin{tabular}{r||c|c|c|c}
     & \textsc{Spanning Tree} & \textsc{Shortest Path} & \textsc{Matching} & \textsc{Vertex/Edge Cut} \\\hline\hline
     \makecell{Discovery \\ Add/Rem.}  & in \PP & in \PP & in \PP & in \PP  \\\hline
     \makecell{Discovery \\ Jumping} & in \PP & in \PP & in \PP & \makecell{\NP-c., \textsf{FPT[$k$]}, \\ \W{1}-hard\textsf{[$b$]}} \\\hline
     \makecell{Discovery \\ Sliding} & in \PP & \makecell{\NP-complete, \\ \textsf{FPT[$k$]}, \textsf{FPT[$\budget$]}} & \makecell{\NP-c., \textsf{FPT[$k$]},\\ \W{1}-hard\textsf{[$b$]}} & \makecell{\NP-c., \textsf{FPT[$k$]}, \\ \W{1}-hard\textsf{[$b$]}}\\\hline
     Rainbow            & in \PP~\cite{BroersmaL97} & \NP-complete & \makecell{\NP-complete \\ on paths~\cite{DBLP:journals/tcs/LeP14}} & \makecell{\NP-complete~\cite{rainbowcuts}, \\ \NP-c. on planar} \\\hline
     Red-Blue          & in \PP & in \PP  & in \PP & \makecell{\NP-c., \textsf{FPT[$k$]}, \\ \W{1}-hard\textsf{[$b$]}}
    \end{tabular}
    \caption{Overview of our results.}
    \label{tab:overview}
\end{table}

\subsection{Related work}

The solution discovery framework is closely related to the  combinatorial reconfiguration framework, introduced by Ito et al.~\cite{ito2011complexity} and studied widely since then. 
In the recon\-figuration variant of a problem we are given an initial solution $S$ and a target solution $T$ and the question is whether $S$ can be transformed into $T$ by a sequence of reconfiguration steps (\eg token additions/removals, token jumps or token slides) such that each intermediate configuration also constitutes a solution. 

The \textsc{Minimum Spanning Tree Reconfiguration} problem by edge exchanges, i.e., token jumps, was first studied by Ito et al.~\cite{ito2011complexity}. 
They showed that the problem is in \P\ by extending the exchange property of matroids to the reconfiguration of weighted matroid bases.
\textsc{Shortest Path Reconfiguration} was introduced by Kami{\'n}ski et al.~\cite{kaminski2011shortest} and shown to be \PSPACE-complete by Bonsma~\cite{bonsma2013complexity}. 
Reconfiguration of perfect matchings was studied by Ito et al.~\cite{ito2022shortest}. 
The \textsc{Vertex Cut Reconfiguration} and \textsc{Minimum Vertex Cut Reconfiguration} problems were studied by Gomes et al.~\cite{gomes2023minimum,gomes2020some}. For further related work on reconfiguration problems we refer the reader to the surveys of van den Heuvel~\cite{van2013complexity}, Nishimura~\cite{nishimura2018introduction}, and Bousquet et al.~\cite{bousquet2022survey}.

Rainbow spanning trees have been investigated by Broersma and Li~\cite{BroersmaL97}. They characterize graphs in which there exists a rainbow spanning tree via matroid intersection.
The question to decide whether a graph contains a rainbow path has been introduced in the classical and influential work of Alon et al.~\cite{alon1995color} that introduced the color coding technique. In the related \textsc{Rainbow} $s$-$t$-\textsc{Connectivity} problem the question is to decide whether there exists a rainbow path between $s$ and~$t$, that is, a path on which no color repeats~\cite{uchizawa2013rainbow}. 
To the best of our knowledge the 
the \textsc{Rainbow Shortest Path} problem has not been studied in the literature. 
We refer the reader to~\cite{DBLP:journals/tcs/ChenLS11} for more background. 
The \rmat~problem is \NP-complete, even when restricted to properly edge-colored paths~\cite{DBLP:journals/tcs/LeP14}. 
The \textsc{Rainbow $s$-$t$-Cut} problem is known to be \NP-complete~\cite{rainbowcuts} on general graphs. We show that this problem remains \NP-complete even if we restrict it to the class of planar~graphs.

We furthermore consider the aforementioned red-blue variant of the studied problems. Our graphs are vertex (resp.\ edge) colored with colors red and blue. 
We are given two integers $k$ and~$b$ and the question is whether there exists a solution of size $k$ using at most~$b$ blue vertices (resp.\ edges). 
To the best of our knowledge these problem variants have  not been studied in the literature, however, variants where we have a cardinality constraint on both colors are related, but seem to be more difficult to solve. For example in the \textsc{Color Constrained Matching} problem~\cite{nomikos2007randomized} we are given a $2$-edge-colored graph (colors red and blue) and two parameters $k$ and~$w$ and we search for a matching of size $k$ with at most~$w$ blue and at most~$w$ red edges.
This problem is known to be at least as hard as the \textsc{Exact Matching} problem~\cite{nomikos2007randomized} (via a logarithmic-space reduction), which was introduced in~\cite{papadimitriou1982complexity}. Here, where we are given a red-blue edge-colored graph and a parameter $b$ and the question is whether there exists a perfect matching with exactly $b$ blue edges. 
The complexity of \textsc{Exact Matching} has been open for more than 40 years~\cite{DBLP:conf/mfcs/MaaloulyS22}.

\subsection{Organization of the paper} 
In the first part of the paper, we focus on the token sliding model and consider  \textsc{Spanning Tree Discovery} (Section~\ref{sec:spanning}), \textsc{Shortest Path Discovery} (Section~\ref{sec:shortest}), \textsc{Matching Discovery} (Section~\ref{sec:matching}), and \textsc{Vertex/Edge Cut Discovery} (Section~\ref{sec:separator}). 
In Section~\ref{sec:rainbow} we present our meta-theorem, linking the solution discovery variant of a problem to its rainbow variant. Section \ref{sec:other-token-models} is devoted to the alternative token models, i.e., the token jumping and unrestricted token addition/removal. We also show that the red-blue variant of an underlying base (graph) problem is always at least as hard as the discovery variant in the token jumping model. 

\section{Preliminaries}\label{sec:prelims}

We denote the set of non-negative integers by $\mathbb{N}$ and the set of non-negative reals by $\mathbb{R}_+$.
For $k \in \mathbb{N}$ we define $[k] = \{1, 2, \dots, k\}$ with the convention $[0] = \varnothing$.

\paragraph*{Graphs.}
We consider finite and loopless graphs. An undirected simple graph $G$ consists of its vertex set~$V(G)$ and edge set $E(G)$, where $E(G)$ is a subset of all two element sets of $V(G)$. Similarly, the edge set $E(G)$ of a directed simple graph is a subset of pairs of its vertices. In a multigraph we allow $E(G)$ to be a multiset. We assume our graphs to be undirected and simple if not stated otherwise.
We denote an edge connecting vertices~$u$ and~$v$ by~$uv$. Observe that $uv = vu$ for every undirected edge $uv \in E(G)$. 
A sequence $v_1,\ldots, v_q$ of pairwise distinct vertices is a path of length $q-1$ if $v_iv_{i+1}\in E(G)$ for all $1\leq i< q$. 
We write $P_q$ for the path of length~$q$. 
The distance $\mathrm{dist}_G(u,v)$ (or simply $\mathrm{dist}(u,v)$ if $G$ is clear) between two vertices $u, v \in V(G)$ is the length of a shortest path starting in $u$ and ending in $v$ in $G$.
A graph is $d$-degenerate if it can be reduced to the empty graph by iterative removal of vertices of degree at most $d$. For example, forests are $1$-degenerate. 
A graph is bipartite if its vertices can be partitioned into two parts $A,B$ such that no edge has both its endpoints in the same part. 
Equivalently, a graph is bipartite if it does not contain cycles of odd length.  
For a vertex subset $S \subseteq V(G)$, we denote by $G[S]$ the subgraph of~$G$ \emph{induced by} $S$, i.e., the graph with vertex set~$S$ and edge set $\{uv \in E(G) \mid u, v \in S\}$. Likewise, for an edge subset $M \subseteq E(G)$, we denote by $G[M]$ the graph with edge set $M$ and vertex set $\{u, v \mid uv \in M\}$.

An edge coloring $\phi : E(G) \to \Cmc$ is a function mapping each edge $e \in E(G)$ to a color \mbox{$\phi(e) \in \Cmc$}. Similarly, a vertex coloring assigns colors to vertices. An edge weight function is a function \mbox{$w: E(G)\rightarrow \mathbb{R}_+$}, and similarly a vertex weight function assigns weights to vertices. 
We denote colored weighted graphs by tuples $(G,w,\phi)$. The weight of a set of vertices/edges is the sum of the weights of its elements. 

\paragraph*{Solution discovery.}
Let $G$ be a graph. A \emph{configuration} of $G$ is either a subset of its vertices or a subset of its edges. We formalize the notions of token moves.
In the \emph{unrestricted token addition/removal} model\footnote{{Recall that this definition differs from the definition in~\cite{sol-discovery}.}}, a configuration $C'$ can be obtained (in one step) from $C$, written $C\vdash C'$, if $C' = C \cup \{x\}$ for an element $x \notin C$, or if $C' = C \setminus \{x\}$ for an element $x \in C$.
In the \emph{token jumping} model, a configuration $C'$ can be obtained (in one step) from $C$ if $C' = (C \setminus \{y\}) \cup \{x\}$ for elements $y \in C$ and $x \notin C$.
In the \emph{token sliding} model, a configuration $C'$ can be obtained (in one step) from $C$ if $C' = (C \setminus \{y\}) \cup \{x\}$ for elements $y \in C$ and $x \notin C$ if $x$ and $y$ are neighbors in $G$, that is, if $x,y \in V(G)$, then $xy \in E(G)$; and if $x,y \in E(G)$, then $x \cap y \neq \varnothing$.
If $C'$ can be obtained from $C$ (in any model), we write $C \vdash C'$.
A \emph{discovery sequence} of length $\ell$ in $G$ is a sequence of configurations $C_0 C_1 \dots C_\ell$ of $G$ such that $C_i \vdash C_{i+1}$ for all $0 \leq i < \ell$.

Let $\Pi$ be a vertex (resp.\ edge) selection problem, i.e., a problem defined on graphs such that a solution consists of a subset of vertices (resp.\ edges) satisfying certain requirements. The \textsc{$\Pi$-Discovery} problem is defined as follows. We are
given a graph $G$, a subset $S\subseteq V(G)$ (resp.\ $S\subseteq E(G)$) of size $k$ (which at this point is not necessarily a solution
for $\Pi$), and a budget $\budget$ (as a non-negative integer). 
The goal is to decide whether there exists a discovery sequence $C_0 C_1 \dots C_\ell$ in $G$ for some $\ell \leq b$ such that $S = C_0$ and~$C_\ell$ is a solution for $\Pi$.

Note that for discovery problems in the token sliding model we can always assume that $\budget \leq kn$, where $n$ is the number of vertices in the input graph.
This follows from the fact that each token will have to traverse a path of length at most~$n$ to reach its target position.
For discovery problems in the token jumping model we can always assume $\budget \leq k$, as it is sufficient to move every token at most once.
Similarly, for the unrestricted token addition/removal model we can always assume that $\budget \leq n$ for vertex selection problems and $\budget \leq m$ for edge selection problems, where $m$ is the number of edges in the input graph.
As $k$ is trivially upper-bounded by $n$ for vertex selection problems (resp.\ $m$ for edge selection problems), all solution discovery variants we consider are in~$\NP$ and proving \NP-hardness suffices to prove \NP-completeness.

\paragraph*{Parameterized complexity.}
A \emph{parameterized problem} is a language $L\subseteq \Sigma^*\times \mathbb{N}$, where $\Sigma$ is a fixed finite alphabet. For an instance $(x,\kappa)\in \Sigma^*\times \mathbb{N}$, $\kappa$ is called the \emph{parameter}.
The problem $L$ is called \emph{fixed-parameter tractable}, \textsf{FPT} for short, if there exists an algorithm that on input $(x,\kappa)$ decides in time $f(\kappa) \cdot |(x,\kappa)|^c$ whether $(x,\kappa)\in L$, for a computable function $f$ and constant~$c$. 

The \emph{\textsf{W}-hierarchy} is a collection of parameterized complexity classes $\FPT \subseteq \W{1} \subseteq \W{2}\subseteq \ldots$. It is standard to assume that the inclusion $\FPT\subseteq \W{1}$ is strict.
Therefore, showing intractability in the parameterized setting is usually accomplished by establishing an \textsf{FPT}-reduction from a \textsf{W}-hard problem.  

Let $L,L'\subseteq \Sigma^*\times\mathbb{N}$ be parameterized problems. A \emph{parameterized reduction} from $L$ to $L'$ is an algorithm that, given an instance $(x, \kappa)$ of $L$, outputs an instance $(x',\kappa')$ of $L'$ such that 
$(x, \kappa)\in L \Leftrightarrow (x',\kappa')\in L'$, $\kappa'\leq g(\kappa)$ for some computable function $g$, and the running time of the algorithm is bounded by $f(\kappa) \cdot |(x,\kappa)|^c$ for some computable function $f$ and constant~$c$.
We refer to the textbooks~\cite{cygan2015parameterized,DBLP:series/mcs/DowneyF99,FlumGrohe2006} for extensive background on parameterized complexity.

\section{Spanning trees}\label{sec:spanning}
A \emph{spanning tree} in a connected graph $G$ is a subset of edges $E' \subseteq V(E)$, where $|E'| = |V(G)|-1$ and $G[E']$ is a tree containing all vertices of $G$. In the \textsc{Spanning Tree} problem we are given a graph $G$ and the goal is to compute a spanning tree in $G$.

\subsection{Rainbow minimum spanning trees}

We reduce the problem of discovering spanning trees in the sliding model to the problem of finding (weighted) rainbow spanning trees.
A spanning tree $T \subseteq E(G)$ in a weighted edge-colored multigraph $(G, w, \varphi)$ is a {\em rainbow spanning tree} if every edge in $T$ has a distinct color, i.e., $\forall e, e' \in T$ we have $\phi(e) = \phi(e')$ if and only if $e = e'$. 
In the {\sc Rainbow Minimum Spanning Tree (\rmst)} problem, we are given $(G, w, \varphi)$ and
the goal is to compute a rainbow spanning tree of minimum total weight in $G$, or report that no such rainbow spanning tree exists. The {\sc Rainbow Spanning Tree (\rst)} problem can be defined similarly, \ie by dropping the weights or assuming that all weights are uniform. 

Rainbow spanning trees and their existence have been discussed by Broersma and Li~\cite{BroersmaL97}. In our reduction we will construct an instance of {\sc Rainbow Minimum Spanning Tree} that trivially guarantees the existence of at least one rainbow spanning tree. The following theorem shows that we can find a \rmst\ efficiently, even in multigraphs. 
This follows from similar arguments as in~\cite{BroersmaL97}. For the sake of completeness, we give a short proof. 

\begin{theorem}\label{thm:rmst-in-p}
 The {\sc Rainbow Minimum Spanning Tree} problem in multigraphs can be solved in polynomial time.
\end{theorem}

\begin{proof}
The \rmst problem can be solved as a \textsc{Weighted Matroid Intersection} problem in which one asks for a maximum-weight common independent set of two given matroids on the same ground set.\footnote{A {\em matroid} is a non-empty, downward-closed set system $(E, I)$ with ground set~$E$ and a family of subsets $I \subseteq 2^{E}$ that satisfies the augmentation property: if $A,B \in I$  and $|A|<|B|$, then $A+j \in  I$ for some $j\in B\setminus{A}$. Given a matroid $\mathcal{M}=(E,I)$, a set $A\subseteq E$ is called \emph{independent} if $A\in I$. 
    An independent set in the intersection of two matroids $(E,I_1)$ and $(E,I_2)$ on the same ground set $E$ is a set $A\subseteq E$ that is independent in both matroids, i.e., $A\in I_1\cap I_2$. We skip further details on matroids here and refer the reader to~\cite{schrijver2003combinatorial}.}
    Given an edge-colored weighted multigraph $(G,w,\phi)$, consider the following two matroids, $\mathcal{M}_1$ and $\mathcal{M}_2$, on the ground set $E = E(G)$. Matroid $\mathcal{M}_1=(E,I_1)$ is the graphic matroid, where the family of independent sets $I_1$ consists of the subsets of edges which are forests of $G$, that is,  $I_1 = \{F \subseteq E \mid G[F] \text{ is cycle-free}\}$. 
    Matroid $\mathcal{M}_2=(E,I_2)$ is the partition matroid defined by the edge coloring; consider the partition of $E$ into $E_1 \cup E_2 \cup \ldots \cup E_k$ where each~$E_i$ represents edges of the same color $i\in \Cmc$. Then the independent sets of $\mathcal{M}_2$ are defined as $I_2=\{F \subseteq E \mid |F\cap E_i| \leq 1, \forall i\in \Cmc\}$.
    
    A rainbow spanning tree is a subset of elements that is a maximum independent set in both, the graphic and the partition matroid, i.e., $\mathcal{M}_1$ and  $\mathcal{M}_2$. The task of finding such an independent set of maximum weight is known as the \textsc{Weighted Matroid Intersection} problem and several polynomial-time algorithms are known for the problem; e.g., \cite{Edmonds1970,Edmonds1979,Frank81}. To solve the \rmst problem, that asks for a {\em minimum weight} solution, we first multiply all weights by $-1$, add a large value $W$ such that all weights are positive, and then run a \textsc{Weighted Matroid Intersection} algorithm (adding $W$ to make all weights positive guarantees that the solution indeed forms a spanning tree rather than a forest). 
\end{proof}

\subsection{Spanning tree discovery}

In the {\sc Spanning Tree Discovery (STD)} problem in the token sliding model, we are given a graph $G$, an edge subset $S \subseteq E(G)$ with $|S|=n-1$ as a starting configuration, and a non-negative integer $\budget$. The goal is to decide whether there is a spanning tree of $G$ that can be discovered (starting from $S$) using at most~$\budget$ token slides.

\begin{theorem}\label{thm:std-in-p}
 The {\sc Spanning Tree Discovery} problem in the token sliding model can be solved in polynomial time. 
\end{theorem}

\begin{proof}
We prove the theorem by reducing \textsc{Spanning Tree Discovery} to the {\sc Rainbow Minimum Spanning Tree} problem in multigraphs, which can be solved in polynomial time using results from matroid theory (\Cref{thm:rmst-in-p}). 

To reduce the \textsc{Spanning Tree Discovery} problem to the \rmst problem, consider an instance of \textsc{STD} consisting of a graph $G$ and a set $S \subseteq E(G)$ of size $n-1$. 
For the construction we choose some color set $\Cmc$ of size $n-1$ and associate each edge $e\in S$ with a unique color from~$\Cmc$. Let $i_e \in \Cmc$ denote the color associated with edge $e \in S$. 
Note that by the bijection each color is also associated with a unique edge in $S$.

We construct an edge-colored weighted multigraph $(H, w, \phi)$ with $V(H) = V(G)$ and $E(H) = E(G) \times \Cmc$ and $\Cmc=[n-1]$ as follows. 
For each edge $e \in E(G)$ and color $i_{e'} \in \Cmc$, a selection of the edge $(e,i_{e'})$ in~$H$ denotes that in the \textsc{STD} problem a token slides from an edge $e' \in S$ with color $i_{e'}$ to the edge~$e$. 
We define $w$ such that the weight of an edge $(e,i_{e'})$ in $H$ denotes the sliding distance from $e'$ to $e$ w.r.t.\ the number of edges. 
That means, for each original edge $e\in E(G)$ and each color $i_{e'}$ in $\Cmc$ we have one of the copies in $H$ of this color, $(e,i_{e'})$, and its weight is the shortest path from $e'$ to $e$.
Consequently, for an edge $e\in S$ of color $i_{e}$ in the input graph $G$, we set $w(e,i_e) = 0$.
Now we have constructed an edge-colored weighted multigraph $H$. Note that a spanning tree in $H$ corresponds to a spanning tree in $G$. Additionally, a rainbow spanning tree of minimum total cost in $H$ corresponds to a spanning tree in $G$ together with a token sliding sequence of minimum total cost (w.r.t.~number of edges) and thus a solution to the {\sc STD} problem in $G$.
With this reduction, \Cref{thm:std-in-p} now directly follows from \Cref{thm:rmst-in-p}.
\end{proof}

In the appendix we show that the weighted version of the discovery problem, i.e., where we seek a spanning tree of minimum weight, can also be solved efficiently (\Cref{thm:wstd-in-p}). 
\section{Shortest paths}\label{sec:shortest}

A \emph{shortest path} in a graph between two of its vertices, say $s$ and $t$, is a path connecting $s$ and $t$ of minimum length. In the \textsc{Shortest Path} problem we are given a graph $G$ and $s,t\in V(G)$ and the goal is to compute a shortest path between $s$ and $t$ in $G$. We first prove \NP-hardness of the rainbow variant of the problem even when restricted to $2$-degenerate bipartite graphs. To do so, we describe a reduction from the \NP-complete~\cite{garey1976planar} \textsc{Hamiltonian Path} problem to the \textsc{Rainbow Shortest Path} problem. With a few minor modifications, the same reduction is then adapted to establish \NP-hardness of the  \textsc{Shortest Path Discovery} problem restricted to $2$-degenerate bipartite graphs.

\subsection{Rainbow shortest paths}
In the \textsc{Rainbow Shortest Path (Rainbow SP)} problem we are given a vertex-colored graph $(G, \varphi)$ and two vertices $s,t\in V(G)$. 

A shortest path $P$ from $s$ to $t$ in $G$ is a {\em rainbow shortest path} if every vertex in $P$ has a distinct color, \ie $\forall v, v' \in V(P)$, we have $\phi(v) = \phi(v')$ if and only if $v = v'$. The \textsc{Rainbow Shortest Path} problem asks for a rainbow shortest path $P$ from $s$ to $t$ in~$G$ (one can define the edge-colored variant of the problem in a similar way). 

Rainbow paths have been studied in the literature before, specifically in relation to the rainbow vertex-connection or edge-connection number of a graph. 
We refer the reader to~\cite{DBLP:journals/tcs/ChenLS11} for more details. 

\begin{theorem}
\label{thm:rainbowpath-np}
The \textsc{Rainbow Shortest Path} problem is \NP-complete on the class of \mbox{$2$-degenerate} bipartite graphs. 
\end{theorem}

\begin{proof}
The fact that the problem is in \NP is immediate. We show \NP-hardness by a reduction from the \textsc{Hamiltonian Path} problem, which is known to be \NP-complete~\cite{garey1976planar}.    
Given an instance~$G$ of \textsc{Hamiltonian Path}, where $G$ is a graph with $n$ vertices denoted by $V(G) = \{v_1, \dots, v_n\}$, we construct an instance $(H,\phi,s,t)$ of \textsc{Rainbow Shortest Path} as follows. See \Cref{fig:path:hardness:rainbow} for an illustration.

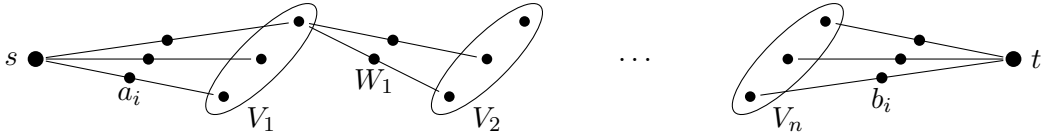
\begin{figure}[h]
    \centering
    \begin{tikzpicture}

    \fill circle (3pt) node (s) {} node[left=1mm] {$s$};

    \fill (s) ++(3.5,0.5) circle (2pt) node (v11) {};
    \fill (s) ++(3,0) circle (2pt) node (v12) {};
    \fill (s) ++(2.5,-.5) circle (2pt) node (v13) {};
    \draw[rotate around={-45:(v12)}] (v12) ellipse (0.3cm and 1cm) node[below=5mm] {$V_1$};

    \draw (s) --node[pos=1/2, fill, circle, inner sep = 1.5pt, outer sep = 1.5pt] {} (v11);
    \draw (s) --node[pos=1/2, fill, circle, inner sep = 1.5pt, outer sep = 1.5pt] {} (v12);
    \draw (s) --node[pos=1/2, fill, circle, inner sep = 1.5pt, outer sep = 1.5pt] {} node[below] {$a_i$} (v13);

    \fill (v11) ++(3,0) circle (2pt) node (v21) {};
    \fill (v12) ++(3,0) circle (2pt) node (v22) {};
    \fill (v13) ++(3,0) circle (2pt) node (v23) {};
    \draw[rotate around={-45:(v22)}] (v22) ellipse (0.3cm and 1cm) node[below=5mm] {$V_2$};

    \draw (v11) --node[pos=1/2, fill, circle, inner sep = 1.5pt, outer sep = 1.5pt] {} (v22);
    \draw (v11) --node[pos=1/2, fill, circle, inner sep = 1.5pt, outer sep = 1.5pt] {} node[below] {$W_1$} (v23);

    \node at (8,0) {$\dots$};
    \fill (v21) ++(4,0) circle (2pt) node (vn1) {};
    \fill (v22) ++(4,0) circle (2pt) node (vn2) {};
    \fill (v23) ++(4,0) circle (2pt) node (vn3) {};
    \draw[rotate around={-45:(vn2)}] (vn2) ellipse (0.3cm and 1cm) node[below=5mm] {$V_n$};

    \fill (vn2) ++(3,0) circle (3pt) node (t) {} node[right=1mm] {$t$};
    \draw (t) --node[pos=1/2, fill, circle, inner sep = 1.5pt, outer sep = 1.5pt] {} (vn1);
    \draw (t) --node[pos=1/2, fill, circle, inner sep = 1.5pt, outer sep = 1.5pt] {} (vn2);
    \draw (t) --node[pos=1/2, fill, circle, inner sep = 1.5pt, outer sep = 1.5pt] {} node[below] {$b_i$} (vn3);

    \end{tikzpicture}
    \caption{An illustration of the hardness reduction for the \textsc{Rainbow Shortest Path}~problem.}
    \label{fig:path:hardness:rainbow}
\end{figure} 

First, $H$ consists of two vertices $s$ and $t$. For every pair $i,j\leq n$ we add a new vertex $v_{i,j}$ in $H$. We define $V_i = \{v_{i,j} \mid j \leq n\}$. Then, for every $i < n$ and $e \in E(G)$ we add a new vertex $w_{i, e}$ and define $W_i = \{w_{i,e} \mid e \in E(G)\}$. Finally, for every $i \leq n$ we add new vertices $a_i$ and $b_i$. We define $A = \{a_i \mid i \leq n\}$ and $B = \{b_i \mid i \leq n\}$. 
We connect the vertices as follows. For every $i \leq n$, we insert the edges $\{s, a_i\}, \{a_i, v_{1, i}\}, \{v_{n,i}, b_i\}$, and $\{b_i, t\}$. Finally, for every $i < n$ and $e = \{v_j, v_\ell\}$ we insert the edges $\{v_{i,j}, w_{i,e}\}$ and $\{w_{i,e}, v_{i+1,\ell}\}$. We assign all vertices in $A$ the same color, all vertices in $B$ the same new color, all vertices $\{v_{i,j} \mid i \leq n\}$ the same new color (this set represents all copies of a vertex $v \in V(G)$), all vertices in $W_i$ the same new color, all vertices in $B$ the same new color, and $s$ and $t$ receive two fresh new colors.  
This finishes the construction of the instance $(H,\phi,s,t)$. 

Observe that $H$ is indeed bipartite and $2$-degenerate; all vertices of the form $w_{i,j}$ are of degree two. Their removal results in a forest, which is $1$-degenerate. Bipartiteness follows from the fact that there are no edges in $G[S]$, $G[T]$, $G[V_i]$, nor $G[W_i]$. 

We show that $G$ contains a Hamiltonian path if and only if $(H,\phi,s,t)$ is a yes-instance of~\textsc{Rainbow Shortest Path}. 
Observe that every shortest $s$-$t$-path in $H$ consists of $s$, exactly one of the vertices $a_i \in A$, exactly one vertex from every $V_i$ and one from every $W_i$, exactly one of the vertices $b_i \in B$ and $t$. Hence, there is no path of length less than $2n+2$ and every such path of length $2n+2$ is exactly of this form.

Now assume that $G$ contains a Hamiltonian path $v_{i_1} v_{i_2} \dots v_{i_n}$. We pick the following vertices in $H$. 
We pick $a_{i_1}$ and $b_{i_n}$, the vertices $v_{j,i_j}$, and, for every $j < n$, we pick the vertices $w_{j,e}$ where $e = \{v_{i_j}, v_{i_{j+1}}\}$. It is not hard to see that these vertices indeed form a rainbow shortest path from~$s$ to $t$ in $H$.  

Conversely, assume that $(H,\phi,s,t)$ is a yes-instance of \textsc{Rainbow Shortest Path}. By the above observations, every shortest path must be of the same form (picking vertices from the same sets).  
As every rainbow shortest path cannot contain two copies of the same vertex from $V(G)$ the claim follows, finishing the proof.
\end{proof}

\subsection{Shortest path discovery}

In the \textsc{Shortest Path Discovery (SPD)} problem in the token sliding model we are given a graph $G$, two vertices $s,t\in V(G)$, a starting configuration $S\subseteq V(G)$ with $|S|=k$, which is equal to the number of vertices on a shortest path between $s$ and $t$ (including $s$ and $t$), and a non-negative integer~$\budget$. 
The goal is to decide whether we can discover a shortest path between $s$ and $t$ (starting from $S$) using at most $\budget$ token slides. 
We denote an instance of \textsc{Shortest Path Discovery} by a tuple $(G, s, t,S, b)$. 

\begin{theorem}
\label{thm:path-np}
The \textsc{Shortest Path Discovery} problem in the token sliding model is \NP-complete on the class of $2$-degenerate bipartite graphs. 
\end{theorem}

\begin{proof} 
We again show \NP-hardness by a reduction from the \textsc{Hamiltonian Path} problem, which is known to be \NP-complete~\cite{garey1976planar}. The reduction is very similar to the reduction in~\Cref{thm:rainbowpath-np}. 
Given an instance $G$ of \textsc{Hamiltonian Path}, where $G$ is a graph with $n$ vertices denoted by $V(G) = \{v_1, \dots, v_n\}$, we construct an instance $(H,s,t,S,\budget)$ of \textsc{SPD} as follows. See \Cref{fig:path:hardness} for an illustration.

\begin{figure}[h]
    \centering
    \begin{tikzpicture}
    \tikzstyle{r} = [rectangle, fill, minimum size = 4pt, inner sep = 2pt, outer sep = 1pt]

    \node [r, minimum size = 5pt] (s) {} node[left = 0mm of s] {$s$};
    \node[r] (ss) at (0,-1) {} node[left = 0mm of ss] {$s'$};
    \draw (s) -- (ss);

    \fill (s) ++(3.5,0.5) circle (2pt) node (v11) {};
    \fill (s) ++(3,0) circle (2pt) node (v12) {};
    \fill (s) ++(2.5,-.5) circle (2pt) node (v13) {};
    \draw[rotate around={-45:(v12)}] (v12) ellipse (0.3cm and 1cm) node[below=5mm] {$V_1$};

    \fill (v11) ++(-1,2) circle (2pt) node (u11) {};
    \fill (v12) ++(-1,2) circle (2pt) node (u12) {};
    \fill (v13) ++(-1,2) circle (2pt) node (u13) {};
    \draw[rotate around={-45:(u12)}] (u12) ellipse (0.3cm and 1cm) node[left=5mm] {$U_1$};

    \node[r] (x1) at ($(u12) + (0,1)$) {} node[above left = -1mm and -1mm of x1] {$x_1$};
    \draw (x1) -- (u11);
    \draw (x1) -- (u12);
    \draw (x1) -- (u13);

    \draw (v11) -- (u11);
    \draw (v12) -- (u12);
    \draw (v13) -- (u13);

    \draw (s) --node[pos=1/2, fill, circle, inner sep = 1.5pt, outer sep = 1.5pt] {} (v11);
    \draw (s) --node[pos=1/2, fill, circle, inner sep = 1.5pt, outer sep = 1.5pt] {} (v12);
    \draw (s) --node[pos=1/2, fill, circle, inner sep = 1.5pt, outer sep = 1.5pt] {} node[below] {$s_i$} (v13);

    \fill (v11) ++(3,0) circle (2pt) node (v21) {};
    \fill (v12) ++(3,0) circle (2pt) node (v22) {};
    \fill (v13) ++(3,0) circle (2pt) node (v23) {};
    \draw[rotate around={-45:(v22)}] (v22) ellipse (0.3cm and 1cm) node[below=5mm] {$V_2$};

    \fill (v21) ++(-1,2) circle (2pt) node (u21) {};
    \fill (v22) ++(-1,2) circle (2pt) node (u22) {};
    \fill (v23) ++(-1,2) circle (2pt) node (u23) {};
    \draw[rotate around={-45:(u22)}] (u22) ellipse (0.3cm and 1cm) node[left=5mm] {$U_2$};

    \node[r] (x2) at ($(u22) + (0,1)$) {} node[above left = -1mm and -1mm of x2] {$x_2$};
    \draw (x2) -- (u21);
    \draw (x2) -- (u22);
    \draw (x2) -- (u23);

    \draw (v21) -- (u21);
    \draw (v22) -- (u22);
    \draw (v23) -- (u23);

    \draw (v11) --node[pos=1/2, fill, circle, inner sep = 1.5pt, outer sep = 1.5pt] {} (v22);
    \draw (v11) --node[pos=1/2, fill, circle, inner sep = 1.5pt, outer sep = 1.5pt] {} node[below] {$W_1$} (v23);

    \node at (8,0) {$\dots$};
    \fill (v21) ++(4,0) circle (2pt) node (vn1) {};
    \fill (v22) ++(4,0) circle (2pt) node (vn2) {};
    \fill (v23) ++(4,0) circle (2pt) node (vn3) {};
    \draw[rotate around={-45:(vn2)}] (vn2) ellipse (0.3cm and 1cm) node[below=5mm] {$V_n$};

    \fill (vn1) ++(-1,2) circle (2pt) node (un1) {};
    \fill (vn2) ++(-1,2) circle (2pt) node (un2) {};
    \fill (vn3) ++(-1,2) circle (2pt) node (un3) {};
    \draw[rotate around={-45:(un2)}] (un2) ellipse (0.3cm and 1cm) node[left=5mm] {$U_n$};

    \draw (vn1) -- (un1);
    \draw (vn2) -- (un2);
    \draw (vn3) -- (un3);

    \node[r] (xn) at ($(un2) + (0,1)$) {} node[above left = -1mm and -1mm of xn] {$x_n$};
    \draw (xn) -- (un1);
    \draw (xn) -- (un2);
    \draw (xn) -- (un3);

    \node[r, minimum size=5pt] (t) at ($(vn2) + (3,0)$) {} node[right = 0mm of t] {$t$};
    \draw (t) --node[pos=1/2, fill, circle, inner sep = 1.5pt, outer sep = 1.5pt] {} (vn1);
    \draw (t) --node[pos=1/2, fill, circle, inner sep = 1.5pt, outer sep = 1.5pt] {} (vn2);
    \draw (t) --node[pos=1/2, fill, circle, inner sep = 1.5pt, outer sep = 1.5pt] {} node[below] {$t_i$} (vn3);

    \node[r] (z1) at ($(v13) + (1,-1.75)$) {} node[below = 0mm of z1] {$z_1$};
    \node[r] (zn) at ($(vn3) + (1,-1.75)$) {} node[below = 0mm of zn] {$z_n$};
    \node at (6, -2.25) {$\dots$};

    \draw[dashdotted] (z1) -- (v13);
    \draw[dashdotted] (z1) -- (v23);
    \draw[dashdotted] (z1) -- (vn3);

    \path[dashdotted]
    (zn) edge[bend left] (v11)
    (zn) edge[bend left=4] (v21)
    (zn) edge[bend right=40] (vn1)
    ;

    \end{tikzpicture}
    \caption{An illustration of the hardness reduction for the \textsc{Shortest Path Discovery}~problem.}
    \label{fig:path:hardness}
\end{figure}
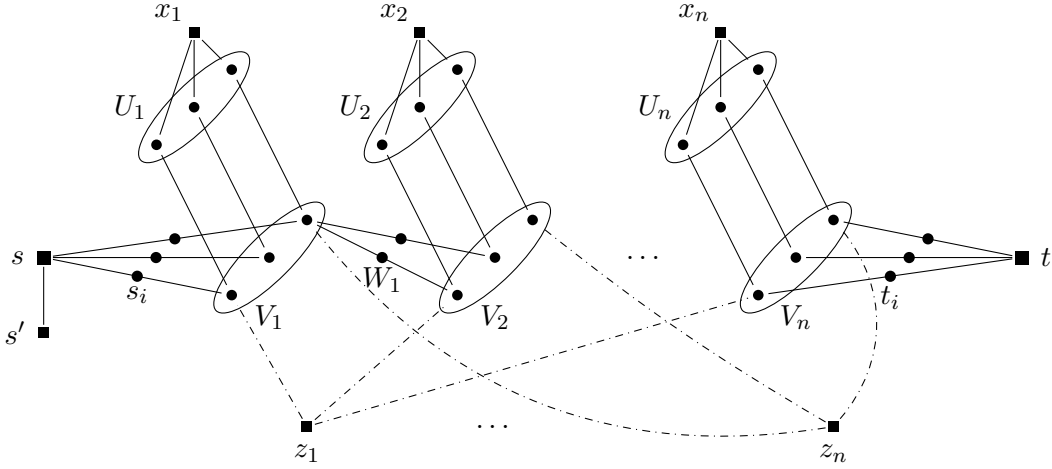 

First, $H$ consists of three vertices $s$, $s'$, and $t$. For every $i,j\leq n$ we insert two new vertices $u_{i,j}$ and $v_{i,j}$ in $H$. We define $U_i = \{u_{i,j} \mid j \leq n\}$ and $V_i = \{v_{i,j} \mid j \leq n\}$. Furthermore, for every $i \leq n$ we add a fresh vertex $x_i$. Then, for every $i < n$ and $e \in E(G)$ we insert a fresh vertex $w_{i, e}$ and define $W_i = \{w_{i,e} \mid e \in E(G)\}$. Finally, for every $i \leq n$ we insert fresh vertices $s_i$, $t_i$, and $z_i$.
We connect the vertices as follows. First, we insert the edge $s's$. For every $i \leq n$, we insert the edges $s_is, s_iv_{1, i}, v_{n,i}t_i$, and $t_it$. Furthermore, for every $i, j \leq n$ we add the edges $u_{i,j} v_{i,j}, x_i u_{i,j}$, and we connect $z_j$ to all $v_{i,j}$ via paths of length $2n+2$. Finally, for every $i < n$ and $e = v_j v_\ell$ we insert the edges $v_{i,j} w_{i,e}$ and $w_{i,e} v_{i+1,\ell}$.
This finishes the construction of $H$. We define $S = \{s,s',t\} \cup \{z_j \mid j \leq n\} \cup \{x_i \mid i \leq n\}$ and define the budget as $b = n(2n+2) + 3n + 2$.
Observe that $H$ is indeed 2-degenerate; all vertices of the form $u_{i,j}$ are of degree 2, as well as the internal vertices of the paths connecting the $x_j$ with the $v_{i,j}$. Their removal results in a forest, which is 1-degenerate. Bipartiteness follows from the fact that for every $i, j \leq n$ the distance between any two vertices in $U_i$ and $U_j$ (and resp.\ $V_i$ and $V_j$ as well as $W_i$ and $W_j$) is even. Hence, $H$ does not contain a cycle of odd length.
 
We show that $G$ contains a Hamiltonian path if and only if $(H,s,t,S,\budget)$ is a yes-instance of~\textsc{SPD}. 
Towards this goal observe that $G$ contains a Hamiltonian path if and only if $H$ contains an $s$-$t$-path that uses exactly one vertex from each $V_i$. To be precise, every shortest $s$-$t$-path consists of $s$, exactly one of the vertices $s_i$, exactly one vertex from every $V_i$ and $W_i$, exactly one of the vertices~$t_i$ and $t$. Hence, there is no such path of length smaller than $2n+2$ and every such path of length $2n+2$ is exactly of this form (observe that $|S| = 2n+3$, implying that every path that can be discovered from $S$ has length $2n+2$).

Now assume that $G$ contains a Hamiltonian path $v_{i_1} v_{i_2} \dots v_{i_n}$. We move the tokens as follows: we move the token on every $z_j$ to $v_{j,i_j}$ using the paths of length $2n + 2$ for a cost of $n(2n+2)$. 
Furthermore, for every $j < n$ we move the token on every $x_j$ via $u_{j,i_j}$ and $v_{j,i_j}$ to the vertex $w_{j,e}$ where $e = v_{i_j} v_{i_{j+1}}$ and similarly we move the token on $x_n$ via $u_{n, i_n}$ and $v_{n,i_n}$ to $t_{i_n}$ for a cost of~$3n$. 
Finally, we move the token on $s$ to $s_{i_1}$ and the token on $s'$ to $s$ for a cost of 2. 
Hence, with a budget of $b = n(2n+2) + 3n + 2$ we discover the path $s s_{i_1} v_{1, i_1} w_{1, e_1} v_{2, i_2} w_{2,i_2} \dots v_{n, i_n} t_{i_n} t$ with $e_i = v_{i_j} v_{i_{j+1}}$ of length $2n+2$.

Conversely, assume that $(H,s,t,S,\budget)$ is a yes-instance of \textsc{SPD}. By the above observation, every shortest path that can be discovered from $S$ contains exactly one vertex from every $U_i$. 
As every shortest $s$-$t$-path must also contain a vertex from $W_i$ for every $i \leq n$, and by the choice of the budget, the discovered path is of the form $s s_{i_1} v_{1, i_1} w_{1, e_1} v_{2, i_2} w_{2,i_2} \dots v_{n, i_n} t_{i_n} t$, and hence $v_{i_1} v_{i_2} \dots v_{i_n}$ is a Hamiltonian path in $G$, finishing the proof.
\end{proof}

\begin{remark*}
Recall that \textsc{Hamiltonian Path} remains \NP-complete on cubic planar graphs. As the paths leading from $s_x$ to $H_i$ have length $2n+2$ and since we may assume $G$ to be cubic, the class of all graphs arising in the construction is not only $2$-degenerate but also has bounded expansion.
\end{remark*}

It remains an interesting open problem to determine whether \textsc{Shortest Path Discovery} is polynomial-time solvable on classes of graphs excluding a topological minor, or even on planar graphs. 

We now show that the problem is fixed-parameter tractable with respect to the parameter $k$. 

\begin{theorem}\label{thm:path-k}
    The \textsc{Shortest Path Discovery} problem in the token sliding model is fixed-parameter tractable with respect to parameter $k$. 
\end{theorem}
\begin{proof}
    Let $(G,s,t,S,\budget)$ be an instance of \textsc{SPD}. For every $v\in V(G)$ we compute its distance to~$s$ and delete $v$ if the distance is larger than $k$. 
    We enumerate the tokens in $S$ as $s_0,s_1,\ldots, s_{k-1}$ such that token~$s_i$ shall slide to a vertex at distance $i$ from $s$. There are $k!$ such enumerations. Now we orient and assign weights to the edges of $G$ to obtain a weighted directed graph $H$. 
    If $uv$ is an edge such that $\mathrm{dist}_G(s,u)=i$ and $\mathrm{dist}_G(s,v)=i+1$, then we orient the edge as $(u,v)$ and assign it weight $\mathrm{dist}_G(s_{i+1},v)$, that is, the cost of moving token $s_{i+1}$ to vertex $v$. We delete all edges that did not receive a weight, that is, all edges that connect vertices at the same distance from $s$. 

    Since $H$ contains only edges between vertices of distance $i$ to distance $i+1$ from $s$, an $s$-$t$-path in $H$ corresponds to a shortest $s$-$t$-path in~$G$. 
    Furthermore, by definition of the weights, a shortest $s$-$t$-path $P$ (with respect to the weights) corresponds exactly to the discovery of a shortest path in~$G$ where token $s_i$ slides to the vertex of $P$ which is at distance $i$ from $s$. 
    We can therefore simply compute a shortest $s$-$t$-path in $H$ and verify whether its weight is at most $\budget - \mathrm{dist}_G(s, s_0)$. Since~$H$ is a DAG we can compute such a path in time $\Oof(n+m)$ and, consequently, the full algorithm runs in time $\Oof(k!(n+m))$.
\end{proof}

Finally, we prove that the \textsc{Shortest Path Discovery} problem remains fixed-parameter tractable when  parameterized by $\budget$.

\begin{theorem}\label{thm:path-b}
    The \textsc{Shortest Path Discovery} problem in the token sliding model is fixed-parameter tractable with respect to parameter $\budget$. 
\end{theorem}

\begin{proof}
Let $(G,s,t,S,\budget)$ be an instance of \textsc{SPD}. Let us call the vertices at distance $i$ from $s$ the vertices at level $i$. If we can discover a solution with budget $\budget$, then it must be the case that all but at most $\budget$ tokens of $S$ are already in their correct positions. In particular, there can be at most~$\budget$ many levels that do not contain a token from $S$ and at most $\budget$ many levels containing two or more tokens (and each level can contain at most $\budget+1$ many tokens). We call levels not containing exactly one token \emph{very bad} levels. Now, for a level $i$ containing exactly one token on vertex $v$ (not a very bad level), we say that $i$ is a \emph{bad} level whenever one of the following is true:
\begin{itemize}
\item $i - 1 \geq 0$, level $i - 1$ contains exactly one token on vertex $u$, and $uv \notin E(G)$;  or 
\item $i + 1 \leq \mathrm{dist}(s,t) + 1$, level $i + 1$ contains exactly one token on vertex $w$, and $vw \notin E(G)$.
\end{itemize}  

Note that we can have at most $\budget$ very bad levels and at most $3\budget$ bad levels (every $3$ bad levels require at least one unit of the budget); otherwise we can reject the instance. Each bad or very bad level contains at most $\budget + 1$ tokens and at most $\budget$ tokens do not belong to a level between $s$ and~$t$. Hence, the total number of tokens that potentially have to move is so far $5\budget(\budget + 1) = 5\budget^2 + 5\budget$ (located on at most $5\budget$ levels). 

We call a level that is neither bad nor very bad a \emph{good} level. Note that we can group the good levels into at most $5\budget + 1$ groups of consecutive levels. Moreover, the tokens of each group form a path. 
We denote those paths by $P_1,\ldots, P_q$, for $q\leq 5\budget+1$. Consider some $P_i$ with at least $2\budget+1$ many vertices, say $P_i=v_1\ldots v_{\ell}$ for $\ell \geq 2\budget+1$. Let $v_x$ be a vertex of $P_i$ such that $\ell + 1 \leq x \leq \ell - (\budget + 1)$. In other words, there are at least $\budget$ vertices preceding and at least $\budget$ vertices succeeding $v_x$ in $P_i$. 
In what follows we assume, without loss of generality, that the token on $v_x$ does not move whenever $v_x$ belongs to the final shortest path between $s$ and $t$. This is a safe assumption because if any other token must pass via $v_x$ to reach its destination we simply swap the roles of both tokens. In other words, getting a token on $v_x$ does not consume any units of the budget. 
We claim that in a shortest discovery sequence of at most $\budget$ token slides the token on vertex~$v_x$ does not move. This follows from the fact that otherwise it must be the case that either all tokens on~$v_y$, $1 \leq y < x$, or all tokens on $v_z$, $\ell \geq z > x$, must move. However, as there are at least~$\budget$ tokens to move in either case this contradicts the fact that our discovery sequence slides no more than~$\budget$ tokens. To prove that either every token before or after $v_x$ must move in a shortest discovery sequence we aim towards a contradiction. Assume that the token at level $x$ moves but there are two levels $y < x$ and $z > x$ such that the tokens on level $y$ and $z$ do not move. 
Since the subpath between levels $y$ and $z$ is already a path connecting $v_y$ and $v_z$ (going through $v_x$), the movement of the token at level $x$ can be ignored, contradicting our assumption of a shortest discovery sequence. Since at most $\budget$ tokens can move, and by symmetry, we conclude that at most $2\budget$ many tokens on the two boundaries of $P_i$ can move.
We call the tokens on good levels that are not at distance at most $\budget$ from some boundary of a $P_i$ \emph{fixed} tokens. Note that the number of tokens that are not fixed is at most $10\budget^2 + 2\budget$. 

Putting it all together, we conclude that the set of tokens that can potentially move has size at most $15\budget^2 + 7\budget$. 
We can now proceed just as in the proof of \Cref{thm:path-k} to obtain the required algorithm for parameter $\budget$; we can now guess the subset of tokens that move as well as the level each moving token will occupy.
\end{proof}
\section{Matchings}\label{sec:matching}

A \emph{matching} in a graph $G$ is a set of edges $M \subseteq E(G)$ such that each vertex $v \in V(G)$ appears in at most one edge in $M$. In the \textsc{Matching} problem we are given a graph $G$ and an integer $k$ and the goal is to compute a matching of size $k$ in $G$. 

\subsection{Rainbow matchings}
We show \textsf{NP}-hardness of the \mdis~problem via a reduction from the \rmat~problem. 
Let $(G, \varphi)$ be an edge-colored graph.
A matching $M \subseteq E(G)$ is said to be a \emph{rainbow matching} if all edges in $M$ have pairwise distinct colors. 
An edge-coloring $\phi$ is {\em proper} if every pair of adjacent edges have distinct colors, where edges are adjacent whenever they are incident to a common vertex.  
Formally, the \rmat~problem is defined as follows. 
Given an edge-colored graph $(G, \phi:E(G) \to \Cmc)$ and an integer $k$, decide whether there exists a rainbow matching of size~$k$ in $G$. 

\begin{theorem}[\cite{DBLP:journals/tcs/LeP14}]
The \rmat~problem is \NP-complete, even when restricted to properly edge-colored paths. 
\end{theorem}

\subsection{Matching discovery}

In the \mdis~problem in the token sliding model we are given a graph $G$, a starting configuration $S~\subseteq~E(G)$ of size $k$, and a non-negative integer $b$.
The goal is to decide whether we can discover a $k$-sized matching~$M$ (starting from $S$) using at most $b$ token slides. Recall that a token on some edge $e$ can only slide (in one step) to some edge $e'$ such that $e$ and $e'$ are adjacent, i.e., share a vertex. 
We denote an instance of \mdis by a tuple $(G, S,b)$. 
We always use $k$ to denote the size of~$S$.

\subsubsection{\NP-hardness on 2-degenerate bipartite graphs}

\begin{theorem}
\label{thm:matching:hardness}
    The \mdis~problem in the token sliding model is \NP-hard on 2-degenerate bipartite graphs.
\end{theorem}
\begin{proof}
We show \NP-hardness by a reduction from the \rmat~problem on properly edge-colored paths.
Let $(G, \phi, k)$ be an instance the \rmat~problem where $G$ is a path and $\phi$ is a proper edge-coloring. Let $\Cmc$ denote the set of colors appearing in $G$. 
We construct an instance $(H, S, b)$ of the \mdis~problem as follows. 
The graph $G$ is a bipartite graph since $G$ is a path. 
Let us partition the vertex set of $G$ into two sets $A$ and $B$. 
Let $\ell = |\Cmc| - k$. 
For each edge $e= uv \in E(G)$, we add two vertices $x_u^e$ and $x_v^e$ in $V(H)$, and connect them by an edge. 
Let $D = \{u \in V(G) \mid |N_G(u)| = 2\}$.
We use the set $D$ to denote the set of all degree-two vertices in $G$. 
For each vertex $u \in D$, we add a path of three vertices $y^1_u, y^2_u, y^3_u$ in $H$. 
Let $e_1$ and $e_2$ be the edges incident on a degree two vertex $u$, then we add the edges $x_u^{e_1}y_u^1$ and $x_u^{e_2}y_u^1$ in $H$. 
For each color $c \in \Cmc$, we add a path of four vertices $z_c^1,z_c^2,z_c^3,z_c^4$ in $H$. 
For each color $c \in \Cmc$, and for each edge $e= uv \in E(G)$ with $\phi(e) = c$ (without loss of generality assume that $u \in A$), we add an edge $z_c^1 x_u^e$ in $H$. 
Finally, we add two $\ell$-sized vertex subsets $W$ and $W'$. 
For each color $c \in \Cmc$, for each $w \in W$ and $w' \in W'$, we add the edges $z_c^1w$ and $z_c^4w'$ 
in $H$.
If $\ell = 0$, that is, $|\Cmc| = k$, then we will have the empty vertex subsets $W$ and $W'$. 
This completes the construction of the graph~$H$. 
Observe that $H$ is a $2$-degenerate bipartite graph; we construct the graph $H$ with no odd-cycle. 
Moreover, removal of degree-one and degree-two vertices results in a forest. 
Therefore, the graph~$H$ is $2$-degenerate. 

Now, we define the initial configuration $S$ as follows:
$$S =  \bigcup_{u \in D}\{y_u^1y_u^2, y_u^2y_u^3\} \cup \bigcup_{c \in \Cmc}\{z_c^1z_c^2, z_c^2z_c^3, z_c^3z_c^4\}.$$ Finally, we set the budget as $b = 3k+2\ell + n-2$. 

We show that $G$ contains a rainbow matching of size at least $k$ if and only if $(H, S, b)$ is a yes-instance of the \mdis~problem.

\begin{figure}[h]
    \centering
    \begin{tikzpicture}
    \coordinate (o) at (0,0);
    \coordinate (u1) at (-2.5, 0);
    \coordinate (u2) at (-1.5, 0);
    \coordinate (u3) at (-0.5, 0);
    \coordinate (u4) at (0.5, 0);
    \coordinate (u5) at (1.5, 0);
    \coordinate (u6) at (2.5, 0);
    \draw[red] (u1) -- (u2);
    \draw[red, decoration = {zigzag,segment length = 2mm, amplitude = 1mm}, decorate] (u1) -- (u2);
    \draw[blue] (u2) -- (u3);
    \draw[red] (u3) -- (u4);
    \draw[blue] (u4) -- (u5);
    \draw[blue, decoration = {zigzag,segment length = 2mm, amplitude = 1mm}, decorate] (u4) -- (u5);
    \draw[green] (u5) -- (u6);
    \fill[black, draw=black] (u1) circle (0.05cm) node[below] {$u_1$};
    \fill[white, draw=black] (u2) circle (0.05cm) node[black, below] {$u_2$};
    \fill[black, draw=black] (u3) circle (0.05cm) node[below] {$u_3$};
    \fill[white, draw=black] (u4) circle (0.05cm) node[black, below] {$u_4$};
    \fill[black, draw=black] (u5) circle (0.05cm) node[below] {$u_5$};
    \fill[white, draw=black] (u6) circle (0.05cm) node[black, below] {$u_6$};
    \node[black] at ($(o)+(0, -1)$) {(a)};

    \coordinate (o) at (0,-4); 
    \coordinate (u1) at ($(o) + (-5.5, 0)$);
    \coordinate (u21) at ($(o) + (-4.5, 0)$);
    \coordinate (u22) at ($(o) + (-3, 0)$);
    \coordinate (u31) at ($(o) + (-2, 0)$);
    \coordinate (u32) at ($(o) + (-0.5, 0)$);
    \coordinate (u41) at ($(o) + (0.5, 0)$);
    \coordinate (u42) at ($(o) + (2, 0)$);
    \coordinate (u51) at ($(o) + (3, 0)$);
    \coordinate (u52) at ($(o) + (4.5, 0)$);
    \coordinate (u6) at ($(o) + (5.5, 0)$);
    \coordinate (yu21) at ($(u21) + (0.75,0.5)$);
    \coordinate (yu22) at ($(u21) + (1.5,1)$);
    \coordinate (yu23) at ($(u21) + (2.25,1.5)$);
    \coordinate (yu31) at ($(u31) + (0.75,0.5)$);
    \coordinate (yu32) at ($(u31) + (1.5,1)$);
    \coordinate (yu33) at ($(u31) + (2.25,1.5)$);
    \coordinate (yu41) at ($(u41) + (0.75,0.5)$);
    \coordinate (yu42) at ($(u41) + (1.5,1)$);
    \coordinate (yu43) at ($(u41) + (2.25,1.5)$);
    \coordinate (yu51) at ($(u51) + (0.75,0.5)$);
    \coordinate (yu52) at ($(u51) + (1.5,1)$);
    \coordinate (yu53) at ($(u51) + (2.25,1.5)$);
    \coordinate (z11) at ($(o) + (-1.5, -1.5)$);
    \coordinate (z12) at ($(z11) + (0, -0.75)$);
    \coordinate (z13) at ($(z12) + (0, -0.75)$);
    \coordinate (z14) at ($(z13) + (0, -0.75)$);
    \coordinate (z21) at ($(o) + (0, -1.5)$);
    \coordinate (z22) at ($(z21) + (0, -0.75)$);
    \coordinate (z23) at ($(z22) + (0, -0.75)$);
    \coordinate (z24) at ($(z23) + (0, -0.75)$);
    \coordinate (z31) at ($(o) + (1.5, -1.5)$);
    \coordinate (z32) at ($(z31) + (0, -0.75)$);
    \coordinate (z33) at ($(z32) + (0, -0.75)$);
    \coordinate (z34) at ($(z33) + (0, -0.75)$); 
    \coordinate (w1) at ($(z32) + (1, 0)$);
    \coordinate (w2) at ($(z33) + (1, 0)$);

    \draw[black, dotted] (u1) -- (u21);
    \draw[gray, decoration = {zigzag,segment length = 2mm, amplitude = 1mm}, decorate] (u1) -- (u21);
    \draw[black, dotted] (u22) -- (u31);
    \draw[black, dotted] (u32) -- (u41);
    \draw[black, dotted] (u42) -- (u51);
    \draw[gray, decoration = {zigzag,segment length = 2mm, amplitude = 1mm}, decorate] (u42) -- (u51);
    \draw[black, dotted] (u52) -- (u6);
    \draw[black, dotted] (u21) -- (yu21);
    \draw[black, dotted] (u22) -- (yu21);
    \draw[gray, decoration = {zigzag,segment length = 2mm, amplitude = 1mm}, decorate] (u22) -- (yu21);
    \draw[black] (yu21) -- (yu22);
    \draw[black] (yu22) -- (yu23);
    \draw[gray, decoration = {zigzag,segment length = 2mm, amplitude = 1mm}, decorate] (yu22) -- (yu23);
    \draw[black, dotted] (u31) -- (yu31);
    \draw[black, dotted] (u32) -- (yu31);
    \draw[gray, decoration = {zigzag,segment length = 2mm, amplitude = 1mm}, decorate] (u31) -- (yu31);
    \draw[black] (yu31) -- (yu32);
    \draw[black] (yu32) -- (yu33);
    \draw[gray, decoration = {zigzag,segment length = 2mm, amplitude = 1mm}, decorate] (yu32) -- (yu33);
    \draw[black, dotted] (u41) -- (yu41);
    \draw[black, dotted] (u42) -- (yu41);
    \draw[gray, decoration = {zigzag,segment length = 2mm, amplitude = 1mm}, decorate] (u41) -- (yu41);
    \draw[black] (yu41) -- (yu42);
    \draw[black] (yu42) -- (yu43);
    \draw[gray, decoration = {zigzag,segment length = 2mm, amplitude = 1mm}, decorate] (yu42) -- (yu43);
    \draw[black, dotted] (u51) -- (yu51);
    \draw[black, dotted] (u52) -- (yu51);
    \draw[gray, decoration = {zigzag,segment length = 2mm, amplitude = 1mm}, decorate] (u52) -- (yu51);
    \draw[black] (yu51) -- (yu52);
    \draw[black] (yu52) -- (yu53);
    \draw[gray, decoration = {zigzag,segment length = 2mm, amplitude = 1mm}, decorate] (yu52) -- (yu53);
    \draw[black] (z11) -- (z12);
    \draw[gray, decoration = {zigzag,segment length = 2mm, amplitude = 1mm}, decorate] (z11) -- (z12);
    \draw[black] (z12) -- (z13);
    \draw[black] (z13) -- (z14);
    \draw[gray, decoration = {zigzag,segment length = 2mm, amplitude = 1mm}, decorate] (z13) -- (z14);
    \draw[black] (z21) -- (z22);
    \draw[gray, decoration = {zigzag,segment length = 2mm, amplitude = 1mm}, decorate] (z21) -- (z22);
    \draw[black] (z22) -- (z23);
    \draw[black] (z23) -- (z24);
    \draw[gray, decoration = {zigzag,segment length = 2mm, amplitude = 1mm}, decorate] (z23) -- (z24);
    \draw[black] (z31) -- (z32);
    \draw[black] (z32) -- (z33);
    \draw[gray, decoration = {zigzag,segment length = 2mm, amplitude = 1mm}, decorate] (z32) -- (z33);
    \draw[black] (z33) -- (z34);
    \draw[black, dotted] (z11) -- (u1);
    \draw[black, dotted] (z11) -- (u32);
    \draw[black, dotted] (z21) -- (u31);
    \draw[black, dotted] (z21) -- (u51);
    \draw[black, dotted] (z31) -- (u52);
    \draw[black, dotted] (z11) -- (w1);
    \draw[black, dotted] (z21) -- (w1);
    \draw[black, dotted] (z31) -- (w1);
    \draw[gray, decoration = {zigzag,segment length = 2mm, amplitude = 1mm}, decorate] (z31) -- (w1);
    \draw[black, dotted] (z14) -- (w2);
    \draw[black, dotted] (z24) -- (w2);
    \draw[black, dotted] (z34) -- (w2);
    \draw[gray, decoration = {zigzag,segment length = 2mm, amplitude = 1mm}, decorate] (z34) -- (w2);

    \draw[black, fill = black] (u1) circle (0.05cm) node[left] {$x_{u_1}^{u_1u_2}$};
    \draw[black, fill = white] (u21) circle (0.05cm); 
    \draw[black, fill = white] (u22) circle (0.05cm); 
    \draw[black, fill = black] (u31) circle (0.05cm) node[above] {$x_{u_3}^{u_2u_3}$};
    \draw[black, fill = black] (u32) circle (0.05cm) node[above] {$x_{u_3}^{u_3u_4}$};
    \draw[black, fill = white] (u41) circle (0.05cm); 
    \draw[black, fill = white] (u42) circle (0.05cm);
    \draw[black, fill = black] (u51) circle (0.05cm); 
    \draw[black, fill = black] (u52) circle (0.05cm); 
    \draw[black, fill = white] (u6) circle (0.05cm) node[right] {$x_{u_6}^{u_5u_6}$};
    \draw[black, fill = black] (yu21) circle (0.05cm); 
    \draw[black, fill = white] (yu22) circle (0.05cm); 
    \draw[black, fill = black] (yu23) circle (0.05cm); 
    \draw[black, fill = white] (yu31) circle (0.05cm) node[above] {$y_{u_3}^1$};
    \draw[black, fill = black] (yu32) circle (0.05cm) node[above] {$y_{u_3}^2$};
    \draw[black, fill = white] (yu33) circle (0.05cm) node[above] {$y_{u_3}^3$};
    \draw[black, fill = black] (yu41) circle (0.05cm); 
    \draw[black, fill = white] (yu42) circle (0.05cm); 
    \draw[black, fill = black] (yu43) circle (0.05cm);
    \draw[black, fill = white] (yu51) circle (0.05cm); 
    \draw[black, fill = black] (yu52) circle (0.05cm); 
    \draw[black, fill = white] (yu53) circle (0.05cm); 
    \draw[black, fill = white] (z11) circle (0.05cm); 
    \draw[black, fill = black] (z12) circle (0.05cm); 
    \draw[black, fill = white] (z13) circle (0.05cm); 
    \draw[black, fill = black] (z14) circle (0.05cm); 
    \draw[black, fill = white] (z21) circle (0.05cm) node[above] {$z_{blue}^1$};
    \draw[black, fill = black] (z22) circle (0.05cm) node[left] {$z_{blue}^2$};
    \draw[black, fill = white] (z23) circle (0.05cm) node[left] {$z_{blue}^3$};
    \draw[black, fill = black] (z24) circle (0.05cm) node[below] {$z_{blue}^4$};
    \draw[black, fill = white] (z31) circle (0.05cm); 
    \draw[black, fill = black] (z32) circle (0.05cm); 
    \draw[black, fill = white] (z33) circle (0.05cm); 
    \draw[black, fill = black] (z34) circle (0.05cm);
    \draw[black, fill = black] (w1) circle (0.05cm) node[right] {$w_1$};
    \draw[black, fill = white] (w2) circle (0.05cm) node[right] {$w_1'$};
    \node[black] at ($(z24)+(0, -1)$) {(b)};

\end{tikzpicture}
    \caption{An illustration of the \NP-hardness reduction. The filled and unfilled vertices show that the graphs are bipartite. (a): A input path $G$. A feasible solution is marked in zigzag-edges. (b): The reduced graph $H$. The solid edges denote the initial configuration, and the rest of the graph edges are dotted. The gray zigzag edges denote a feasible solution.}
    \label{fig:enter-label}
\end{figure}
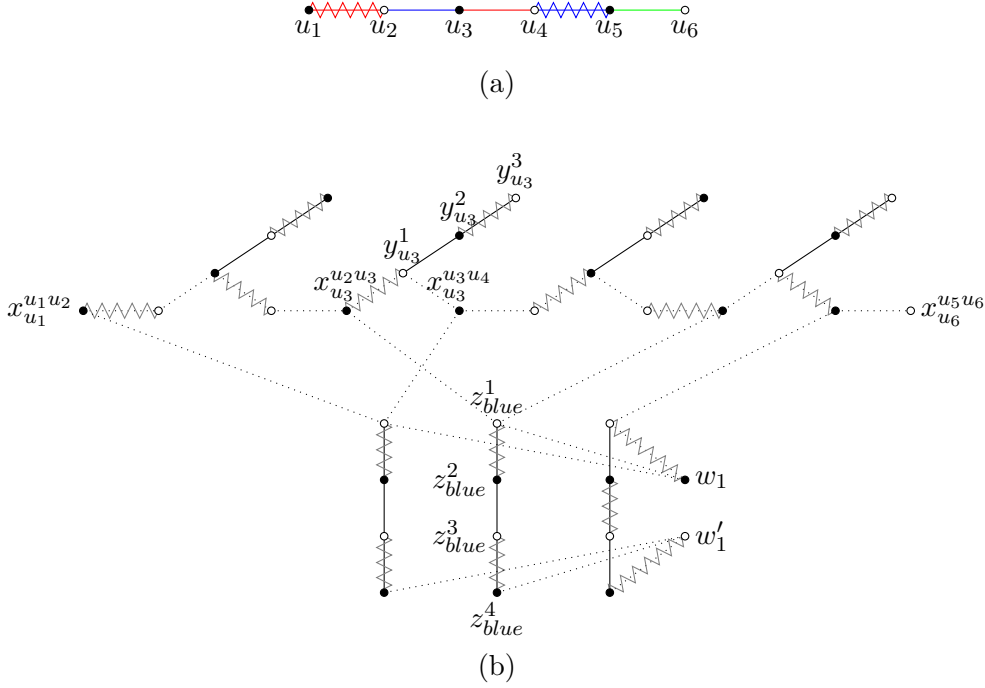

Assume that $(G, \phi, k)$ is a yes-instance of the \rmat~problem. 
That is, let $R \subseteq E(G)$ be a rainbow matching of size at least $k$. 
Without loss of generality, we assume that $R$ is of size exactly $k$. 
Otherwise, we remove some edges to make $R$ a rainbow matching of size exactly $k$.
Let $P = \{c \in \Cmc \mid \exists e \in R, \phi(e) = c\}$ be the set of matched colors. 
It is clear that $|P| = k$ since $R$ is a rainbow matching. 
We define a injective function $\pi: D \to V(H)$ such that for every degree-two vertex $u$, $\pi(u) = x_u^e$ for some edge $e$ incident on $u$ but not in $R$. 
Since $M$ is a matching, for every vertex $u \in D$, there always exists an unmatched incidental edge. 
Therefore, the function~$\pi$ is well defined. 
We define two bijections $\alpha: \{z_c^1 \mid c \in \Cmc\setminus P\} \to W$ and $\beta: \{z_c^4 \mid c \in \Cmc\setminus P\} \to W'$. 
Since $H[\{z_c^1 \mid c \in \Cmc\setminus P\} \to W]$ and $H[\{z_c^4 \mid c \in \Cmc\setminus P\} \to W']$ are complete bipartite graphs of the same order, we pick $\alpha$ and $\beta$ as some arbitrary bijective functions. 
We construct a matching $M$ in $H$ as follows:

\[M = \bigcup_{e=uv \in M} \{x_u^ex_v^e\} \cup \bigcup_{u \in D}\{\pi(u)y_u^1, y_u^2y_u^3\} \cup \bigcup_{c \in P} \{z_c^1z_c^2, z_c^3z_c^4\} \cup \bigcup_{c \in \Cmc\setminus P}\{\alpha(z_c^1)z_c^1, z_c^2z_c^3, z_c^4\beta(z_c^4)\}.\]

It is clear that $M$ is a matching since no two edges share a common incidental vertex. 
We claim that the matching $M$ can be obtained from $S$ in at most $b$ sliding steps. 
For each $u \in D$, $y_u^2y_u^3$ is a common edge in both matchings ($S$ and $M$), and we slide a token from $y_u^1y_u^2$ to $y_u^1\pi(u)$. 
This costs $n-2$ sliding steps. 
Next, for each color $c \in P$, we keep the tokens on $z_c^1z_c^2$ and $z_c^3z_c^4$, and we move a token from $z_c^2z_c^3$ to $x_u^ex_v^e$ for the edge $e = uv \in R$ with $\phi(e) = c$. 
Each such token movement costs three sliding steps, hence, we spent $3k$ sliding steps. 
Finally, for each color $c \in \Cmc \setminus P$, we keep the token on $z_c^2z_c^3$, and we move the tokens from $z_c^1z_c^2$ and $z_c^3z_c^4$ to $\alpha(z_c^1)z_c^1$ and $z_c^4\beta(z_c^4)$, respectively. 
For each such color, the token movement costs two sliding steps, hence, we spent $2(|\Cmc|-|P|) = 2\ell$ sliding steps. 
Therefore, $M$ can be obtained from $S$ in $3k+2\ell+n-2$ sliding steps. 
Thus, $(H, S, b)$ is a yes-instance of the \mdis~problem.

Now assume that $(H, S, b)$ is a yes-instance of the \mdis~problem. 
That is, let $M \subseteq E(H)$ be a matching of size $|S|$ that can be obtained from $S$ using at most $b$ sliding steps. 
We first describe some interesting properties of any feasible solution to the instance $(H, S, b)$ of the \mdis~problem. 
If the initial configuration contains a path of length two, then we must spend at least one token slide to fix the path. 
Similarly, if the initial configuration contains a path of length three, then we must spend at least two token slides to fix the path. 
If both end vertices of the path have unsaturated neighbors, then we can fix the matching in exactly two token slides. 
Otherwise, we need at least three token slides to fix the matching. 
Since $S$ contains $n-2$ paths of length two and $|\Cmc|$ paths of length three, a feasible solution must use at least $n-2+2|\Cmc|$ slides. 
Moreover, at most $\ell$ paths of length three can be fixed in exactly two steps since $W$ and $W'$ are of size $\ell$. 
Therefore, a feasible solution must use at least $n-2 + 3k + 2\ell = b$ token slides. 

Let $X = \bigcup_{e=uv \in E(G)}\{x_u^e, x_v^e\}$. 
Observe that $H[X]$ is a matching of size $n-1$. 
Let us define a bijection $\pi: E(G) \to E(H[X])$ such that $\pi$ maps each edge $e=uv \in E(G)$ to the edge $x_u^ex_v^e$. 
For each $u \in D$, the token on $y_u^1y_u^2$ can be slid to $y_u^1x_u^e$ for some edge $e$ incident on $u$ since these are the possible options with one token slide.
For any vertex $u \in D$ with incident edges $e=uv_1$ and $e' = uv_2$, the matching $M$ cannot contain both $x_u^ex_{v_1}^e$ and $x_u^{e'}x_{v_2}^{e'}$. 
In case both edges are in the matching $M$, then the token on $y_u^1y_u^2$ needs to slide for at least two steps which contradicts  feasibility. 
Therefore, $R = \bigcup_{e \in M\cap E(H[X])} \{\pi^{-1}(e)\}$ is a matching in $G$. 
Next, we discuss the rainbow property of $R$. 
We can fix at most $\ell$ paths of length three in exactly two steps. 
Let $P \subseteq \Cmc$ be the set of $k$ colors such that the paths of length three corresponding to the colors in $P$ are fixed in three steps. 
Since there paths do not have unmatched neighbors on both end vertices, the only option is to move the token on the middle edge to an unmatched edge with both end vertices being unmatched. 
By our construction of the graph $H$, a token on a path (the middle token) corresponding to a color $c$ can reach an unmatched edge $e$ in three steps if and only if $e \in E(H[X])$ and $\phi(\pi^{-1}(e)) = c$. 
Therefore, $R$ is a matching and is colorful. 
Thus, $(G, \phi, k)$ is a yes-instance of the \rmat~problem. 
\end{proof}

\tikzset{
xdashed/.style={to path={
\pgfextra{
 \draw[black, thick, dashed]
      (\tikztostart) -- (\tikztotarget);} (\tikztotarget) \tikztonodes}},
xsolid/.style={to path={
\pgfextra{
 \draw[black, thick]
      (\tikztostart) -- (\tikztotarget);} (\tikztotarget) \tikztonodes}},
xnothing/.style={to path={
\pgfextra{
 \draw[white, thick]
      (\tikztostart) -- (\tikztotarget);} (\tikztotarget) \tikztonodes}}

}

\subsubsection{\textsf{W[1]}-hardness for parameter \texorpdfstring{$b$}{b} on 3-degenerate graphs}

\begin{theorem}
    The \mdis~problem in the token sliding model is \textsf{W[1]}-hard when parameterized by $b$, even on the class of $3$-degenerate graphs. 
\end{theorem}

We present a parameterized reduction from the \mcc~problem, which is known to be \textsf{W[1]}-hard~\cite{downey1995fixed}. 
Recall that in the \mcc~problem, we are given a graph~$G$ and an integer $\kappa$, where $V(G)$ is partitioned into $\kappa$ independent sets \mbox{$V_1$, $V_2$, $\ldots$, $V_\kappa$}, and the goal is to find a multicolored  clique of size $\kappa$, i.e., a clique containing one vertex from each set~$V_i$, for $i \in [\kappa]$.

Let $(G, \kappa)$ be an instance of the \mcc~problem. 
The edge set $E(G)$ can be partitioned into $E = \{E_{i,j} \mid 1 \leq i < j \leq \kappa\}$ such that for each $1 \leq i < j \leq \kappa$, we have \mbox{$E_{i,j} = \{uv \in E \mid u\in V_i$}, $v \in V_j\}$. 
We construct an instance $(H, S \subseteq E(H), b)$ of the \mdis~problem.
The graph $H$ has an induced subgraph $H_i$ for each $i \in [\kappa]$ and an induced subgraph $H_{i,j}$ for each $1 \leq i < j \leq \kappa$. 
We refer to these induced subgraphs as {\em edge-blocks} and {\em vertex-blocks}, respectively.
Both vertex-blocks and edge-blocks are constructed using a fundamental structure called gadget $\Ttt$, and they are connected using a structure called chain $\Ctt$. \\

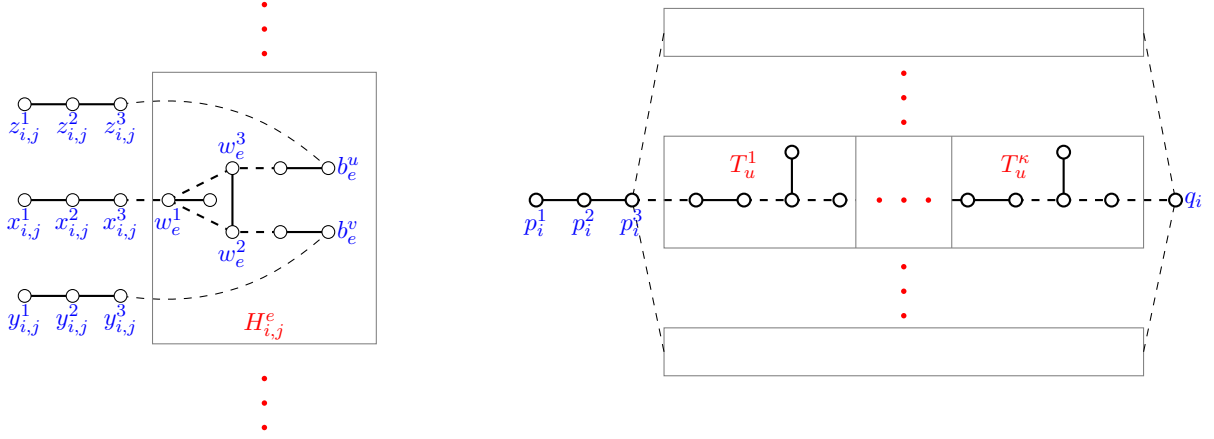
\begin{figure}[ht]
    \centering
    \begin{tikzpicture}[scale=0.85, every node/.style={transform shape}]
        \def\x{9.5}
        \def\y{4}
        \def\z{1}
        \draw[black, thick] (4,1) -- (4.75,1) -- (5.5,1);
        \draw[black, thick, dashed] (5.5,1) -- (6.25,1) -- (7.25, 1.5) -- (8, 1.5);
        \draw[black, thick] (7, 1) -- (6.25, 1);
        \draw[black, thick] (7.25, 1.5) -- (7.25, 0.5);
        \draw[black, thick, dashed] (6.25, 1) -- (7.25, 0.5) -- (8, 0.5);
        \draw[black, thick] (8,1.5) -- (8.75,1.5);
        \draw[black, thick] (8,0.5) -- (8.75,0.5);
        \draw[black, dashed] (5.5, 2.5) to[bend left=25] (8.75, 1.5);
        \draw[black, dashed] (5.5, -0.5) to[bend right=25] (8.75, 0.5);
        \draw[black, fill = white] (4, 1) circle (0.1cm) node[blue, below] {$x_{i,j}^1$};
        \draw[black, fill = white] (4.75, 1) circle (0.1cm) node[blue, below] {$x_{i,j}^2$};
        \draw[black, fill = white] (5.5, 1) circle (0.1cm) node[blue, below] {$x_{i,j}^3$};
        \draw[black, fill = white] (6.25, 1) circle (0.1cm) node[blue, below] {$w_e^1$};
        \draw[black, fill = white] (7.25, 1.5) circle (0.1cm) node[blue, above] {$w_e^3$};
        \draw[black, fill = white] (8, 1.5) circle (0.1cm);
        \draw[black, fill = white] (6.9, 1) circle (0.1cm);
        \draw[black, fill = white] (7.25, 0.5) circle (0.1cm) node[blue, below] {$w_e^2$};
        \draw[black, fill = white] (8, 0.5) circle (0.1cm);
        \draw[black, fill = white] (8.75, 1.5) circle (0.1cm) node[blue, right] {$b_e^u$};
        \draw[black, fill = white] (8.75, 0.5) circle (0.1cm) node[blue, right] {$b_e^v$};
        \draw[black, thick] (4,2.5) -- (4.75,2.5) -- (5.5,2.5);
        \draw[black, thick] (4,-0.5) -- (4.75,-0.5) -- (5.5,-0.5);
        \draw[black, fill = white] (4, 2.5) circle (0.1cm) node[blue, below] {$z_{i,j}^1$};
        \draw[black, fill = white] (4.75, 2.5) circle (0.1cm) node[blue, below] {$z_{i,j}^2$};
        \draw[black, fill = white] (5.5, 2.5) circle (0.1cm) node[blue, below] {$z_{i,j}^3$};
        \draw[black, fill = white] (4, -0.5) circle (0.1cm) node[blue, below] {$y_{i,j}^1$};
        \draw[black, fill = white] (4.75, -0.5) circle (0.1cm) node[blue, below] {$y_{i,j}^2$};
        \draw[black, fill = white] (5.5, -0.5) circle (0.1cm) node[blue, below] {$y_{i,j}^3$};
        \draw[gray, thin] (6,3) rectangle (9.5, -1.25) node[below=1.5, midway, red, thick] {$H_{i,j}^e$};
        \path (7.75,3.5) -- (7.75,4) node [red, font=\Huge, midway, sloped] {$\dots$};
        \path (7.75,-1.75) -- (7.75,-2.75) node [red, font=\Huge, midway, sloped] {$\dots$};

        \draw (12,\z) -- (12.75,\z) to[xsolid] (13.5,\z) to[xdashed] (14.5,\z) to[xsolid] (15.25, \z) -- (16,\z) to[xdashed] (17,\z);
        \draw (18.5, \z) to[xsolid] (19.5,\z) -- (20.25,\z) -- (21, \z) to[xdashed] (22,\z);
        \draw[black, thick] (16, \z) -- (16, \z+0.75);
        \draw[black, thick] (20.25, \z) -- (20.25, \z+0.75);
        \draw[gray, thin] (14, \z+1) rectangle (21.5, \z-0.75);
        \draw[gray, thin] (14, \z+3) rectangle (21.5, \z+2.25);
        \draw[gray, thin] (14, \z-2) rectangle (21.5, \z-2.75);
        \draw[gray, thin] (17,\z+1) -- (17, \z-0.75);
        \draw[gray, thin] (18.5,\z+1) -- (18.5, \z-0.75);
        \draw[black, dashed] (14,\z-2.375) -- (13.5,\z) -- (14,\z+2.625);
        \draw[black, dashed] (21.5,\z-2.375) -- (22,\z) -- (21.5,\z+2.625);
        \path (17.4,\z) -- (18.25,\z) node [red, font=\Huge, midway, sloped] {$\dots$};
        \path (17.75,\z+1.35) -- (17.75,\z+2) node [red, font=\Huge, midway, sloped] {$\dots$};
        \path (17.75,\z-1) -- (17.75,\z-2) node [red, font=\Huge, midway, sloped] {$\dots$};
        \draw[black, thick, fill=white] (12,\z) circle (0.1cm) node[blue, below] {$p_i^1$};
        \draw[black, thick, fill=white] (12.75,\z) circle (0.1cm) node[blue, below] {$p_i^2$};
        \draw[black, thick, fill=white] (13.5,\z) circle (0.1cm) node[blue, below] {$p_i^3$};
        \draw[black, thick, fill=white] (14.5,\z) circle (0.1cm);
        \draw[black, thick, fill=white] (15.25,\z) circle (0.1cm) node[red, above=0.2cm] {$T_{u}^1$};
        \draw[black, thick, fill=white] (16,\z) circle (0.1cm);
        \draw[black, thick, fill=white] (16.75,\z) circle (0.1cm);
        \draw[black, thick, fill=white] (18.75,\z) circle (0.1cm);
        \draw[black, thick, fill=white] (19.5,\z) circle (0.1cm) node[red, above=0.2cm] {$T_{u}^{\kappa}$};
        \draw[black, thick, fill=white] (20.25,\z) circle (0.1cm);
        \draw[black, thick, fill=white] (21,\z) circle (0.1cm);
        \draw[black, thick, fill=white] (22,\z) circle (0.1cm) node[blue, right] {$q_i$};
        \draw[black, thick, fill=white] (20.25,\z+0.75) circle (0.1cm);
        \draw[black, thick, fill=white] (16,\z+0.75) circle (0.1cm);
        
    \end{tikzpicture}
    \caption{An illustration of an edge-block (left) and a vertex-block (right)}
    \label{fig:mdis:whardblocks}
\end{figure}

\noindent{\bf $\Ttt$-gadget.} A gadget $\Ttt$ consists of a path of four vertices $c_1, c_2, c_3, c_4$ and a neighbor $c_5$ to $c_3$. We define $S(\Ttt) = \{c_1c_2, c_3c_5\}$ as the initial configuration for the gadget $\Ttt$.\\

\noindent{\bf Edge-block.} For each $1 \leq i < j \leq \kappa$, we construct the edge-block $H_{i,j}$ as follows. 
First, in $H_{i,j}$, we add three paths each having three vertices $X_{i,j} = \{x_{i,j}^1, x_{i,j}^2, x_{i,j}^3\}, Y_{i,j} = \{y_{i,j}^1, y_{i,j}^2, y_{i,j}^3\}$, and $Z_{i,j} = \{z_{i,j}^1, z_{i,j}^2, z_{i,j}^3\}$. 
Next, for each edge $e=uv \in E_{i,j}$, we construct a subgraph $H_{i,j}^e$ that will be added to $H_{i,j}$.
In $H_{i,j}^e$, we begin by adding a cycle on three vertices $w_e^1, w_e^2, w_e^3$ and a neighbor~$w_e^4$ to $w_e^1$. 
Then, for $e = uv$, we add four vertices $a_e^u$, $b_e^u$, $a_e^v$, and $b_e^v$, and edges $a_e^ub_e^u$ and $a_e^vb_e^v$ in~$H_{i,j}^e$. 
Further, we connect $a_e^u$ and $w_e^3$ by an edge and connect $a_e^v$ and $w_e^2$ by an edge.  
Finally, we add the following edges to connect $H_{i,j}^e$ with the rest of $H_{i,j}$: $ b_e^vy_{i,j}^3, b_e^uz_{i,j}^3$, and $w_e^1x_{i,j}^3$. 
We define $S(H_{i,j}) = E(X_{i,j}) \cup E(Y_{i,j}) \cup E(Z_{i,j}) \cup \{\{a_e^ub_e^u,a_e^vb_e^v\} \cup \{w_e^1w_e^4, w_e^2w_e^3\} \mid e=uv \in E_{i,j}\}$ as the initial configuration for $H_{i,j}$.
An illustration of an edge-block is given in Figure~\ref{fig:mdis:whardblocks} (left).\\

\noindent {\bf Vertex-block.} For each $i \in [\kappa]$, we construct the vertex-block $H_i$ as follows. 
We add a path of three vertices $P_i = \{p_i^1, p_i^2, p_i^3\}$ and a vertex $q_i$. 
For each vertex $u \in V_i$, and for each $j \in [\kappa] \setminus \{i\}$, we add a $\Ttt$-gadget $\Ttt_u^j$. 
For each $j \in [\kappa] \setminus \{i\}$ with $j < \kappa$, we add an edge $c_4(\Ttt_u^j)c_1(\Ttt_u^{j+1})$ (or $c_4(\Ttt_u^j)c_1(\Ttt_u^{j+2})$ if $i = j+1$). 
Finally, we add two edges $p_i^3c_1(\Ttt_u^1)$ (or $p_i^3c_1(\Ttt_u^2)$ if $i =1$) and $c_4(\Ttt_u^{\kappa})q_i$ (or $c_4(\Ttt_u^{\kappa-1})q_i$
if $i = \kappa$). 
We define $S(H_{i}) = \{E(P_i) \cup \{S(\Ttt_u^j) \mid u \in V_{i} \wedge j \in [\kappa] \setminus \{i\}\}\}$ as the initial configuration for~$H_{i}$.
An illustration of a vertex-block is given in Figure~\ref{fig:mdis:whardblocks} (right).\\

\noindent{\bf Chain $\Ctt$.} We connect the edge-blocks and vertex-blocks using a structure called a {\em chain}. 
For an edge $e \in E(G)$ and a vertex $u$ incident with $e$, Let $\Ctt_e^u$ be a chain corresponding to the pair $e$ and~$u$. 
Let $\ell = 8\kappa^2$. 
In $\Ctt_e^u$, we add a path of $\ell+1$ vertices $R = \{r_i \mid i \in [\ell+1]\}$. 
For each $i \in [\ell]$, we add a neighbor $t_i$ to $r_i$. 
We call $r_1$ as $d_e^u$ and $r_{\ell+1}$ as $f_e^u$. 
Finally, we define $S(\Ctt_e^u) = \{r_it_i \mid i \in [\ell]\}$.\\

\noindent{\bf Graph $H$.} We have $\kctwo$ edge-blocks and $\kappa$ vertex-blocks in $H$. 
For each $1 \leq i < j \leq \kappa$, for each edge $e\in E_{i,j}$, and for each vertex $u$ incident with $e$ with $u \in V_i$ (assume the other vertex incident with~$e$ is $v \in V_j$), we add a chain $\Ctt_e^u$.
We add the following edges to connect the chain with blocks: $b_e^ud_e^u$ and $f_e^u c_5(\Ttt_u^j)$. 
This completes the construction of the graph. \\

\noindent{\bf Initial configuration and budget.} We define the initial configuration $S$ as follows:

\begin{eqnarray*}
S &=& \bigcup_{1 \leq i < j \leq \kappa} S(H_{i.j}) \cup \bigcup_{i \in [\kappa]} S(H_i) \cup \bigcup_{e=uv \in E(G)} \big(S(\Ctt_e^u) \cup S(\Ctt_e^v)\big)
\end{eqnarray*}
and we set budget $b = (2\ell+12)\kctwo + 2\kappa^2-\kappa$. \\

\begin{lemma}
\label{lem:mdis:whard:forward}
If $(G,\kappa)$ is a yes-instance of the \mcc~problem, then $(H, S, b)$ is a yes-instance of the \mdis~problem in the sliding model.   
\end{lemma}
\begin{proof}
Let $K = \{u_1,\ldots,u_{\kappa}\} \subseteq V(G)$ be a $\kappa$-clique in $G$ such that for each $i \in [\kappa]$, $u_i \in V_i$. 
For $1 \leq i < j \leq \kappa$, let $e_{i,j} = u_iu_j$.
Without loss of generality, we assume that the path $Z_{i,j}$ is connected with the vertex $b_{e_{i,j}}^{u_i}$ and the path $Y_{i,j}$ is connected with the vertex $b_{e_{i,j}}^{u_j}$. 
We do the following sequence of token slides:
\begin{enumerate}
    \item slide the token on $x_{i,j}^2x_{i,j}^3$ to $w_{e}^3a_{e_{i,j}}^{u_i}$ in three steps,
    \item slide the token on $a_{e_{i,j}}^{u_i}b_{e_{i,j}}^{u_i}$ to $f_{e_{i,j}}^{u_i}c_5(\Ttt_{u_i}^j)$ via $\Ctt_{e_{i,j}}^{u_i}$ in $\ell+2$ steps,
    \item slide the token on $c_5(\Ttt_{u_i}^j)c_3(\Ttt_{u_i}^j)$ to $c_3(\Ttt_{u_i}^j)c_4(\Ttt_{u_i}^j)$ in one step,
    \item slide the token on $w_{e}^2w_{e}^3$ to $w_{e}^2a_{e_{i,j}}^{u_j}$ in one step,
    \item slide the token on $a_{e_{i,j}}^{u_j}b_{e_{i,j}}^{u_j}$ to $f_{e_{i,j}}^{u_j}c_5(\Ttt_{u_j}^i)$ via $\Ctt_{e_{i,j}}^{u_j}$ in $\ell+2$ steps,
    \item slide the token on $c_5(\Ttt_{u_j}^i)c_3(\Ttt_{u_j}^i)$ to $c_3(\Ttt_{u_j}^i)c_4(\Ttt_{u_j}^i)$ in one step,
    \item slide the token on $z_{i,j}^2z_{i,j}^3$ to $z_{i,j}^3b_{e_{i,j}}^{u_i}$ in one step, and
    \item slide the token on $y_{i,j}^2y_{i,j}^3$ to $y_{i,j}^3b_{e_{i,j}}^{u_j}$ in one step.
\end{enumerate}

\begin{figure}[ht]
    \centering
    \begin{tikzpicture}[scale=0.85, every node/.style={transform shape}]
        \def\x{9.5}
        \def\y{4}
        \def\z{2.5}
        \def\w{-1}

        \draw[black, thick] (4,1) -- (4.75,1) -- (5.5,1);
        \draw[black, thick, dashed] (5.5,1) -- (6.25,1) -- (7.25, 1.5) -- (8, 1.5);
        \draw[black, thick] (7, 1) -- (6.25, 1);
        \draw[black, thick] (7.25, 1.5) -- (7.25, 0.5);
        \draw[black, thick, dashed] (6.25, 1) -- (7.25, 0.5) -- (8, 0.5);
        \draw[black, thick] (8,1.5) -- (8.75,1.5);
        \draw[black, thick] (8,0.5) -- (8.75,0.5);
        \draw[black, thick, dashed] (8.75, 1.5) -- (9.5, 1.5);
        \draw[black, thick, decorate,decoration={zigzag,segment length = 2mm, amplitude = 1mm}] (9.5, 1.5) -- (10.75, 1.5);
        \draw[black, thick, dashed] (8.75, 0.5) -- (9.5, 0.5);
        \draw[black, thick, decorate,decoration={zigzag,segment length = 2mm, amplitude = 1mm}] (9.5, 0.5) -- (10.75, 0.5);
        \draw[black, dashed] (5.5, 2.5) to[bend left=25] (8.75, 1.5);
        \draw[black, dashed] (5.5, -0.5) to[bend right=25] (8.75, 0.5);
        \draw[black, thick, dashed] (10.75, 1.5) to[bend left=60] (18.125,\z+0.75);
        \draw[black, thick, dashed] (10.75, 0.5) to[bend left= 10] (18.125,\w+0.75);
        \draw[black, fill = white] (4, 1) circle (0.1cm); 
        \draw[black, fill = white] (4.75, 1) circle (0.1cm);
        \draw[black, fill = white] (5.5, 1) circle (0.1cm);
        \draw[black, fill = white] (6.25, 1) circle (0.1cm);
        \draw[black, fill = white] (7.25, 1.5) circle (0.1cm); 
        \draw[black, fill = white] (8, 1.5) circle (0.1cm);
        \draw[black, fill = white] (6.9, 1) circle (0.1cm);
        \draw[black, fill = white] (7.25, 0.5) circle (0.1cm);
        \draw[black, fill = white] (8, 0.5) circle (0.1cm);
        \draw[black, fill = white] (8.75, 1.5) circle (0.1cm);
        \draw[black, fill = white] (8.75, 0.5) circle (0.1cm);
        \draw[black, fill = white] (9.5, 1.5) circle (0.1cm);
        \draw[black, fill = white] (9.5, 0.5) circle (0.1cm);
        \draw[black, fill = white] (10.75, 1.5) circle (0.1cm);
        \draw[black, fill = white] (10.75, 0.5) circle (0.1cm);
        \draw[black, thick] (4,2.5) -- (4.75,2.5) -- (5.5,2.5);
        \draw[black, thick] (4,-0.5) -- (4.75,-0.5) -- (5.5,-0.5);
        \draw[black, fill = white] (4, 2.5) circle (0.1cm);
        \draw[black, fill = white] (4.75, 2.5) circle (0.1cm); 
        \draw[black, fill = white] (5.5, 2.5) circle (0.1cm); 
        \draw[black, fill = white] (4, -0.5) circle (0.1cm); 
        \draw[black, fill = white] (4.75, -0.5) circle (0.1cm); 
        \draw[black, fill = white] (5.5, -0.5) circle (0.1cm); 
        \draw[red, thick, ->] (5.125, 1.125) -- (5.125, 2) -- (7.675, 2) -- (7.675, 1.675) node[green!50!black, midway, left=1.25cm] {$1$};
        \draw[red, thick, ->] (8.375, 1.375) -- (8.375, 1.125) -- (11.25, 1.125) -- (11.25, 2.25) node[green!50!black, midway, right] {$2$};
        \draw[red, thick, ->] (18.25, \z+0.5) -- (18.65, \z+0.5) -- (18.65, \z+0.1) node[green!50!black, midway, right] {$3$};
        \draw[red, thick, ->] (7.365, 1) -- (7.75, 1) -- (7.75, 0.675) node[green!50!black, midway, right] {$4$};
        \draw[red, thick, ->] (8.375, 0.675) -- (8.375, 0.9) -- (11.25, 0.9) -- (11.25, 0.675) node[green!50!black, midway, right] {$5$};
        \draw[red, thick, ->] (18.25, \w+0.5) -- (18.65, \w+0.5) -- (18.65, \w+0.1) node[green!50!black, midway, right] {$6$};
        \draw[red, thick, ->] (5.125, 2.625) -- (5.125, 3) -- (6.675, 3) -- (6.675, 2.625) node[green!50!black, midway, right] {$7$};
        \draw[red, thick, ->] (5.125, -0.375) -- (5.125, 0) -- (6.675, 0) -- (6.675, -0.375) node[green!50!black, midway, right] {$8$};

        \draw (12,\z) -- (12.75,\z) to[xsolid] (13.5,\z) to[xdashed] (14,\z); 
        \draw (16.625, \z) to[xsolid] (17.375, \z) -- (18.125,\z) to[xdashed] (18.875,\z);
        \draw  (21.5,\z) to[xdashed] (22,\z);
        \draw[black, thick] (18.125,\z) -- (18.125,\z+0.75);
        \draw[gray, thin] (14, \z+1) rectangle (21.5, \z-0.75);
        \draw[gray, thin] (16.125,\z+1) -- (16.125, \z-0.75);
        \draw[gray, thin] (19.375,\z+1) -- (19.375, \z-0.75);
        \node [red, font=\Huge] at (15.0625,\z) {$\dots$};
        \node [red, font=\Huge] at (20.4375,\z) {$\dots$};
        \draw[black, thick, fill=white] (12,\z) circle (0.1cm) node[blue, below] {$p_i^1$};
        \draw[black, thick, fill=white] (12.75,\z) circle (0.1cm) node[blue, below] {$p_i^2$};
        \draw[black, thick, fill=white] (13.5,\z) circle (0.1cm) node[blue, below] {$p_i^3$};
        \draw[black, thick, fill=white] (16.625, \z) circle (0.1cm);
        \draw[black, thick, fill=white] (17.375, \z) circle (0.1cm) node[blue, above=0.2cm] {$\Ttt_{u}^j$};
        \draw[black, thick, fill=white] (18.125,\z) circle (0.1cm);
        \draw[black, thick, fill=white] (18.875,\z) circle (0.1cm);
        \draw[black, thick, fill=white] (18.125,\z+0.75) circle (0.1cm);
        \draw[black, thick, fill=white] (22,\z) circle (0.1cm) node[blue, right] {$q_i$};
        
        \draw (12,\w) -- (12.75,\w) to[xsolid] (13.5,\w) to[xdashed] (14,\w); 
        \draw (16.625, \w) to[xsolid] (17.375, \w) -- (18.125,\w) to[xdashed] (18.875,\w);
        \draw  (21.5,\w) to[xdashed] (22,\w);
        \draw[black, thick] (18.125,\w) -- (18.125,\w+0.75);
        \draw[gray, thin] (14, \w+1) rectangle (21.5, \w-0.75);
        \draw[gray, thin] (16.125,\w+1) -- (16.125, \w-0.75);
        \draw[gray, thin] (19.375,\w+1) -- (19.375, \w-0.75);
        \node [red, font=\Huge] at (15.0625,\w) {$\dots$};
        \node [red, font=\Huge] at (20.4375,\w) {$\dots$};
        \draw[black, thick, fill=white] (12,\w) circle (0.1cm) node[blue, below] {$p_j^1$};
        \draw[black, thick, fill=white] (12.75,\w) circle (0.1cm) node[blue, below] {$p_j^2$};
        \draw[black, thick, fill=white] (13.5,\w) circle (0.1cm) node[blue, below] {$p_j^3$};
        \draw[black, thick, fill=white] (16.625, \w) circle (0.1cm);
        \draw[black, thick, fill=white] (17.375, \w) circle (0.1cm) node[blue, above=0.2cm] {$\Ttt_{v}^i$};
        \draw[black, thick, fill=white] (18.125,\w) circle (0.1cm);
        \draw[black, thick, fill=white] (18.875,\w) circle (0.1cm);
        \draw[black, thick, fill=white] (18.125,\w+0.75) circle (0.1cm);
        \draw[black, thick, fill=white] (22,\w) circle (0.1cm) node[blue, right] {$q_j$};
    \end{tikzpicture}
    \caption{An illustration of token movements in a feasible solution. The zigzag-lines denote the chains. }
    \label{fig:mdis:whardtokenmove}
\end{figure}
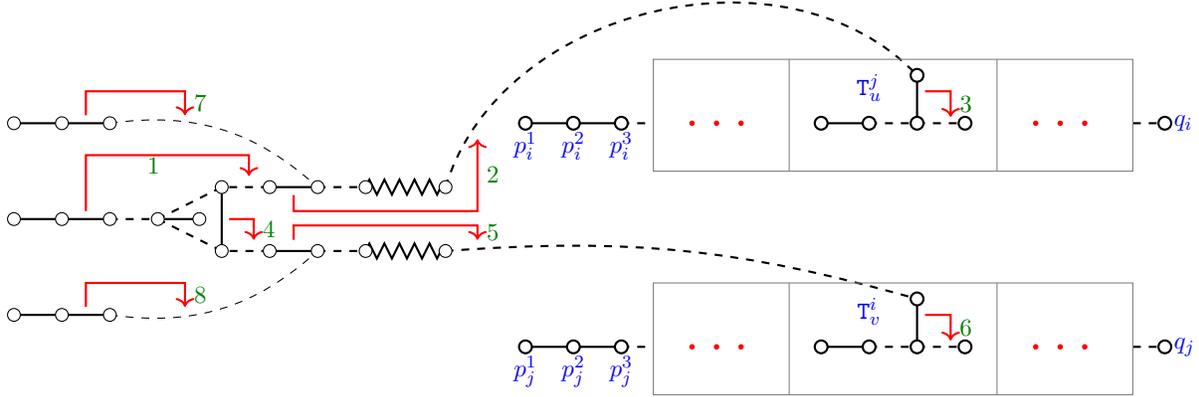
The above steps are illustrated in Figure~\ref{fig:mdis:whardtokenmove}. 
At the end of the above procedure, we spent $(2\ell+12)\kctwo$ token slides and we still have to move the tokens on the vertex-blocks. 
For each $i \in [\kappa]$, observe that the token movements due to the edge-blocks will create a $4(\kappa-1) + 1$-length augmenting path  between $p_i^3$ and $q_i$ and going through the gadgets of $u_i$. 
Therefore, we can slide every alternate edge to move the token from $p_i^2p_i^3$ to $c_4(T_{u_i}^{\kappa})q_i$ (or $c_4(T_{u_i}^{\kappa-1})q_i$ if $i = \kappa$).
Therefore, we can reach a matching starting from $S$ after at most $(2\ell+12)\kctwo + \kappa(2(\kappa-1)+1) = (2\ell+12)\kctwo + \kappa(2\kappa-1) = (2\ell+12)\kctwo + 2\kappa^2-\kappa = b$ token slides. 
\end{proof}

Let $M \subseteq E(H)$ be a feasible solution to the \mdis~instance $(H, S, b)$. We state some useful properties of $M$ that help us prove  Lemma~\ref{lem:mdis;whard:reverse}. 

\pagebreak
\begin{claim}
For every chain $\Ctt$ in $H$, $S(\Ctt) \subseteq M$.
    \label{clm:mdis:whard:chains}
\end{claim}

The above claim ensures that the chains act as the connectors between blocks and do not play any role in the solutions. 
That is, the matched edges in a chain are the same in both initial configuration and in any feasible solution. 
The length of the chain restricts the number of tokens moved across the chains. 
\begin{claim}
    \label{clm:mdis:whard:tokenlimit}
    The number of tokens moved between the blocks through the chains is at most $2\kctwo$. 
\end{claim}
\begin{proof}
    If a token is moved from a block to another block through a chain $\Ctt$, it has to slide at least $\ell+2$ steps since the chain $\Ctt$ has $\ell$ matching edges which can be treated as fixed according to Claim~\ref{clm:mdis:whard:chains}. 
    By contradiction, if we move $2\kctwo + 1$ tokens between the blocks through the chains, then it takes at least $(2\kctwo+1)(\ell+2) > (2\ell+2)\kctwo + 2\kappa^2-\kappa = b$ since $\ell = 8\kappa^2$. 
\end{proof}
We continue by providing an exact characterization on how many tokens are moved from the edge-blocks across the chains.
\begin{claim}
\label{clm:mdis:whard:edgeperblock}
    For each $1 \leq i < j \leq \kappa$, exactly two tokens must be moved from $H_{i,j}$ where one token moves to $H_i$ and the other one to $H_j$ through the chains. 
\end{claim}
\begin{proof}
    The number of vertices in the edge-block $H_{i,j}$ is $8|E_{i,j}|+9$ and the number of tokens in the block is $4|E_{i,j}|+6$. 
    Since $H_{i,j}$ can accommodate a matching of size at most $4|E_{i,j}|+4$, at least two tokens must be moved out of the edge-block. 
    A token on an edge-block can move to a vertex-block by crossing exactly one chain, but a token on an edge-block can move to another edge-block by crossing at least two chains. 
    Since there are $\kctwo$ edge-blocks and at most $2\kctwo$ tokens can move through the chains, there are exactly two tokens moved from each edge-block. 
\end{proof}
    Next, we count the minimum number of token slides required in each edge block. 
\begin{claim}
\label{clm:mdis:whard:edgeblocksix}
    For each $1 \leq i < j \leq \kappa$, the matching $M \cap E(H_{i,j})$ can be obtained from $S \cap E(H_{i,j})$ using at least six token slides. 
\end{claim}
\begin{proof}
    It is clear from the above claim that the edge-block $H_{i,j}$ must accommodate a matching of size $4|E_{i,j}|+4$. 
    This implies the following conditions.
    \begin{itemize}
        \item There exists an edge $e=uv \in E_{i,j}$ with $u \in V_i, v \in V_j$ such that the edges $w_e^1w_e^4, w_e^2a_e^v, w_e^3a_e^u$ are matched and the edges $w_e^2w_e^3, a_e^ub_e^u, a_e^vb_e^v$ are unmatched in $H_{i,j}^e$, and
        \item for every other edge $e' \not= e \in E_{i,j}$, the matched edges in $H_{i,j}^{e'}$ is the same as in the initial configuration. 
    \end{itemize}
    Otherwise, we cannot accommodate a matching of size $4|E_{i,j}|+4$. 
    We call the edge $e$ as {\em compatible to $M$} in $H_{i,j}$. 
    Let $M' = M \cap E(H_{i,j})$ and $e =uv$ with $u \in V_i, v \in V_j$ be the edge compatible to $M$ in $H_{i,j}$.
    It is clear that $(S \cap E(H_{i,j})) \setminus M' = \{x_{i,j}^2x_{i,j}^3, y_{i,j}^2y_{i,j}^3, z_{i,j}^2z_{i,j}^3, w_e^2w_e^3, a_e^ub_e^u, a_e^vb_e^v\}$ and $M' \setminus (S \cap E(H_{i,j})) = \{z_{i,j}^3b_e^u, y_{i,j}^3b_e^v, w_e^3a_e^u, w_e^2a_e^v\}$. 
    We have shown that exactly two tokens must move out from $H_{i,j}$. 
    Moreover, the edges $a_e^ub_e^u$ and $a_e^vb_e^v$ are not in $M'$. 
    Also, these edge are closest to the chains, so the tokens on these edges will slide through the chains. 
    Therefore, the tokens on the edges $a_e^ub_e^u$ and $a_e^vb_e^v$ are moved to some free vertices on the vertex-blocks $H_i$ and $H_j$, respectively. 
    The best possible way to move the tokens from the initial solution $S$ restricted to $H_{i,j}$ to $M'$ requires at least six token steps. 
    The token sliding is illustrated in steps 1, 4, 7, and 8 in Figure~\ref{fig:mdis:whardtokenmove}. 
\end{proof}
Now we show the characterization on how tokens are moved within the vertex-blocks. 
\begin{claim}
\label{clm:mdis:whard:vertexblocklb}
    For each $i \in [\kappa]$, a token on $P_i$ must move at least $2\kappa-1$ sliding steps. 
\end{claim}
\begin{proof}
    We have shown that the excess tokens in the edge-blocks should move to the vertex blocks through the chains. 
    It is clear that no tokens can slide from a vertex-block to a edge-block through chains. 
    Otherwise, the token has to cross the chains to reach other blocks, which falsifies Claim~\ref{clm:mdis:whard:edgeperblock}. 
    Therefore, a token on a vertex block can settle within the vertex-block. 
    In the path $P_i$, both edges have tokens. 
    At least one token on $P_i$ must move to another edge within the vertex-block. 
    Moreover, it must be an edge incident on $q_i$. 
    The minimum distance between an edge incident on~$q_i$ and an edge in $P_i$ is at least $4\kappa-3$. 
    At the best case, if there exists an augmenting path between these edges, then one can slide the token in $2\kappa-1$ steps. 
\end{proof}
\begin{lemma}    
\label{lem:mdis;whard:reverse}
If $(H, S, b)$ is a yes-instance of the \mdis~problem in the sliding model then $(G,\kappa)$ is a yes-instance of the \mcc~problem.  
\end{lemma}
\begin{proof}
Let $M \subseteq E(H)$ be a feasible solution for the instance $(H, S, b)$ of the \mdis~problem in the sliding model, where $b = (2\ell+12)\kctwo + 2\kappa^2-\kappa$. 
For each $1 \leq i < j \leq \kappa$, let $e_{i,j}=u_iu_j \in E_{i,j}$ be the compatible edge to $M$ in $H_{i,j}$. 
The minimal number of steps used in any edge-block is at least six (by Claim~\ref{clm:mdis:whard:edgeblocksix}), and both tokens moving through the chain need $\ell+2$ steps for traversing the chain. Thus, we must spend $2(\ell+2)+6$ sliding-steps per edge-block to obtain~$M$ from $S$. 
The minimum number of token slides required for any vertex-block is at least $2\kappa-1$ (by Claim~\ref{clm:mdis:whard:vertexblocklb}). 
Therefore, at least $2\kappa^2-\kappa$ slides are required for vertex-blocks. 
Thus, the remaining budget is at most $2\kctwo$. 

Now consider a token that is pushed ``out'' from an edge-block into a chain and all the way to a vertex-block; the token cannot remain on the chain. For instance, the token on $a_{e_{i,j}}^{u_i}b_{e_{i,j}}^{u_i}$ is moved to $f_{e_{i,j}}^{u_i}c_5(\Ttt_{u_i}^j)$, and the token on $c_5(\Ttt_{u_i}^j)c_3(\Ttt_{u_i}^j)$ has to move further as it cannot be in~$M$. 
Otherwise the token on $f_{e_{i,j}}^{u_i}c_5(\Ttt_{u_i}^j)$ has to move at least three steps to find an unmatched edge. 
This contradicts the remaining budget.
Therefore, the token on $c_5(\Ttt_{u_i}^j)c_3(\Ttt_{u_i}^j)$ must move to $c_3(\Ttt_{u_i}^j)c_4(\Ttt_{u_i}^j)$ in one step. 
Similarly in each vertex-block, $\kappa-1$ steps (one for each gadget-$\Ttt$) must be used to ``fix'' the tokens received from the edge-blocks. 
Therefore, at least $2\kctwo$ slides are required to handle the tokens that move from edge-blocks to vertex-blocks. 

For each $i \in [\kappa]$, at most one token on $P_i$ can be in $M$. 
The only possibility to move any token on $P_i$ is to match an edge incident with $q_i$. 
Otherwise, the token has to cross the chains to reach other blocks, which falsifies Claim~\ref{clm:mdis:whard:edgeperblock}. 
The minimum number of steps required to move a token from~$P_i$ to some edge incident with $q_i$ is $2\kappa-1$. 
Moreover, this is possible if and only if there exists an augmenting path of length $4\kappa-3$. 
Therefore, $M$ has an alternating path of length $4\kappa-3$ that starts from $p_i^1p_i^2$ to $c_4(\Ttt_{u}^{\kappa})q_i$ for some $u \in V_i$. 
This alternating path can be formed by the extra tokens pushed from the edge-blocks. 
That is, for each $j \not= i \in [\kappa]$, one among the two tokens pushed from~$H_{i,j}$ should reach the gadget $\Ttt_u^j$. 
Similarly for each $i \in [\kappa]$, there exists a vertex $u_i \in V_i$ such that each gadget $\Ttt_{u_i}^j$ for $j \not=i \in [\kappa]$ must receive a token from the respective edge-blocks. 
Finally, $2\kctwo$ tokens can reach $(\kappa-1)$ gadgets of each $\kappa$ vertices if there exists a clique of size $\kappa$ in $G$. 
\end{proof}

\subsubsection{\FPT for parameter \texorpdfstring{$k$}{k}}
We next show that the \mdis~problem is fixed-parameter tractable with respect to parameter $k + b$. Then, we show that for any instance of the problem one can bound~$\budget$ by a (quadratic) function of $k$ which implies fixed-parameter tractability of the problem when parameterized by $k$ alone. We start with some relevant definitions and lemmas.

\begin{definition}
For any set of edges $M \subseteq E(G)$ in a graph $G = (V, E)$, a vertex $v \in V(G)$ is \emph{overloaded} if $v$ is incident to at least two edges in $M$, $v$ is \emph{saturated} if it is incident to a single edge of $M$, and $v$ is \emph{unsaturated} otherwise. 
\end{definition}

Note that if $M$ is a matching, then there can be only saturated and unsaturated vertices.  

\begin{definition}
For a token $t$ on edge $e = \{u,v\} \in S$ and a vertex $w \in V(G)$, let $d(t,w)$ denote the minimum number of slides the token has to slide in order to become incident to  vertex $w$. 
For an integer $i \in \{1, \ldots, b\}$ and a token $t \in S$, let $Y^{i}_{t}$ denote the set $\{w \in V(G) \mid d(t,w) = i - 1\}$ and let~$Z^{i}_{t}$ denote the set $\{w \in V(G) \mid d(t,w) = i\}$. 
Let $G^i_t$ denote the graph $G[Y^{i}_{t} \cup Z^{i}_{t}] - E(G[Z^{i}_{t}])$.
\end{definition}

\begin{lemma}\label{lem:VertexCoverNeighbors}
Let $(G, S, b)$ be a yes-instance of the \mdis~problem in the sliding model. 
Let $b' \le b$ be the number of slides used by some token $t \in S$ to reach an edge $e = uv$. 
Note that both $u$ and $v$ as well as the edge $e$ are contained in $G^{b'}_t$. 
If either $u$ or $v$ has at least $2k + 1$ neighbors in $G^{b'}_t$, then there exists at least one neighbor, say $w$, such that sliding $t$ to $e' = uw$ or $wv$ (and leaving the final positions of all other tokens unchanged) produces another feasible solution. 
\end{lemma}

\begin{proof}
Assume, without loss of generality, that $u \in Y^{b'}_{t}$. 
Given that in the original solution $t$ is incident to $u$, we know that $t$ will also be incident to a neighbor $v$ of $u$ in $G^{b'}_t$. 
We also know that a single token can saturate at most $2$ vertices. 
Thus, $k - 1$ tokens can saturate at most $2k - 2$ vertices. 
This implies that if all tokens in $S \setminus \{t\}$ end up incident to vertices in $N_G(u) \cap (Y^{b'}_{t} \cup Z^{b'}_{t})$, they can saturate at most $2k - 2$ of those vertices. 

If $u$ has more than $2k - 2$ neighbors, there must exist one unsaturated vertex $w$ such that $e' = uw$ is a safe destination for the token $t$ while all tokens in $S \setminus \{t\}$ move to their original destinations. 
Notice that since $u \in Y^{b'}_{t}$, the distance from $t$ to $e$ is equal to the distance from $t$ to~$e'$, as needed. 
Similar arguments hold if $v$ has more than $2k - 2$ neighbors as all of those neighbors must belong to~$Y^{b'}_t$, which completes the proof (recall that we delete edges between vertices in $Z^{b'}_{t}$ when constructing $G^{b'}_t$). 
\end{proof}

\begin{lemma}\label{lem:MatchingYZ}
Let $(G, S, b)$ be a yes-instance of the \mdis~problem. 
Let $b' \le b$ be the number of slides used by some token $t \in S$ to reach an edge $e = uv$. 
If a matching $M$ of size at least $2k - 1$ exists in $G^{b'}_t$, then there exists an edge $e' \neq e \in M$ such that sliding $t$ to $e'$ (and leaving the final positions of all other tokens unchanged) produces another feasible solution. 
\end{lemma}

\begin{proof}
Given that in the original solution, $t$ slides $b'$ steps, we know that $t$ ends up on one of the edges in $G^{b'}_t$. 
In fact, that edge either has both endpoints in $Y^{b'}_t$ or one endpoint in $Y^{b'}_t$ and the other in $Z^{b'}_t$. 
We also know that a single token can saturate at most $2$ vertices, thus $k - 1$ tokens can saturate at most $2k - 2$ vertices. 
This implies that if all tokens in $S$ besides the token $t$ end up incident to vertices in $G^i_t$, they can saturate at most $2k - 2$ vertices in that subgraph. 
The latter vertices can appear as both endpoints of a matching edge in $M$ or as one of the endpoints. 
In either case, there must exist at least one edge $e$ in $M$ that is not incident to any saturated vertices after all tokens in $S \setminus \{t\}$ move to their destinations. 
Since $M$ is a matching, it is safe to replace $e$ by $e'$ as the destination edge for the token $t$, as needed. 
\end{proof}

\begin{theorem}\label{thm:MatchingDiscoveryFPTkb}
The \mdis~problem  in the token sliding model parameterized by both the number of tokens $k$ and the budget $b$ is fixed-parameter tractable. 
\end{theorem}

\begin{proof}
Let $(G, S, b)$ be an instance of the \mdis~problem, where $|S| = k$.
The algorithm proceeds as follows.
We first guess a set $S' \subseteq S$ of tokens that will slide form their original positions. 
If such a guess leaves any overloaded vertices we can ignore it and proceed to the next guess.
Note that this \textit{guessing} procedure requires $\mathcal{O}(2^k)$ time in the worst case.
Next, for each $r \in \{0, 1, \ldots, b\}$ and for a fixed $S'$, we guess a partition of $r$ over the tokens in $S'$.
In other words, we try all possible ways of distributing the budget $r$ over the tokens in $S'$. 

Now, for a fixed subset $S'$, a fixed $r$, and a fixed distribution $D$ of $r$ over the tokens of $S'$, let $r_i$ be the budget allocated to a token $t_i \in S'$ under $D$. 
The algorithm computes the sets $Y^{r_i}_{t_i}$ and $Z^{r_i}_{t_i}$ for each $t_i \in S'$. 
The next step is to produce a set of candidate target edges for each $t_i$, which we denote by $T_i$. 
To do so, we run a maximum matching algorithm in the graph $G^{r_i}_{t_i} - (S \setminus S')$. 
There are two cases to consider: 

\begin{itemize}
\item If, for $t_i \in S'$, the algorithm finds a maximum matching $M$ in $G^{r_i}_{t_i} - (S \setminus S')$ such that $|M| \ge 2k - 1$, then we set $T_i$ to any $2k - 1$ edges of $M$. 
\item If $|M| < 2k - 1$, then $|V(M)| < 4k - 2$ and $G^{r_i}_{t_i} - (S \setminus S')$ has a vertex cover of size at most $4k - 2$. 
When $|E(G^{r_i}_{t_i} - (S \setminus S'))| \leq (4k - 2)(2k + 2)$, we simply set $T_i = E(G^{r_i}_{t_i} - (S \setminus S'))$. 
Otherwise, we let $T_i$ include $M$ as well as at most $2k + 1$ distinct edges incident to every vertex of the vertex cover.   
\end{itemize}

Once the candidate sets have been constructed, the algorithm exhaustively checks whether it can reach a valid matching in the graph $G$, and if so, terminates with a yes. 
The algorithm terminates with a no if no valid matching is reached during this search. 

It is not hard to see that this algorithm runs in \FPT-time. 
It loops $b + 1$ times and in each iteration, it chooses one of $2^k$ subsets of the $k$ tokens to move. 
It also distributes the budget $r \le b$ on this subset, thus far having $\binom{b + k - 1}{k - 1}$ choices in total per iteration. 
Per one such choice, the algorithm runs up to $k$ times the maximum matching procedure in $\Oof(\sqrt{n} \cdot m)$ time. 
Additionally, per one such choice, it encounters up to $(2k - 1)^2$ target edges per moving token. 
Checking whether the constructed solution is valid can be performed in time $\Oof(k^2)$. 
The whole algorithm requires $(b + 1) \cdot \binom{b + k - 1}{k - 1} \cdot k \cdot \Oof(\sqrt{n} \cdot m) + (b + 1) \cdot \binom{b + k - 1}{k - 1} \cdot (2k - 1)^2 \cdot \Oof(k^2)$ time, which is $\Oof\big((b + 1) \cdot \binom{b + k - 1}{k - 1} \cdot k \cdot \Oof(\sqrt{n} \cdot m)\big)$.

The correctness of the algorithm follows from Lemmas~\ref{lem:VertexCoverNeighbors} and~\ref{lem:MatchingYZ} and the fact that we exhaustively try candidate edges per token.
In particular, if a solution exists, the algorithm must find a subset $S' \subseteq S$ of tokens that move, and the values $r_{t_i}$ for each $t_i \in S'$ that represents the number of slides performed by each token $t_i$.
Given Lemma~\ref{lem:MatchingYZ}, we know that for a certain $t_i$, if the size of a maximum matching~$M$ is more than $2k - 2$, then one edge $e \in M$ is a safe destination for~$t_i$.
If for a certain token $t_i$ the maximum matching algorithm returns a matching $M$ of size less than~$2k - 1$, then, by the definition of a maximum matching, we know that all edges are incident to the endpoints of $M$.
There are at most $4k - 2$ endpoints which we know the token $t_i$ is incident to one of them after applying the solution $\mathcal{S}$.
Let that vertex be $u$.
By Lemma~\ref{lem:VertexCoverNeighbors}, we know there exists one vertex $v$ such that $e = uv$ is a safe destination for $t_i$.
\end{proof}

We now show that the same algorithm is also a fixed-parameter tractable algorithm for  parameter~$k$ alone by proving that $k$ upper-bounds the parameter $b$.
For that, we first recall the following optimality criterion.

\begin{theorem}[Berge~\cite{berge1957two}]\label{thm:Berge}
A matching $M$ in a graph $G$ is a maximum matching if and only if~$G$ has no $M$-augmenting path, i.e., a path that starts and ends on free (unsaturated) vertices, and alternates between edges in $M$ and edges not in $M$.  
\end{theorem}

\begin{lemma}\label{lem:MatchingDiscoverykUpperBound}
In any yes-instance $(G, S, b)$ of the \textsc{Matching Discovery} problem in the sliding model, $b$ can be upper bounded by $2(k^2 + k)$, where $k = |S|$.
\end{lemma}

\begin{proof}
We prove this lemma by induction on the number of tokens in an instance $(G, S, b)$ of \textsc{Matching Discovery} with $n = |V(G)|$ and $m = |E(G)|$.
Let $\mathcal{P}(i)$ be the proposition that if $(G, S, b)$ is a yes-instance of the \textsc{Matching Discovery} problem with $|S| \le i$, $|V(G)| = n$, and $|E(G)| = m$, then there exists a solution to $(G, S, b)$ using at most $2(i^2 + i)$ slides. 
The statement holds trivially for $i = 1$.
Hence, assume $\mathcal{P}(i)$ is true for an arbitrary integer $i > 1$. We show that this implies that $\mathcal{P}(i + 1)$ is also true.

Given a yes-instance $(G, S, b)$ with $|V(G)| = n$, $|E(G)| = m$ and $|S| = i + 1$, we form an instance $(G, S \setminus \{t\}, b)$ by deleting an arbitrary token $t$ from $S$.
$(G, S \setminus \{t\}, b)$ is clearly a yes-instance and there exists a solution using at most $2(i^2 + i)$ slides by the induction hypothesis.

Given that $(G, S, b)$ is also a yes-instance, we know that the matching in  $(G, S \setminus \{t\}, b)$ is not maximum and, by Theorem \ref{thm:Berge}, an augmenting path $P$ of  $r \le 2i + 1$ edges exists in $G$ after solving $(G, S \setminus \{t\}, b)$. 
Let $\sigma$ be a feasible solution to $(G, S \setminus \{t\}, b)$, i.e.,  $\sigma$ is a sequence of slides transforming~$S$ to a matching. 
After applying $\sigma$, token $t$ in $(G, S, b)$ is either on some edge of $P$ or not.
In the former case, at most $i + 1$ slides are needed in addition to the $2(i^2 + i)$ slides to get to a valid matching in $(G, S, b)$, thus $b \leq 2(i^2 + i) + i + 1$.
In the latter case, the token can be at most $2i + 1$ edges away from the path $P$ and in that case, at most $3i + 2$ other slides are needed to get to a valid matching in $(G, S, b)$, thus $b \leq 2(i^2 + i) + 3i + 2$.
In both cases, we form a solution $\sigma'$ to $(G, S, b)$ using at most $2\big((i+1)^2 + (i + 1)\big)$ slides. 
Since the choice of $n$ and $m$ was arbitrary, the statement holds for all yes-instances for any $n$ and $m$. 
\end{proof}

Combining Theorem \ref{thm:MatchingDiscoveryFPTkb} with Lemma \ref{lem:MatchingDiscoverykUpperBound} we get the following result.

\begin{corollary}\label{cor:MatchingDiscoveryFPTk}
The \mdis~problem in the token sliding model is fixed-parameter tractable with respect to parameter $k$. 
\end{corollary}

\section{Vertex/edge cuts}\label{sec:separator}

A \emph{vertex cut} in $G$ between $s$ and $t$ is a set of vertices $C \subseteq V(G)$ such that every $s$-$t$-path contains a vertex of $C$. 
Likewise, an \emph{edge cut} in a graph $G$ between two of its vertices $s$ and $t$ is a set of edges $C \subseteq E(G)$ such that every $s$-$t$-path contains an edge of $C$. 
In the \textsc{Vertex Cut} (resp.\ \textsc{Edge Cut}) problem we are given a graph $G$, vertices $s,t$ and an integer~$k$ and the goal is to compute a vertex cut (resp.\ edge cut) of size~$k$.

\subsection{Rainbow vertex/edge cut}

Given an edge-colored graph $(G, \phi)$ with vertices $s,t \in V(G)$, the \textsc{Rainbow Edge Cut} problem asks whether there exists an edge cut $C \subseteq E(G)$ separating~$s$ and $t$ such that all edges in $C$ have pairwise different colors. The \textsc{Rainbow Vertex Cut} problem can be defined analogously where instead of $C \subseteq E(G)$ we look for a subset of vertices $C \subseteq V(G)$ separating $s$ from $t$. 
The \textsc{Rainbow Edge Cut} problem is known to be \NP-complete~\cite[Theorem 5.5]{rainbowcuts}.
We start by showing that the problem remains \NP-complete on planar graphs.
\begin{theorem}
\label{thm:rainbowcut}
    The \textsc{Rainbow Edge Cut} problem is \NP-complete on planar graphs.
\end{theorem}

\begin{proof}
Containment in $\NP$ is clear as \textsc{Rainbow Edge Cut} on general graphs is in $\NP$. Hence we focus on the hardness proof.
We present a reduction from {\sc Rainbow Matching} on paths.
Let~$(P, \kappa)$ be an instance of {\sc Rainbow Matching} where $P$ is a path on $n$ vertices denoted by $v_1, \dots v_n$ and the edges are colored with colors from a color set $\Cmc$. 

We construct an instance $(G, \psi : E(G) \rightarrow \Cmc')$ of \textsc{Rainbow $s$-$t$-Cut} as follows. 
The new color set is $\Cmc' = \Cmc \cup \{\mathsf{black}\} \cup \{c_i \mid i \leq n-2\}$, that is, $\Cmc'$ uses the colors from $\Cmc$ as well as $n-2$ fresh colors and the color black. 

\begin{figure}[ht]
    \centering
    \begin{tikzpicture}

    \def\n{7}
    \def\k{3}
    
    \foreach \i in {1, ..., \n} {
      \foreach \j in {1, ..., \k} {
	    \fill (\i*1.5, -2*\j+1) circle (2pt) node (v_\j_\i) {} node[above] {\tiny $v_{\j,\i}$};
	   }
	}

    \fill (v_2_1) ++(-2,0) circle (3pt) node (s) {} node[left=1mm] {$s$};
    \fill (v_2_\n) ++( 2,0) circle (3pt) node (t) {} node[right=1mm] {$t$};

    \foreach \j in {1, ..., \k} {
      \foreach[count=\i] \c in {blue, red, blue, green, blue, red} {
        \draw[\c] (v_\j_\i) -- (v_\j_\the\numexpr\i+1);
      }
      
      \draw (s) -- (v_\j_1);
      \draw (v_\j_\n) -- (t);
    }
    
    \foreach \i in {1, ..., \numexpr \n - 1} {
      \foreach \j in {1,...,\numexpr \k - 1} {
       \fill (\i*1.5 + 0.75, -2*\j) circle (2pt) node (u_\j_\i) {} node[left] {\tiny $u_{\j,\i}$};
       \draw (v_\j_\i) -- (u_\j_\i);
      }
    }

    \foreach \j in {1, ..., \numexpr \k - 1} {
      \foreach[count=\i] \c in {orange, lime, cyan, purple, yellow} {
        \draw[\c] (u_\j_\i) -- (v_\j_\the\numexpr\i+1);
        \draw[\c] (u_\j_\i) -- (v_\the\numexpr\j+1_\the\numexpr\i+2);
      }
    }

    \foreach \j in {1, ..., \numexpr \k - 1} {
      \draw (s) -- (u_\the\numexpr\j_1);
      \draw (u_\the\numexpr\j_\the\numexpr\n-1) -- (t);
    }

    \draw[semithick] (s) .. controls (1.5, 2) and (1.5*\n, 2) .. (t);

    \draw[dashed, gray] (2,.5) .. controls (2.1, -1.2) and (2.6, -1.5) .. (3,-2.1) .. controls (4, -2.8) and (6.5, -1.7) .. (7.5,-4) .. controls (8, -4.5) and (10, -4.5) .. (10, -5.5);

    \end{tikzpicture}
    \caption{Illustration of the hardness reduction for \textsc{Rainbow $s$-$t$-Cut} on planar graphs.}
    \label{fig:rainbowcut}
\end{figure}
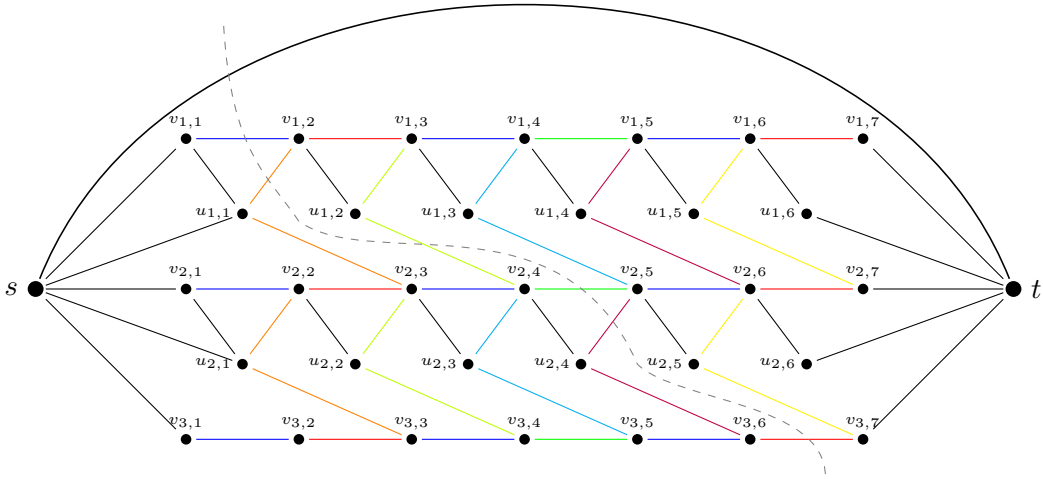

Let us describe the construction of $G$ in detail; see \Cref{fig:rainbowcut} for an illustration.
In the first step, $G$ consists of $\kappa$ disjoint copies of $P$, which we call $P_1, \dots, P_\kappa$. Let the vertices of $P_j$ be called $v_{j,1} \dots v_{j, n}$. Additionally, we add two fresh vertices $s$ and $t$.
For every $j < \kappa$, insert a set $L_j$ of $n-1$ fresh vertices, and call them $u_{j,1}, \dots, u_{j,n-1}$. Hence, $V(G) = \bigcup_{j \leq \kappa} V(P_j) \cup \bigcup_{j < \kappa} V(L_j) \cup \{s,t\}$.

Now we connect $s$ and $t$ with a black edge, which enforces that this black edge must be part of every (rainbow) $s$-$t$-cut. Hence, any other black edge may not be part of any rainbow $s$-$t$-cut of~$G$. 
We connect $s$ and $t$ to the vertices in $P_j$ and $L_j$ as follows. For every $j \leq \kappa$ we insert the edges $\{s, v_{j, 1}\}$ and $\{v_{j,n}, t\}$, which are colored black. 
Likewise, for every $j < \kappa$ we insert the edges $\{s, u_{j,1}\}$ and $\{u_{j,n-1}, t\}$, which are also colored black.
Finally, for every $j < \kappa$ and $i \leq n-2$, we insert the edge $\{v_{j,i}, u_{j,i}\}$ which is colored black and the edges $\{u_{j,i}, v_{j, i+1}\}$ and $\{u_{j,i}, v_{j+1, i+2}\}$, which are colored~$c_{i}$.
This finishes the construction. 

We claim that $P$ has a rainbow matching of size $\kappa$ if and only if~$G$ admits a rainbow $s$-$t$-cut.
Assume that $P$ has a rainbow matching $M = \{e_1, \dots, e_\kappa\}$ of size $\kappa$. Without loss of generality, we assume that the edges in $M$ are ordered with respect to the $v_i$, i.e., if $e_i = \{v_{\ell_i}, v_{\ell_{i+1}}\}$ and $i < j$, then $\ell_i < \ell_j$.
We claim that the set $C$ that consists of the black edge $\{s,t\}$, the copy of $e_i$ in $P_i$, and the obvious edges connecting the $P_i$ and $L_j$ is a rainbow $s$-$t$-cut of $G$. To be precise, we have
\[
C = \{s,t\} \cup \{\{v_{i, \ell_i}, v_{i, \ell_i+1}\}, \{u_{i, \ell_i}, v_{i, \ell_i+1}\} \mid i \leq \kappa\} \cup \bigcup_{i \leq \kappa} \{\{u_{i, j}, v_{i, j+2}\} \mid \ell_i < j \leq \ell_{i+1}-2\}.
\]
\vspace{-5pt}

Observe that $C$ is a cut by construction (see \Cref{fig:rainbowcut}), and that no two edges in $C$ have the same color, as $M$ is a rainbow matching, and for every $j < \kappa$ and $i \leq n-2$ at most one of the edges $\{v_{j,i}, u_{j_i}\}$ and $\{u_{j_i}, v_{j+1, i+2}\}$ is contained in $C$.

Now assume that $G$ admits a rainbow $s$-$t$-cut. As $s$ and $t$ are directly connected, every rainbow $s$-$t$-cut $C$ for $G$ must contain $\{s,t\}$. Furthermore, by construction $C$ contains exactly one edge from every $P_j$, say the edge $\{v_{j,\ell_i}, v_{j,\ell_{i+1}}\}$, as no other black edge is part of $C$.
We claim that $M = \{\{v_{\ell_i}, v_{\ell_{i+1}}\} \mid \{v_{j,\ell_i}, v_{j,\ell_{i+1}}\} \in C \text{ for some } j \leq \kappa\}$ is a rainbow matching of $P$. Obviously, $M$ is rainbow, as $C$ is rainbow. To show that $M$ is indeed a matching, observe that for all $\ell_i \neq \ell_j$ we have $|\ell_i - \ell_j| \geq 2$, that is, $M$ does not contain two (copies of) consecutive edges of $P$. To see this, assume for the sake of contradiction that there are $i \neq j$ such that $|\ell_i - \ell_j| \leq 1$. Let $j_1 < j_2$ be such that $\{v_{j_1, \ell_i}, v_{j_1, \ell_{i+1}}\}$ and $\{v_{j_2, \ell_j}, v_{j_2, \ell_{j+1}}\}$ are contained in $C$. By construction, $C$ must also contain $\{u_{j_1, \ell_i}, v_{j_1, \ell_{i+1}}\}$ and can hence not contain $\{u_{j_1, \ell_i}, v_{j_1+1, \ell_{i}+2}\}$, as they share the same color. This however implies that $\{v_{j_2, \ell_j}, v_{j_2, \ell_{i+1}}\}$ cannot be contained in~$C$, a contradiction. This finishes the proof.
\end{proof}

\vspace{-5pt}
Observe that we can reuse the ideas of the previous proof to show the hardness of the \textsc{Rainbow Vertex Cut} problem.
In fact, we can subdivide every edge (and color the subdivision vertex with the same color as the edge), and color every other vertex black.
Then all observations translate one-to-one and we obtain the following corollary.
\begin{corollary}
\label{cor:rainbowcut}
    The \textsc{Rainbow Vertex Cut} problem is \NP-complete on planar graphs.
\end{corollary}

\begin{theorem}\label{thm:rainbow-separator}
    The \textsc{Weighted Rainbow Vertex/Edge Cut} problem is fixed-parameter tractable when parameterized by~$k$.
\end{theorem}
\begin{proof}
    We present an \FPT algorithm for the \textsc{Weighted Rainbow Vertex Cut} problem. The tractability of \textsc{Weighted Rainbow Edge Cut} follows by considering the vertex cut version in the line graph of the input graph. 

    We follow the approach of~\cite{marx2013finding}. We assume that our input graphs are colored and vertex weighted, where the weights are non-negative integer weights not exceeding $n$. 
    We first compute a subgraph~$H$ of $G$ using the following \emph{treewidth reduction theorem} (Theorem 2.15 of~\cite{marx2013finding}). Given a graph $G$, $T\subseteq V(G)$ and a non-negative integer $k$. 
    Let $C$ be the set of all vertices of $G$ participating in a minimal $s$-$t$ separator of size at most $k$ for some $s,t\in T$. 
    For every fixed $k$ and~$|T|$, there is a 
    linear-time algorithm that computes a graph $H$ with the following properties: 
    \begin{enumerate}
        \item 
        $C\cup T \subseteq V(H)$,
        \item 
        For all $s,t\in T$, a set $K\subseteq V(H)$ with $|K|\leq k$ is a minimal $s$-$t$ separator of $H$ if and only if $K\subseteq C\cup T$ and $K$ is a minimal $s$-$t$ separator of $G$, 
        \item 
        The treewidth of $H$ is at most $h(k,|T|)$ for some function $h$.
        \item 
        $H[C\cup T]$ is isomorphic to $G[C\cup T]$. 
    \end{enumerate}
    
    We apply the theorem to $G$ with $T=\{s,t\}$ to obtain a subgraph $H$ of $G$. We inherit the colors and weights from $G$. 
    Since all minimal separators of size at most $k$ are preserved in $H$ it is sufficient to search for a minimum weight rainbow $s$-$t$ separator of size at most $k$ in $H$. 

    Since the graph $H$ has treewidth at most $h(k,2)$ we can apply the optimization version of Courcelle's Theorem for graphs of bounded treewidth. For this, observe that we can formulate the existence of a rainbow $s$-$t$ separator as a formula $\phi(X)$ such that $\phi(S)$ is true for a set $S$ of vertices in a colored graph $G$ if and only if $S$ is a rainbow separator. 
    We now apply Courcelle's Theorem in the optimization version presented in~\cite{arnborg1991easy,courcelle1993monadic}.
    We remark that the running time in the optimization version of the theorem is usually stated as being polynomial time for each fixed number of weight functions. 
    This could potentially imply an exponential running time in the given weights. 
    However, it is easily observed that this is not the case and the dependence on each weight function is in fact linear. 
    We refer to the discussion before Theorem 5.4 in~\cite{arnborg1991easy} to conclude in our case a running time of $f(k)\cdot n^2$ for some computable function $f$.
\end{proof}

\subsection{Vertex/edge cut discovery}

In the \textsc{Vertex Cut Discovery} problem  in the token sliding model we are given a graph $G$, vertices $s, t \in V(G)$, a starting configuration $S \subseteq V(G)$ of size $k$ and a non-negative integer $\budget$. The goal is to decide whether we can discover an $s$-$t$-separator in $G$ (starting from $S$) using at most $b$ token slides.
Similarly, in the \textsc{Edge Cut Discovery} problem, we are given a graph $G$, vertices $s, t \in V(G)$, a starting configuration $S \subseteq E(G)$ of size $k$ and a non-negative integer $\budget$. The goal is again to decide whether we can discover an $s$-$t$-cut in $G$ (starting from $S$) using at most $b$ token slides. 
We denote an instance of \textsc{Vertex Cut Discovery} resp.\ \textsc{Edge Cut Discovery} by a tuple $(G, s, t,S, b)$. We always use $k$ to denote the size of $S$.
It is easy to see that these problems are polynomial-time equivalent, hence all results for one of the problems immediately translates to the other problem as well.

We prove that \textsc{Vertex Cut Discovery}  in the token sliding model is \NP-hard, fixed-parameter tractable with respect to parameter $k$ and \W{1}-hard with respect to parameter $b$. We start by proving hardness by a reduction from the \textsc{Clique} problem. 
Recall that a clique in a graph is a set of pairwise adjacent vertices. In the \textsc{Clique} problem we are given a graph $G$ and integer parameter~$\kappa$ and the question is to decide whether there exists a clique of size at least $\kappa$ in $G$. It is well-known that the \textsc{Clique} problem is $\NP$-hard and its parameterized variant is \W{1}-hard with respect to the solution size~\cite{downey1995fixed}. 

\begin{theorem}
    \label{thm:separator:hardness}
    The \sdis~problem in the sliding model is \NP-hard and \W{1}-hard with respect to parameter~$\budget$ on 2-degenerate bipartite graphs. 
\end{theorem}
\begin{proof}
We show \NP-hardness by a reduction from the \clique~problem. The reduction is both a polynomial time reduction as well as a parameterized reduction, showing both claimed results. 
Let~$(G,\kappa)$ be an instance of the \textsc{Clique} problem. 
We may assume that $\kappa\geq 4$, hence, $\binom{\kappa}{2}>\kappa$. 

\begin{figure}[ht]
    \centering
    \begin{tikzpicture}
    \tikzstyle{r} = [rectangle, fill, minimum size = 4pt, inner sep = 2pt, outer sep = 1pt]

    \fill (0,0) circle (3pt) node (s) {} node[below=1mm] {$s$};

    \node[r] (z1) at ($(s) + (-1.5,2/3)$) {};
    \node[r] (zk) at ($(s) + (-1.5,-2/3)$) {};
    \node at (-1.5, 0.1) {$\vdots$};

    \draw (-1.5, 0) ellipse (0.25cm and 1.25cm) node[below=1.4cm] {$Z$};
    \draw (s) -- (z1);
    \draw (s) -- (zk);

    \fill (s) ++(2,0.8) circle (2pt) node (xu) {} node[above=1mm] {$x_u$};
    \fill (s) ++(2,-0.8) circle (2pt) node (xv) {} node[below=1mm] {$x_v$};
    \draw (2,0) ellipse (0.3cm and 1.75cm) node[below=1.9cm] {$X$};
    \draw (s) -- (1.97, 1.75);
    \draw (s) -- (1.97, -1.75);

    \node[r] (yuv) at ($(s) + (6,0)$) {} node[above = 0cm of yuv] {$y_{uv}$};
    \draw (yuv) ellipse (0.3cm and 1.75cm) node[below=1.9cm] {$Y$};
    \draw (xu) --
        node[pos=1/2, fill, circle, inner sep = 1.5pt, outer sep = 1.5pt] {} 
    (yuv);    
    \draw (xv) --
        node[pos=1/2, fill, circle, inner sep = 1.5pt, outer sep = 1.5pt] {} 
    (yuv);    

    \fill (s) ++(8,0) circle (3pt) node (t) {} node[below=1mm] {$t$};

    \draw (t) --
    ($(yuv) +(0.03, 1.75)$);
    \draw (t) --
    ($(yuv) +(0.03, -1.75)$);

    \fill (s) ++(4,4) circle (2pt) node (b1) {};
    \fill (s) ++(4,3) circle (2pt) node (b2) {};
    \draw (4,3.5) ellipse (0.2cm and 1cm) node[left = 2mm] {$P_i$};
    \node at (4,3.6) {$\vdots$};

    \path
    (s) edge[bend left=40] (b1)
    (s) edge[bend left=37] (b2)
    (b1) edge[bend left=40] (t)
    (b2) edge[bend left=37] (t)
    ;    

    \end{tikzpicture}
    \caption{An illustration of the hardness reduction for the \sdis~problem.}
    \label{fig:separator:hardness}
\end{figure}
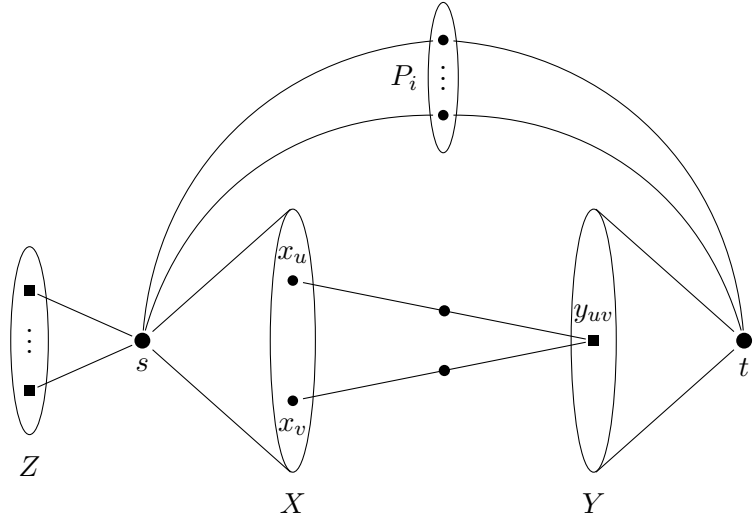 

We construct the following graph $H$ (see \Cref{fig:separator:hardness} for an illustration). 
We add two vertices $s$ and $t$ in $H$ where $s$ has $\kappa$ pendent vertices, which we collect in a set named~$Z$. 
For every vertex $u \in V(G)$, we add a vertex $x_u$ in $H$, which we connect with $s$, and for every $e \in E(G)$ we add a vertex $y_e$ in $H$, which we connect with $t$. Let $X$ denote the set of the $x_u$ and $Y$ denote the set of the $y_e$. 
For every vertex $u \in V(G)$ and for each edge $e \in E(G)$ incident with $u$, connect $x_u$ and $y_e$ via a 1-subdivided edge (that is, a path with one interval vertex).
Furthermore, we add $\binom{\kappa}{2}$ disjoint paths $P_1, \dots, P_{\binom{\kappa}{2}}$, each with a single internal vertex, connecting $s$ and $t$.
We denote the internal vertex of $P_i$ by $p_i$.
This completes the construction of $H$.

Observe that $H$ is 2-degenerate and bipartite. 
To see that $H$ is 2-degenerate observe that all subdivision vertices (the $p_i$ as well as the subdivision vertices connecting $X$ and $Y$) have a degree of~2. 
Their removal yields two stars with centers $s$ and $t$ which are 1-degenerate.
Bipartiteness follows from the subdivisions. We define the initial configuration $S$ as $Z \cup Y$ and $\budget = 2\kappa + 2\binom{\kappa}{2}$.

We claim that $G$ has a clique of size $\kappa$ if and only if $(H,s,t,S,\budget)$ is a yes-instance of the \textsc{Vertex Cut Discovery} problem.

Assume that $G$ has a clique $C$ of size $\kappa$. In order to separate $s$ and $t$, we need to cut all (of the $\binom{\kappa}{2}$-many) paths $P_i$ by moving a token on each $p_i$, and the paths using a vertex from $X$ and a vertex from~$Y$.
Using a budget of $2\kappa$, we move the $\kappa$ tokens on the vertices in $Z$ to the $\kappa$ vertices~$x_u$ for~$u \in C$ by sliding them over $s$. As $C$ is a clique, this frees $\binom{\kappa}{2}$ tokens in~$Y$, namely those $y_e$ where~$e$ is incident with a $x_u$ for $u \in C$ (while the remaining tokens in $Y$ block all other paths connecting $X$ and $Y$). 
We move the tokens on these $y_e$ to the subdivision vertices of the $P_i$ by sliding them over~$t$, again for a cost of 2 per token.
That is, using $\budget = 2\kappa + 2\binom{\kappa}{2}$ slides we discover the vertex cut $\{x_u \mid u \in C\} \cup \{y_{uv} \mid u, v \notin C\} \cup \{p_i \mid i \leq \binom{\kappa}{2}\}$ between $s$ and $t$, witnessing that $(H, s, t, S, \budget)$ is a yes-instance.

Now assume that $(H,s,t,S,\budget)$ is a yes-instance of the \textsc{Vertex Cut Discovery} problem. 
Let $C_0 \dots C_d$ with $S = C_0$ and $d \leq \budget$ be a configuration sequence in $G$ yielding a vertex cut in $G$ between~$s$ and $t$. 
With out loss of generality, we assume that $d$ is minimal, \ie there is no valid solution that can be discovered in less than $d$ steps.
In order to cut the $P_i$, the set $C_d$ must contain~$p_i$ for every $i \leq \binom{\kappa}{2}$. 
Note that the every token in $Z$ as well as every token in $Y$ can reach a $p_i$ in two sliding steps.
Hence, we can furthermore assume that a token from $Y$ never moves to a vertex in $X$. This assumption is justified as $C_d$ contains already $\binom{k}{2}$-many $p_i$, hence at most $k$ vertices of~$X$ are contained in $C_d$ (recall that $\binom{k}{2} \geq k$). 
Now, instead of moving a vertex from $Y$ to~$X$ and moving a vertex from $Z$ to a $p_i$ (each for a cost of 2), we can move a vertex from $Z$ to~$X$ and a vertex from~$Y$ to a $p_i$ for the same cost instead.

Let $z_P$ be the number of tokens moved from $Z$ to some $p_i$ and $z_X = \kappa - z_P$ be the number of tokens moved from $Z$ to $X$.
Then, $\kct - z_P$ tokens are moved from $Y$ to the remaining $p_i$.
Observe that $z_X$ tokens on $X$ can block at most $\binom{z_X}{2} = \binom{\kappa - z_P}{2}$ paths connecting $X$ and $Y$.
Hence, we have $\binom{\kappa - z_P}{2} \geq \binom{\kappa}{2} - z_P$, as all $s$-$t$-paths using vertices from $X$ and $Y$ are cut.
As we assume $\kappa \geq 4$, this inequality is only true for $z_P = 0$, that is, no vertex from $Z$ moves to a $p_i$.
Therefore, the ($\kappa$-many) tokens on $Z$ moved to $X$ and $\kct$ tokens on $Y$ moved to the $p_i$. 
Let $y_{uv} \in Y$ be such a vertex that has no token anymore. Now both $x_u$ and $x_v$ must have a token to cut all the paths between $s$ and~$t$ passing through $y_{uv}$. 
The $\kct$ vertices in $Y$ with no token satisfy the above property if and only if the corresponding edges in $G$ form a $\kappa$-clique in $G$, witnessing $(G, \kappa)$ is a yes-instance.
\end{proof}

Observe that a small modification of the construction yields the same hardness results for the \textsc{Edge Cut Discovery} problem, as we can keep the same graph $H$ and put the tokens on incident edges instead.
To be precise, we put the tokens in $Z$ on the edges connecting $Z$ and $s$ instead, and we put the tokens on $Y$ on the edges connecting $Y$ and $t$ instead. Finally, we set the budget to $k + \binom{k}{2}$. Hence, we can move the tokens on edges incident to $Z$ to the edges $\{s, x_u\}$ for $u \in C$ for total cost of $k$, and the tokens on edges incident to $Y$ to the edges $\{t, p_i\}$ for $i \leq \binom{k}{2}$ for a cost of~$\binom{k}{2}$. 

\begin{corollary}
    The \textsc{Edge Cut Discovery} problem  in the token sliding model is $\NP$-hard and \W{1}-hard with respect to parameter $\budget$ on 2-degenerate bipartite graphs.
\end{corollary}

Finally, combining \Cref{thm:rainbow-separator} and \Cref{thm:meta-theorem} (proved in Section~\ref{sec:rainbow}), we get the following result:

\begin{theorem}
    The \textsc{Vertex/Edge Cut Discovery} problem  in the token sliding model is fixed-parameter tractable with respect to parameter $k$. 
\end{theorem}

\section{Tractability via rainbow problems} \label{sec:rainbow}
In this section, we establish an algorithmic meta-theorem showing tractability of discovery problems  in the token sliding model (when parameterized by $k$) via the tractability of (weighted) rainbow variants of optimization problems. 
We utilize the color-coding technique introduced by Alon et al.~\cite{alon1995color} along with \FPT algorithms for the rainbow problems to design \FPT algorithms for the discovery problems parameterized by $k$. 

Let $\Pi$ be an optimization problem and $(G, k)$ be an instance of $\Pi$.
In the \textsc{Rainbow $\Pi$} problem we consider that $G$ is either a vertex-colored graph or an edge-colored graph. 
The parameter~$k$ refers to the solution size to be optimized. 
We consider problems that seek to optimize the selection of edges or vertices, but not both. 
For instance, the rainbow variants for separators and shortest paths are vertex selection problems whereas for spanning trees and matchings we have edge-selection problems. 
For an edge (or vertex) selection problem, the  \wrbpi~problem refers to the weighted variant where the input graph $G$ additionally has weights on the edges (or vertices) and we seek a solution of weight at most $b$ (amongst all solutions satisfying the cardinality constraint~$k$). 

\begin{theorem}\label{thm:meta-theorem}
For an optimization problem $\Pi$, if the \wrbpi~problem parameterized by $k$ admits an \FPT algorithm then the \pidis~problem in the token sliding model parameterized by $k$ admits an \FPT algorithm. 
\end{theorem}

Without loss of generality, assume that the problem $\Pi$ is an edge selection problem. 
Let $(G, S \subseteq E(G), b)$ be an instance of the \pidis~problem. 
Let $\Cmc$ be a palette of $k$ colors, and $\pi: \Cmc \to S$ be a bijection. 
We color the edges $E(G) \setminus S$ uniformly at random using $\Cmc$, yielding an edge coloring $\phi : E(G) \to \Cmc$.
Now we define a weight function $w: E(G) \to \mathbb{R}_+$ such that for each $e \in E(G)$ we have $w(e) = \mathrm{dist}(e, \pi(\phi(e)))$. 
Intuitively, the weight function denotes the cost of moving a token from the initial configuration to an edge with the same color. 
Observe that the weights of the edges in the initial configuration are zero and they are colored using piece-wise distinct colors. 
Now we have an edge-colored graph $(G, \phi)$. 
Consider the instance $(G, \phi, k, b)$ of the \wrbpi~problem. 

\begin{lemma}
\label{lem:metaThm:reduction}
    If $(G, \phi, k, b)$ is a yes-instance of the  \wrbpi~problem, then $(G,S,b)$ is a yes-instance of the \pidis~problem  in the token sliding model. 
\end{lemma}
\begin{proof}
Let $M \subseteq E(G)$ be a feasible solution for the given instance $(G, \phi, k, b)$ of the \wrbpi problem. 
Note that $M$ is a feasible solution for the instance $(G, k)$ of the problem~$\Pi$ (in the classic setting). 
Otherwise, $M$ cannot be a feasible solution in the discovery variant. 
Therefore~$M$ is a candidate solution for the \pidis~problem. 
Next we argue the feasibility of~$M$. 
We know that $\sum_{e \in M} w(e) \leq b$. 
Let $\mathrm{dist}(S, M)$ be the optimal cost of a transformation from $S$ to~$M$. 
Then, $\mathrm{dist}(S, M)$ can be bounded as:
\[\mathrm{dist}(S, M) \leq \sum_{e \in M} \mathrm{dist}(e, \pi(\phi(e))) = \sum_{e \in M} w(e) \leq b.\]
Thus, $(G, S, b)$ is a yes-instance. 
\end{proof}

Therefore, the color-coding step reduces the discovery problem to the rainbow variant. 

\begin{lemma}
\label{lem:metaThm:probBound}
    If $(G, S, b)$ is a yes-instance of the \pidis~problem, then the probability that $(G, \phi, k, b)$ is a yes-instance of the (edge) \wrbpi~problem is at least $e^{-k}$. 
\end{lemma}

\begin{proof}
Let $M \subseteq E(G)$ be a feasible solution for the instance $(G, S, b)$ of the \pidis~problem. 
Note that $M$ is a feasible solution for the instance $(G, k)$ of the problem $\Pi$ (in the classical setting). 
Otherwise, $M$ cannot be a feasible solution in the discovery variant. 
In order for $M$ to become a feasible solution for the rainbow variant, we want to ensure that $M$ has distinct colors from $\phi$. 
That is, $M$ should be colorful. 
Now we estimate the probability that $M$ is colorful. 
We know that the set~$S$ has fixed colors, and the remaining $m-k$ edges are colored uniformly  at random. 
We care only about the colors on $M$ but not the other edges. 
Therefore, 
$$\Pr[M\text{ is colorful}] \geq \frac{k^{m-2k}k!}{k^{m-k}} \geq \frac{1}{e^k}.$$
\end{proof}

\begin{lemma}\label{thm:lem-rand-weighted}
If the \wrbpi~problem admits an \FPT algorithm, then there exists an \FPT-time randomized algorithm that, given a \pidis~instance, either reports a failure or finds a feasible solution. Moreover, if the algorithm is given a yes-instance, then it returns a solution with probability at least $e^{-k}$. 
\end{lemma}
\begin{proof}
The proof follows immediately from  \Cref{lem:metaThm:reduction,lem:metaThm:probBound}.
\end{proof}
The success probability of our algorithm can be improved to a constant value by repeating the algorithm for $e^k$ times. 
Since, the running time of the algorithm is $f(k)$, for some computable function $f$, times a polynomial in $n$, the total running time is still bounded by a function of $k$ times a polynomial in $n$. 

We derandomize the algorithm of \cref{thm:lem-rand-weighted} using the technique introduced by Alon et al.~\cite{alon1995color}.  
Instead of a random coloring $\phi$, we construct an $(m-k,k)$-perfect hash family $\Fmc$ such that for every $X \subseteq E(G)\setminus S$ with $|X| = k$, there is a function $f \in F$ such that every element of $X$ is mapped to a different element in $\Cmc$ (color palette as a $k$-sized set). 
The family $\Fmc$ can be constructed deterministically in time $O(2^{O(k)}\log m)$~\cite{alon1995color}. 
It follows that the \pidis~problem can be solved in time $f_1(k)f_2(\cdot)n^{O(1)}$ for some computable function $f_1$ where $f_2(\cdot)$ is the time to solve the {\sc Weighted} \rbpi~problem. 
Therefore, the existence of an \FPT algorithm for the {\sc Weighted} \rbpi~problem implies that the \pidis~problem admits an \FPT algorithm with respect to the parameter $k$. 
\section{The jumping and addition/removal models}
\label{sec:other-token-models}

We now turn to the \emph{token jumping} and \emph{unrestricted token addition/removal} models. Recall that in these models we are no longer restricted to sliding tokens along edges.

\subsection{Token jumping} 

We first formally define the \emph{red-blue} variant of a  vertex (resp.\ edge) selection problem and show that it is always at least as hard as the solution discovery variant in the token jumping model.

Let $\Pi$ be an arbitrary vertex (resp.\ edge) selection problem. An instance of the \textsc{Red-Blue-$\Pi$} problem consists of a graph $G$ whose vertices (resp.\ edges) are either colored red or blue, as well as two non-negative integers $k$ and $b$. 
If $\Pi$ is the decision variant of a minimization/maximization problem, the goal is to decide whether there exists a solution $X \subseteq V(G)$  (resp.\ $X \subseteq E(G)$) of~$\Pi$ of size $k$ such that the number of blue elements in $X$ is at most $b$. 
We denote an instance of \textsc{Red-Blue-$\Pi$} by $(G, k, b)$.

\begin{lemma}
Let $\Pi$ be an arbitrary vertex (resp.\ edge) selection problem. There is a parameter-preserving reduction from \textsc{$\Pi$-Discovery} in the token jumping model to \textsc{Red-Blue-$\Pi$} that can be computed in linear time.
\end{lemma}
\begin{proof}
Let $(G, S, b)$ be an instance of \textsc{$\Pi$-Discovery} in the token jumping model. Let $G'$ be a copy of $G$ where we color every vertex (resp.\ edge) in $S$ red and every other vertex (resp.\ edge) blue. It is now easy to verify that $(G', |S|, b)$ is an equivalent \textsc{Red-Blue-$\Pi$} instance.
\end{proof}
\begin{corollary} Let $\Pi$ be an arbitrary vertex (resp.\ edge) selection problem. Then, the following results hold in the token jumping model:
\begin{enumerate}
    \item If \textsc{Red-Blue-$\Pi$} is in \PP, then \textsc{$\Pi$-Discovery} is in \PP.
    \item If \textsc{Red-Blue-$\Pi$} is in \FPT with respect to $k$ or $b$, then \textsc{$\Pi$-Discovery} is in \FPT with respect to $k$ or $b$.
    \item If \textsc{$\Pi$-Discovery} is \NP-hard, then \textsc{Red-Blue-$\Pi$} is \NP-hard.
    \item If \textsc{$\Pi$-Discovery} is \W{1}-hard with respect to $k$ or $b$, then \textsc{Red-Blue-$\Pi$} is \W{1}-hard with respect to $k$ or $b$.
\end{enumerate}    
\end{corollary}
We remark that the other direction does not trivially hold, \ie we can in general not consider an instance of \textsc{Red-Blue-$\Pi$} as an instance of \textsc{$\Pi$-Discovery}, as the number of red vertices/edges might exceed the bound~$k$ on the solution size. Nevertheless, we show that \textsc{Red-Blue Spanning Tree, Red-Blue Shortest Path} and \textsc{Red-Blue Matching} are in \PP, implying that their Discovery variants in the token jumping model are in \PP as well, whereas \textsc{Vertex/Edge Cut Discovery} is \NP-hard in the token jumping model, implying that \textsc{Red-Blue Vertex/Edge Cut Discovery} is \NP-hard as well.

\begin{proposition}
    The \textsc{Red-Blue Spanning Tree} problem can be solved in polynomial time. Hence,
    the \textsc{Spanning Tree Discovery} problem in the token jumping model can be solved in polynomial time.
    \label{prop:spanningTreeJumping}
\end{proposition}
\begin{proof}
Given an instance $(G, k, b)$ of the \textsc{Red-Blue Spanning Tree} problem, we clearly have a no-instance if $k \neq |V(G)|-1$.
Otherwise, we are looking for a 
spanning tree that uses at most $b$ blue edges. 
We find a spanning tree with minimum number of blue edges as follows.
We introduce a weight of $1$ for every blue edge and $0$ for every red edge. Computing a minimum spanning tree on this weighted instance yields a spanning tree in $G$ with the 
minimum number of blue edges. 
We accept the instance if and only if the minimum spanning tree has at most $b$ blue edges.
\end{proof}

\begin{proposition}
    The \textsc{Red-Blue Shortest Path} problem can be solved in polynomial time. Hence,
    the \textsc{Shortest Path Discovery} problem in the token jumping model can be solved in polynomial time.
\end{proposition}
\begin{proof}
    Let $(G, k, b)$ be an instance of the \textsc{Red-Blue Shortest Path} problem.
    If $\mathrm{dist}(s,t)\neq k + 1$ we clearly have a no-instance. Otherwise, we are looking for a shortest $s$-$t$-path that uses the maximum number of red vertices. Computing such a path can be done in polynomial time by computing a shortest path in an appropriately constructed auxiliary directed edge-weighted graph~$H$. The graph $H$ contains~$s$ and $t$ as well as all vertices appearing on any shortest path from~$s$ to~$t$. Vertices are grouped into levels based on their distance to $s$ and we only keep in $H$ the edges connecting vertices from different levels (we direct the edges in the obvious way based on the levels). Next, we replace every vertex $v$ (except $s$ and $t$) by two vertices $u_v$ and $w_v$ connected by an edge from $u_v$ to $w_v$, where $u_v$ inherits all incoming edges from $v$ and $w$ inherits all outgoing edges from $v$. We assign the edge $u_v w_v$ a weight of $1$ if $v$ is blue and $0$ otherwise (all other edges are also assigned a weight of $0$). 
    We accept the instance if and only if the shortest path in $H$ has weight at most~$b$. 
\end{proof}

\begin{lemma}\label{lem:MatchingJumping}
    The \textsc{Red-Blue Matching} problem can be solved in polynomial time. Hence,
    the \textsc{Matching Discovery} problem in the token jumping model can be solved in polynomial time.
\end{lemma}
\begin{proof}
    Let $(G, k, b)$ be an instance of the \textsc{Red-Blue Matching} problem.
    In order to solve this problem, we first check whether $G$ as a matching of size $k$. If this is not the case, we have a no-instance. For the rest of the proof, we assume that $G$ has a matching of size $k$. We construct a graph $H$ as follows. First, we copy $G$ and add weights to edges as follows. We introduce a weight of~$1$ for every red edge and a weight of $1-\epsilon$ for some $\epsilon < 1/|E(G)|$ for every blue edge. Additionally, we add $n-2k$ apex vertices, \ie we add $n-2k$ fresh vertices and connect them to the copies of~$V(G)$ by edges which we assign weight $1-\epsilon$. Now, we compute a maximum weight matching by Edmond's Blossom algorithm \cite{blossom}. Obviously, there is a perfect matching in the constructed graph if and only if there is a matching of size $k$ in $G$. Additionally, note that the weights are chosen such that the maximum weight matching will be a perfect matching, in case one exists. That is, there is a matching of size $k$ in $G$ if and only if maximum weight matching in $H$ is a perfect matching. Let~$r$ denote the number of matching red edges in a maximum weight matching in $H$. Then, the weight of the matching is $(n-k) (1 - \epsilon) + \epsilon r$. We conclude that the matching in $H$ has maximum weight if and only if the corresponding matching $G$ has size exactly $k$ and uses the largest possible number of red edges. We have a yes-instance if and only if this matching uses at least $k-\budget$ red edges.   
\end{proof}
\begin{remark*}
We have proven that the \textsc{Red-Blue Matching} problem can be solved in polynomial time. We remark that the problem is closely related to the \textsc{Exact Matching} problem, whose complexity has been open for more than 40 years~\cite{DBLP:conf/mfcs/MaaloulyS22}. In the \textsc{Exact Matching} problem we are given an unweighted graph with edges colored in red and blue, and aim to find out whether there is a perfect matching with exactly $b$ blue edges.
\end{remark*}

\begin{proposition}
    The \textsc{Vertex/Edge Cut Discovery} problem  in the token jumping model is \NP-complete. Hence,
    the \textsc{Red-Blue Vertex/Edge Cut} problem is \NP-complete.
    \label{prop:SeparatorJumping}
\end{proposition}
\begin{proof}
    The construction in the proof of \Cref{thm:separator:hardness} showing the hardness of the problem in the token sliding model in fact also shows hardness in the token jumping model. We only need to set the budget to $\budget = \kappa + \binom{\kappa}{2}$ (instead of $2\kappa + 2\binom{\kappa}{2}$).
\end{proof}

\begin{theorem}\label{thm:red-blue-separator}
    The \textsc{Weighted Red-Blue Vertex/Edge Cut} problem is fixed-parameter tractable when parameterized by~$k$.
\end{theorem}
\begin{proof}
    The proof parallels the proof of \Cref{thm:rainbow-separator}, where we use the optimization version of Courcelle's Theorem for a formula $\phi(X,Y)$ stating that $X$ is a separator, $Y\subseteq X$ is the set of all blue vertices of $X$. 
    The automaton checking the truth of the formula and minimizing $|X|$ and $|Y|$ is extended with two weight components, which leads to a quadratic blowup, hence, in our case we conclude a running time of $f(k)\cdot n^3$ for some computable function $f$. 
\end{proof}

\subsection{Unrestricted token addition/removal}
Obviously, every spanning tree of a connected graph $G$ has size $|V(G)| - 1$ by definition. Hence, the \textsc{Spanning Tree Discovery} problem in the unrestricted token addition/removal model can easily be reduced to the token jumping model by halving the budget $\budget$ (rounded down). To see this, note that any solution to the \textsc{Spanning Tree Discovery} problem in the unrestricted token addition/removal model adds and removes the same number of tokens. Furthermore, the order in which we add or remove tokens from the initial configuration $S$ does not matter. Thus, we assume that the first modification step to $S$ is a token removal, followed by a token addition, followed by removals/additions in alternating order. This can easily be simulated by jumps.

Similarly, the length of a shortest $s$-$t$-path in a graph $G$ is a fixed number for fixed $s$ and~$t$. Hence, the same argument as above boils the \textsc{Shortest Path Discovery} problem in the unrestricted token addition/removal model down to the \textsc{Shortest Path Discovery} problem in the token jumping model.
When it comes to \textsc{Matching Discovery} and \textsc{Vertex/Edge Cut Discovery}, the unrestricted token addition/removal model allows for two natural variations on the problems. On one hand, and in tune with the rest of the paper, if we impose solutions of size exactly $|S|$ then it is not hard to see that the addition/removal model becomes again equivalent to the jumping model (and therefore the problems are polynomial-time solvable). 

\begin{corollary}
The \textsc{Spanning Tree Discovery}, \textsc{Shortest Path Discovery}, \textsc{Matching Discovery}, and \textsc{Vertex/Edge Cut Discovery} problems in the unrestricted token addition/removal model can be solved in polynomial time.
\end{corollary}

Alternatively, in a more relaxed formulation, if we drop the constraints on solution size and allow solutions of any size then the problems become considerably different. In particular, the \textsc{Matching Discovery} problem becomes equivalent to asking for a smallest subset of $S$ whose removal (from $S$) results in a matching. Similarly (and from the viewpoint of a minimization problem), the \textsc{Vertex/Edge Cut Discovery} problem becomes equivalent to the problem of computing a smallest subset of vertices whose addition to $S$ results in a cut between $s$ and $t$. We show below that, even in the relaxed setting, both problems remain polynomial-time solvable. 

\begin{proposition}
The \textsc{(Relaxed) Matching Discovery} problem in the unrestricted token addition/removal model can be solved in polynomial time.
\end{proposition}

\begin{proof}
Let $(G, S, b)$ be an instance of the \textsc{Matching Discovery} problem. Observe that we may assume that every discovery sequence $C_0 \dots C_\ell$ in $G$ with $S = C_0$ yielding a matching  never adds a token. To see this observe that every subset of a matching is still a matching, hence we can simply remove all additions from the sequence and obtain a cheaper sequence that yields a matching.
We consider the graph $G[S]$, that is, the graph $G$ where only the edges in $S$ and their incident endpoints survive.
We compute in polynomial time an arbitrary maximum matching $M$ of $G[S]$ using Edmonds Blossom algorithm \cite{blossom}. Now we consider the subset $M' = E(G[S]) \setminus M$ which is a minimum set of edges that need to be removed from $G$ to arrive at a matching.
Hence, if $|M'| \leq b$ we conclude that we are dealing with a yes-instance, and otherwise we conclude that we are dealing with a no-instance.
\end{proof}

\begin{proposition}
The \textsc{(Relaxed) Vertex/Edge Cut Discovery} problem can be solved in polynomial time in the unrestricted token addition/removal model.
\end{proposition}

\begin{proof}
We present the result for \textsc{Edge Cut Discovery}, the result of \textsc{Vertex Cut Discovery} is analogous. 
Let $(G, s, t, S, b)$ be an instance of the \textsc{Edge Cut Discovery} problem. Similar to the matching case, we may assume that every configuration sequence $C_0 \dots C_\ell$ in $G$ with $S = C_0$ yielding an edge cut does never remove a token, as every superset of an edge cut is still an edge cut.
We consider the graph $G - S$, that is, the graph obtained from $G$ by removing all edges in~$S$.
We compute in polynomial time an arbitrary minimum $s$-$t$-cut $C$ of $G - S$ (e.g. by using the flow-algorithm of Edmonds and Karp~\cite{edmonds1972theoretical}). The size of $C$ is exactly the minimum number of token additions we need in order to separate $s$ from $t$. 
Hence, if $|C| \leq b$ we conclude that we are dealing with a yes-instance, and otherwise we conclude that we are dealing with a no-instance.
\end{proof}

\section*{Conclusion and future work}\label{sec:conclusion}

We contribute to the new framework of solution discovery via reconfiguration with complexity theoretic and algorithmic results for solution discovery variants of {\em polynomial-time solvable} problems. While we can employ the efficient machinery of \textsc{Weighted Matroid Intersection} for the \textsc{Spanning Tree Discovery} problem, all other problems under consideration are shown to be \NP-complete, namely, \textsc{Shortest Path Discovery}, \textsc{Matching Discovery}, and \textsc{Vertex/Edge Cut Discovery}. For the latter problems, we provide a full classification of tractability with respect to the parameters solution size $k$ and transformation budget $b$. 

We expect further research on this new model capturing the dynamics of real-world situations as well as constraints on the adaptability of solutions. It seems particularly interesting to investigate directed and/or weighted versions of the studied problems. Notice that all base problems that we studied remain polynomial-time solvable in the presence of edge weights and for  cost functions summing the total weight of the solution (spanning tree, path, matching, cut). Interestingly, our result on the {\sc Spanning Tree Discovery} problem can be generalized to the weighted setting with a clever adaptation of the reduction; we give details in Appendix~\ref{app:mst}. For other problems the hardness results clearly hold, but we leave the existence of \FPT algorithms for future work.

Another challenge is the design of efficient algorithms that can compute approximate solutions with respect to the solution size/value or with respect to the allowed transformation budget. 
We note, without proof, that the hardness reduction for the \sdis~problem (Theorem~\ref{thm:separator:hardness}) can be adjusted to give a $n^{1-\epsilon}$-inapproximability of the optimal transformation budget.

\bibliographystyle{plain}
\bibliography{ref}

\begin{thebibliography}{10}

\bibitem{alon1995color}
Noga Alon, Raphael Yuster, and Uri Zwick.
\newblock Color-coding.
\newblock {\em Journal of the ACM}, 42(4):844--856, 1995.

\bibitem{arnborg1991easy}
Stefan Arnborg, Jens Lagergren, and Detlef Seese.
\newblock Easy problems for tree-decomposable graphs.
\newblock {\em Journal of Algorithms}, 12(2):308--340, 1991.

\bibitem{rainbowcuts}
Xuqing Bai, Renying Chang, and Xueliang Li.
\newblock More on rainbow disconnection in graphs.
\newblock {\em Discussiones Mathematicae Graph Theory}, 42:1185--1204, 2020.

\bibitem{berge1957two}
Claude Berge.
\newblock Two theorems in graph theory.
\newblock {\em Proceedings of the National Academy of Sciences}, 43(9):842--844, 1957.

\bibitem{bonsma2013complexity}
Paul Bonsma.
\newblock The complexity of rerouting shortest paths.
\newblock {\em Theoretical computer science}, 510:1--12, 2013.

\bibitem{bousquet2022survey}
Nicolas Bousquet, Amer~E. Mouawad, Naomi Nishimura, and Sebastian Siebertz.
\newblock A survey on the parameterized complexity of the independent set and (connected) dominating set reconfiguration problems.
\newblock {\em arXiv preprint arXiv:2204.10526}, 2022.

\bibitem{BroersmaL97}
Hajo Broersma and Xueliang Li.
\newblock Spanning trees with many or few colors in edge-colored graphs.
\newblock {\em Discuss. Math. Graph Theory}, 17(2):259--269, 1997.

\bibitem{DBLP:journals/tcs/ChenLS11}
Lily Chen, Xueliang Li, and Yongtang Shi.
\newblock The complexity of determining the rainbow vertex-connection of a graph.
\newblock {\em Theor. Comput. Sci.}, 412(35):4531--4535, 2011.

\bibitem{courcelle1993monadic}
Bruno Courcelle and Mohamed Mosbah.
\newblock Monadic second-order evaluations on tree-decomposable graphs.
\newblock {\em Theoretical Computer Science}, 109(1-2):49--82, 1993.

\bibitem{cygan2015parameterized}
Marek Cygan, Fedor~V. Fomin, Lukasz Kowalik, Daniel Lokshtanov, D{\'{a}}niel Marx, Marcin Pilipczuk, Michal Pilipczuk, and Saket Saurabh.
\newblock {\em Parameterized Algorithms}.
\newblock Springer, 2015.

\bibitem{de2019complexity}
Marzio De~Biasi and Juho Lauri.
\newblock On the complexity of restoring corrupted colorings.
\newblock {\em Journal of Combinatorial Optimization}, 37(4):1150--1169, 2019.

\bibitem{downey1995fixed}
Rod~G. Downey and Michael~R. Fellows.
\newblock Fixed-parameter tractability and completeness {II}: On completeness for {W[1]}.
\newblock {\em Theoretical Computer Science}, 141(1-2):109--131, 1995.

\bibitem{DBLP:series/mcs/DowneyF99}
Rodney~G. Downey and Michael~R. Fellows.
\newblock {\em Parameterized Complexity}.
\newblock Monographs in Computer Science. Springer, 1999.

\bibitem{blossom}
Jack Edmonds.
\newblock Paths, trees, and flowers.
\newblock {\em Canad. J. Math.}, 17:449--467, 1965.

\bibitem{Edmonds1970}
Jack Edmonds.
\newblock Submodular functions, matroids, and certain polyhedra.
\newblock {\em Combinatorial Structures and Their Applications}, pages 69--87, 1970.

\bibitem{Edmonds1979}
Jack Edmonds.
\newblock Matroid intersection.
\newblock {\em Ann. Discret. Math.}, 4:39--49, 1979.

\bibitem{edmonds1972theoretical}
Jack Edmonds and Richard~M. Karp.
\newblock Theoretical improvements in algorithmic efficiency for network flow problems.
\newblock {\em Journal of the ACM}, 19(2):248--264, 1972.

\bibitem{sol-discovery}
Michael~R. Fellows, Mario Grobler, Nicole Megow, Amer~E. Mouawad, Vijayaragunathan Ramamoorthi, Frances~A. Rosamond, Daniel Schmand, and Sebastian Siebertz.
\newblock On solution discovery via reconfiguration.
\newblock {\em CoRR}, abs/2304.14295, 2023.

\bibitem{FlumGrohe2006}
J{\"{o}}rg Flum and Martin Grohe.
\newblock {\em Parameterized Complexity Theory}.
\newblock Texts in Theoretical Computer Science. An {EATCS} Series. Springer, 2006.

\bibitem{Frank81}
Andr{\'{a}}s Frank.
\newblock A weighted matroid intersection algorithm.
\newblock {\em J. Algorithms}, 2(4):328--336, 1981.

\bibitem{garey1976planar}
Michael~R Garey, David~S. Johnson, and R~Endre Tarjan.
\newblock The planar hamiltonian circuit problem is {NP}-complete.
\newblock {\em SIAM Journal on Computing}, 5(4):704--714, 1976.

\bibitem{garnero2018fixing}
Valentin Garnero, Konstanty Junosza-Szaniawski, Mathieu Liedloff, Pedro Montealegre, and Pawe{\l} Rzazewski.
\newblock Fixing improper colorings of graphs.
\newblock {\em Theoretical Computer Science}, 711:66--78, 2018.

\bibitem{gomes2023minimum}
Guilherme Gomes, Cl{\'e}ment Legrand-Duchesne, Reem Mahmoud, Amer~E Mouawad, Yoshio Okamoto, Vinicius F~dos Santos, and Tom~C van~der Zanden.
\newblock Minimum separator reconfiguration.
\newblock {\em arXiv preprint arXiv:2307.07782}, 2023.

\bibitem{gomes2020some}
Guilherme Gomes, S{\'e}rgio~H Nogueira, and Vinicius F~dos Santos.
\newblock Some results on vertex separator reconfiguration.
\newblock {\em arXiv preprint arXiv:2004.10873}, 2020.

\bibitem{GuptaRSZ19}
Sushmita Gupta, Sanjukta Roy, Saket Saurabh, and Meirav Zehavi.
\newblock Parameterized algorithms and kernels for rainbow matching.
\newblock {\em Algorithmica}, 81(4):1684--1698, 2019.

\bibitem{ito2011complexity}
Takehiro Ito, Erik~D Demaine, Nicholas~JA Harvey, Christos~H Papadimitriou, Martha Sideri, Ryuhei Uehara, and Yushi Uno.
\newblock On the complexity of reconfiguration problems.
\newblock {\em Theoretical Computer Science}, 412(12-14):1054--1065, 2011.

\bibitem{ito2022shortest}
Takehiro Ito, Naonori Kakimura, Naoyuki Kamiyama, Yusuke Kobayashi, and Yoshio Okamoto.
\newblock Shortest reconfiguration of perfect matchings via alternating cycles.
\newblock {\em SIAM Journal on Discrete Mathematics}, 36(2):1102--1123, 2022.

\bibitem{kaminski2011shortest}
Marcin Kami{\'n}ski, Paul Medvedev, and Martin Milani{\v{c}}.
\newblock Shortest paths between shortest paths.
\newblock {\em Theoretical Computer Science}, 412(39):5205--5210, 2011.

\bibitem{DBLP:journals/tcs/KaminskiMM12}
Marcin Kaminski, Paul Medvedev, and Martin Milanic.
\newblock Complexity of independent set reconfigurability problems.
\newblock {\em Theoretical Computer Science}, 439:9--15, 2012.

\bibitem{DBLP:journals/tcs/LeP14}
Van~Bang Le and Florian Pfender.
\newblock Complexity results for rainbow matchings.
\newblock {\em Theor. Comput. Sci.}, 524:27--33, 2014.

\bibitem{DBLP:conf/mfcs/MaaloulyS22}
Nicolas~El Maalouly and Raphael Steiner.
\newblock Exact matching in graphs of bounded independence number.
\newblock In Stefan Szeider, Robert Ganian, and Alexandra Silva, editors, {\em 47th International Symposium on Mathematical Foundations of Computer Science, {MFCS} 2022, August 22-26, 2022, Vienna, Austria}, volume 241 of {\em LIPIcs}, pages 46:1--46:14. Schloss Dagstuhl - Leibniz-Zentrum f{\"{u}}r Informatik, 2022.

\bibitem{marx2013finding}
D{\'a}aniel Marx, Barry O'sullivan, and Igor Razgon.
\newblock Finding small separators in linear time via treewidth reduction.
\newblock {\em ACM Transactions on Algorithms (TALG)}, 9(4):1--35, 2013.

\bibitem{nishimura2018introduction}
Naomi Nishimura.
\newblock Introduction to reconfiguration.
\newblock {\em Algorithms}, 11(4):52, 2018.

\bibitem{nomikos2007randomized}
Christos Nomikos, Aris Pagourtzis, and Stathis Zachos.
\newblock Randomized and approximation algorithms for blue-red matching.
\newblock In {\em Mathematical Foundations of Computer Science 2007: 32nd International Symposium, MFCS 2007 {\v{C}}esk{\`y} Krumlov, Czech Republic, August 26-31, 2007 Proceedings 32}, pages 715--725. Springer, 2007.

\bibitem{papadimitriou1982complexity}
Christos~H Papadimitriou and Mihalis Yannakakis.
\newblock The complexity of restricted spanning tree problems.
\newblock {\em Journal of the ACM (JACM)}, 29(2):285--309, 1982.

\bibitem{schrijver2003combinatorial}
Alexander Schrijver.
\newblock {\em Combinatorial Optimization $-$ Polyhedra and Efficiency}.
\newblock Springer, 2003.

\bibitem{uchizawa2013rainbow}
Kei Uchizawa, Takanori Aoki, Takehiro Ito, Akira Suzuki, and Xiao Zhou.
\newblock On the rainbow connectivity of graphs: complexity and fpt algorithms.
\newblock {\em Algorithmica}, 67:161--179, 2013.

\bibitem{van2013complexity}
Jan van~den Heuvel.
\newblock The complexity of change.
\newblock {\em Surveys in Combinatorics}, 409(2013):127--160, 2013.

\end{thebibliography}

\appendix
\section{Minimum spanning tree discovery}
\label{app:mst}

In this section, we consider the weighted version of the {\sc Spanning Tree Discovery} problem in the token sliding model and show that we can generalize the results from \Cref{sec:spanning} to the weighted setting. 

In the {\sc Minimum Spanning Tree Discovery} (MSTD) problem in the token sliding model, we are given an edge-weighted graph $(G, \omega)$, an edge subset $S \subseteq E(G)$ with $|S| = n-1$ as a starting configuration, and a non-negative integer $b$. 
The goal is to decide whether there exists a minimum-weight spanning tree of $G$ that can be discovered from $S$ using at most $b$ token slides. 
We show the following result.

\begin{theorem}\label{thm:wstd-in-p}
 The {\sc Minimum Spanning Tree Discovery} problem  can be solved in polynomial time in the token sliding model. 
\end{theorem}
\begin{proof}
We adapt the algorithm used in Section~\ref{sec:spanning} to solve the Unweighted {\sc (Unweighted) Spanning Tree Discovery} problem. There, we reduced the unweighted problem to the \rmst problem, which we showed to be solvable in polynomial time (Theorem~\ref{thm:rmst-in-p}). Now we argue that we can also give a reduction for the weighted problem by carefully setting the weights in the \rmst instance.

Given an edge-weighted graph $(G, \omega)$ and a starting configuration $S$, we construct an edge-colored weighted multigraph $(H,w,\phi)$ similar as in the proof of Theorem~\ref{thm:std-in-p}. 
We let $V(H) = V(G)$ and $E(H) = E(G) \times \Cmc$ and $\Cmc=[n-1]$. 
For each edge $e \in E(G)$ and color $i_{e'} \in \Cmc$, a selection of the edge $(e,i_{e'})$ in~$H$ denotes that in the \textsc{STD} problem a token slides from an edge $e' \in S$ with color~$i_{e'}$ to the edge~$e$. 
We define the weights $w$ as follows. 
For each edge $e = (e_1, i_{e_2}) \in E(H)$ for some $e_1 \in E(G)$ and $e_2 \in S$,
$$w(e) = n^2\omega(e_1) + d(e_1, e_2),$$
where $d(e_1,e_2)$ denotes the minimum sliding distance for a token sliding from edge $e_1$ to $e_2$. 

We claim that an optimal \rmst in the edge-colored weighted multigraph $(H,w,\phi)$ corresponds to a feasible token sliding solution for the solution discovery of a minimum spanning tree in the original graph $(G, \omega)$ from starting configuration $S$. 

Let $T \subseteq E(G)$ be a spanning tree for the graph $G$ and $c:T \to \Cmc$ be a bijection. 
Observe that~$(T,c)$ uniquely encodes a rainbow spanning tree in $H$. 
That is, $T' = \{(e,c(e)) \mid e \in T\}$ is a rainbow spanning tree in $H$. 
Further, each color $i \in \Cmc$ denotes a unique edge $e' \in S$. 
Let $s:\Cmc \to S$ be a bijection. 
The weight of the rainbow spanning tree $T'$ is
\begin{equation}
\label{eqn:mstd:equivalence}
w(T') = n^2\omega(T) + \sum_{e\in T}d(e, s(c(e)).    
\end{equation}

Similarly, given a rainbow spanning tree $T'$ of $H$, we can construct an equivalent spanning tree~$T$ of $G$ and a bijection $c:T \to \Cmc$ such that the weight functions $w$ and $\omega$ are as in Equation~\eqref{eqn:mstd:equivalence}. 
We know that the budget $b$ in the \textsc{MSTD} problem can be at most $(n-1)^2$. 
Therefore, the equivalent spanning tree of an optimal solution for the rainbow spanning tree problem on $H$ is an optimal solution for the \textsc{MSTD} problem in $G$. This concludes the proof of the theorem.
\end{proof}

\end{document}